\documentclass[prodmode,acmtocl]{acmsmall}
\newcommand{\qedhere}{\hfill\qed}

\usepackage{amsmath}
\usepackage{wrapfig}
\usepackage{graphicx}

\usepackage{tikz}
\usetikzlibrary{arrows,automata,shadows}
\usetikzlibrary{chains,shapes,decorations,decorations.pathreplacing}
\usetikzlibrary{decorations.pathmorphing}

\usepackage{hyperref}

\usepackage{xcolor}
\definecolor{vgreen}{rgb}{.1,.5,0}
\definecolor{vred}{rgb}{.7,0,0}
\definecolor{vblue}{rgb}{.1,.15,.62}

\newdef{counterexample}[theorem]{Counterexample}

\usepackage{prettyref}
\newcommand{\rref}[2][]{\prettyref{#2}}
\newrefformat{ch}{Sect.\,\ref{#1}}
\newrefformat{sec}{Sect.\,\ref{#1}}
\newrefformat{app}{Appendix\,\ref{#1}}
\newrefformat{def}{Def.\,\ref{#1}}
\newrefformat{thm}{Theorem\,\ref{#1}}
\newrefformat{prop}{Proposition\,\ref{#1}}
\newrefformat{lem}{Lemma\,\ref{#1}}
\newrefformat{cor}{Corollary\,\ref{#1}}
\newrefformat{ex}{Example\,\ref{#1}}
\newrefformat{tab}{Table\,\ref{#1}}
\newrefformat{fig}{Figure\,\ref{#1}}
\newrefformat{case}{case\,\ref{#1}}

\usepackage[bbsets,setrelation,Dfprime]{math}
\usepackage[modernsign,substopindex,shortmquant,mquantifiertype,mconnectiveformal,bracketinterpret,modifopindex,seqarrow,seqoptional,sidenotecalculus,abbrseqcontext,shortterms,nosigmaterms,novarterms]{logic}

\usepackage[pretest,nocommandblocks,keywordfont=textsfs]{progreg}
\usepackage[bracketinterpret]{dL}

\usepackage{mstochastic}

\usepackage{editingsupport}

\definecolor{vblue}{rgb}{.1,.15,.62}

\tikzstyle{state}=[circle,very thick,draw=blue!30,top color=white,bottom color=blue!20,circular drop shadow]
\tikzstyle{trans}=[draw=vblue,->,>=stealth',semithick]

\newcommand{\drawstate}[4][1cm]{%
  \begin{minipage}{#1}%
    \(%
      \begin{array}{@{}c@{}}%
        \ifthenelse{\equal{#2}{}}{}{\textsf{#2}\\}
        #3\ifthenelse{\equal{#4}{}}{}{\\}
        #4
      \end{array}%
    \)%
  \end{minipage}%
}

\def\linterpretationsname{\mathcal{S}}

\renewcommand{\entails}[1][]{\vDash}
\renewcommand{\lcloseop}{\ltrue}

\renewcommand{\mapply}[3][]{#2(#3)}

\newcommand{\der}[1]{\DD{(#1)}}%

\newcommand{\bebecomes}{\mathrel{::=}}
\newcommand{\alternative}{~|~}

\newcommand{\precond}{I}
\newcommand{\ivr}{\chi}%
\newcommand{\inv}{\phi}
\newcommand{\var}{\varphi}
\newcommand{\postcond}{g}

\newcommand*{\genDE}[1]{\theta}%

\newcommand{\sol}{x}%
\newcommand{\solf}{y}%
\newcommand{\stime}{\sol_0}%

\newcommand{\FOD}{FOD\xspace}
\newcommand{\FOQD}{\text{FOQD}\xspace}%

\newcommand{\onew}[2][]{n_{#2}}%

  \let\laforall\lforall%

\newcommand{\hastype}[2]{#1:#2}

  \renewcommand{\idomain}[2]{\iget[state]{#1}\ifthenelse{\equal{#2}{}}{}{(#2)}}
\newcommand{\jupd}{\mathcal{A}}%

\newcommand{\stdI}{\dLint[state=\nu]}
\newcommand{\I}{\stdI}
\newcommand{\It}{\dLint[state=\omega]}
\newcommand{\If}{\DALint[flow=\varphi]}

\usepackage{ifpdf}
\ifpdf
\fi

\acmVolume{V}
\acmNumber{N}
\acmArticle{A}
\acmYear{YYYY}
\acmMonth{0}

\title{Dynamic Logics of Dynamical Systems}
\author{ANDR\'E PLATZER
\affil{Carnegie Mellon University}}

\begin{abstract}
We study the logic of dynamical systems, that is, logics and proof principles for properties of dynamical systems.
\emph{Dynamical systems} are mathematical models describing how the state of a system evolves over time.
They are important for modeling and understanding many applications, including embedded systems and cyber-physical systems.
In \emph{discrete dynamical systems}, the state evolves in discrete steps, one step at a time, as described by a difference equation or discrete state transition relation.
In \emph{continuous dynamical systems}, the state evolves continuously along a function, typically described by a differential equation.
Hybrid dynamical systems or \emph{hybrid systems} combine both discrete and continuous dynamics.
\emph{Distributed hybrid systems} combine distributed systems with hybrid systems, i.e., they are multi-agent hybrid systems that interact through remote communication or physical interaction.
\emph{Stochastic hybrid systems} combine stochastic dynamics with hybrid systems.

We survey \emph{dynamic logics} for specifying and verifying properties for each of those classes of dynamical systems.
A dynamic logic is a first-order modal logic with a pair of parametrized modal operators for each dynamical system to express necessary or possible properties of their transition behavior.
Due to their full basis of first-order modal logic operators, dynamic logics can express a rich variety of system properties, including safety, controllability, reactivity, liveness, and quantified parametrized properties, even about relations between multiple dynamical systems.
In this survey, we focus on some of the representatives of the family of \emph{differential dynamic logics}, which share the ability to express properties of dynamical systems having continuous dynamics described by various forms of differential equations.

We explain the dynamical system models, dynamic logics of dynamical systems, their semantics, their axiomatizations, and proof calculi for proving logical formulas about these dynamical systems.
We study \emph{differential invariants}, i.e., induction principles for differential equations.
We survey theoretical results, including soundness and completeness and deductive power.
Differential dynamic logics have been implemented in automatic and interactive theorem provers and have been used successfully to verify safety-critical applications in automotive, aviation, railway, robotics, and analogue electrical circuits.
\end{abstract}




\acmformat{Platzer, A.  YYYY. Dynamic logics of dynamical systems.}

\begin{document}

\begin{bottomstuff}
This material is an extended version of \cite{DBLP:conf/lics/Platzer12a} and based upon work supported by the National Science Foundation under
NSF CAREER Award CNS-1054246, NSF EXPEDITION CNS-0926181, and under Grant Nos.
CNS-1035800 and CNS-0931985, by the ONR award N00014-10-1-0188, by the Army Research Office under Award No. W911NF-09-1-0273, and by the German Research Council (DFG) as part of the Transregional Collaborative Research Center ``Automatic Verification and Analysis of Complex Systems'' (SFB/TR 14 AVACS).

Author's addresses: A. Platzer, Computer Science Department,
Carnegie Mellon University.
\end{bottomstuff}

\maketitle

\section{Introduction}

\newsavebox{\exbox}%
\sbox{\exbox}{\rotatebox[origin=c]{180}{$\exists$}}%

\newsavebox{\Rval}%
\sbox{\Rval}{$\scriptstyle\mathbb{R}$}
\irlabel{qear|\usebox{\Rval}}
\irlabel{qeer|\usebox{\Rval}}
\newsavebox{\Ival}%
\sbox{\Ival}{$\mathcal{I}$}

Dynamical systems study the mathematics of change \cite{HirschSmaleDevaney,PerkoDEDS06}.
Dynamical systems are mathematical models for describing how the state of a system evolves over time in a state space.
They can describe, for example, the temporal evolution of the state of an embedded system or of a cyber-physical system, i.e., a system combining and integrating cyber (computation and/or communication) with physical effects.
Cars \cite{DBLP:conf/hybrid/DeshpandeGV96}, aircraft \cite{DBLP:journals/tac/TomlinPS:98}, robots \cite{DBLP:journals/fmsd/PlakuKV09}, and power plants \cite{DBLP:journals/circsm/FourlasKV04} are prototypical examples.
But dynamical systems are more general and can also describe and analyze chemical processes \cite{DBLP:journals/ejc/RileyKR10,DBLP:conf/hybrid/KerkezGDB10}, biological systems \cite{DBLP:conf/lics/Tiwari11}, medical models \cite{DBLP:conf/cav/GrosuBFGGSB11,DBLP:conf/emsoft/KimASLJZJ11}, and many other behavioral phenomena.
Since dynamical systems occur in so many different contexts, different variations of dynamical system models are relevant for applications, including discrete dynamical systems described by difference equations or discrete transitions relations \cite{Galor}, continuous dynamical systems described by differential equations \cite{HirschSmaleDevaney,PerkoDEDS06}, hybrid dynamical systems alias hybrid systems combining discrete and continuous dynamics \cite{DBLP:conf/rex/MalerMP91,DBLP:journals/tcs/AlurCHHHNOSY95,DBLP:conf/hybrid/Branicky95,DBLP:conf/lics/Henzinger96,DBLP:journals/tac/BranickyBM98,DavorenNerode_2000,DBLP:journals/ieee/AlurHLP00,DBLP:journals/jar/Platzer08,DBLP:journals/logcom/Platzer10,Platzer08,Platzer10,DBLP:conf/lics/Platzer12b}, distributed hybrid systems or multi-agent hybrid systems \cite{DBLP:conf/hybrid/DeshpandeGV96,DBLP:conf/hybrid/Rounds04,DBLP:conf/hybrid/KratzSPL06,DBLP:journals/taas/GilbertLMN09,DBLP:conf/csl/Platzer10,DBLP:journals/lmcs/Platzer12b}, and stochastic hybrid systems that take stochastic effects into account \cite{DBLP:journals/roystats/Davis84,DBLP:journals/jcopt/GhoshAM97,DBLP:conf/hybrid/HuLS00,BujorianuL06,Cassandras2006,DBLP:conf/hybrid/MeseguerS06,DBLP:journals/tsmc/KoutsoukosR08,DBLP:journals/jlp/FranzleTE10,DBLP:conf/cade/Platzer11}.

For many of the applications that can be understood as dynamical systems, we are interested in analyzing and predicting their behavior, e.g., because the applications are safety-critical or performance-critical.
For car control systems, for example, it is important to verify that the controllers choose only safe control choices that can never lead to collisions with other traffic participants at any later point in time \cite{DBLP:conf/hybrid/DeshpandeGV96,DBLP:conf/fm/LoosPN11}.

This illustrates a central point about the analysis of dynamical systems.
Whether a \emph{current} control choice is safe or unsafe in a dynamical system depends on whether the states that the dynamical system could reach after this control choice \emph{in the future} will be safe or unsafe.
Whether a dynamical system is safe or unsafe depends on whether it will \emph{always} choose safe control choices \emph{at all times}.
Whether we can find that out depends on whether we can find a \emph{proof} that the dynamical system is safe or whether we can find a proof that it is unsafe.

What we can accept as a proof or other form of evidence depends on how critical it is that the answer is right.
If the answer is that the dynamical system is unsafe, then a test scenario demonstrating one bad behavior is good evidence, because it can be used for debugging purposes.
If the dynamical system is suspected unsafe, then an expert's engineering judgment can be good evidence, because that would already prevent premature manufacturing and/or deployment of a potentially unsafe system design.
If the answer is that the dynamical system is safe, we prefer  stronger evidence than a series of successful test scenarios.
After all, most dynamical systems have large or even (uncountably) infinite state spaces, so that no finite set of tests alone could demonstrate that the system will be safe in the infinitely many other possible situations that could not be tested.
This issue is particularly daunting for the complex systems found in practical applications, e.g., because they follow complex control logic or many of their features interact or because their physical interactions are difficult etc.

For those reasons, we pursue the question of what constitutes a proof about a dynamical system and how we can systematically obtain proofs to show whether the system is safe or unsafe.
Safety, in this introductory discussion, should be broadly construed, because the approaches we study in this article work for much more complicated properties than classical safety properties as well, including liveness, controllability, reactivity, quantified parametrized properties and so on.

Our technical vehicle for answering these questions from a logically foundational perspective is our study of logics of dynamical systems.
We survey logics for studying properties of the behavior of dynamical systems and proof approaches for proving those properties deductively.
Dynamic logic \cite{DBLP:conf/focs/Pratt76} has been developed and used very successfully for conventional discrete programs, both for theoretical \cite{DBLP:conf/stoc/HarelMP77,Segerberg77,DBLP:conf/mfcs/Parikh78,DBLP:journals/jcss/FischerL79,Harel_1979,DBLP:journals/tcs/KozenP81,DBLP:journals/jcss/MeyerP81,DBLP:journals/jacm/Peleg87,DBLP:journals/iandc/Istrail84,Harel_et_al_2000,DBLP:conf/lics/Leivant06} and practical purposes \cite{DBLP:conf/cade/ReifSS97,Harel_et_al_2000,KeYBook2007}.
We consider extensions of dynamic logic to dynamical systems, including logic for hybrid systems \cite{DBLP:conf/tableaux/Platzer07,DBLP:journals/jar/Platzer08,DBLP:journals/logcom/Platzer10,Platzer08,Platzer10,DBLP:conf/lics/Platzer12b}, logic for distributed hybrid systems \cite{DBLP:conf/csl/Platzer10,DBLP:journals/lmcs/Platzer12b}, and logic for stochastic hybrid systems \cite{DBLP:conf/cade/Platzer11}.
We emphasize that the logic of dynamical systems approach we survey in this article lends itself to many interesting theoretical investigations as witnessed by a number of highly nontrivial theoretical results \cite{DBLP:conf/tableaux/Platzer07,DBLP:journals/jar/Platzer08,DBLP:journals/logcom/Platzer10,Platzer08,Platzer10,DBLP:conf/csl/Platzer10,DBLP:conf/cade/Platzer11,DBLP:journals/lmcs/Platzer12,DBLP:conf/lics/Platzer12b,DBLP:journals/lmcs/Platzer12b}, while, at the same time, enabling the practical verification of complex applications across different fields \cite{Platzer08,DBLP:conf/fm/PlatzerC09,DBLP:conf/icfem/PlatzerQ09,Platzer10,DBLP:conf/fm/LoosPN11,DBLP:conf/itsc/LoosP11,DBLP:conf/icfem/RenshawLP11,DBLP:conf/iccps/MitschLP12,DBLP:conf/acc/ArechigaLPK12} and inspiring algorithmic approaches based directly on these logics \cite{DBLP:conf/cav/PlatzerC08,Platzer08,DBLP:journals/fmsd/PlatzerC09,DBLP:conf/cade/PlatzerQ08,DBLP:conf/cade/PlatzerQR09,Platzer10,DBLP:conf/icfem/RenshawLP11}.

We remind the reader that this is not an isolated phenomenon.
Logics have been used very successfully in many different ways, including deduction and model checking, for verifying several other classes of systems, including finite-state systems \cite{ClarkeGrumberg_MC_1999,BaierKL08}, programs \cite{DBLP:conf/focs/Pratt76,Harel_et_al_2000,KeYBook2007,BradleyManna07,AptdeBoerOlderog10}, and real-time systems \cite{DBLP:conf/lics/Dutertre95,ZhouH04,OlderogD08,BaierKL08}.
Hybrid systems verification, for example, has generally received significant attention by the research community, including a number of verification tools \cite{DBLP:journals/sttt/HenzingerHW97,DBLP:conf/hybrid/MitchellT05,RatschanS07,DBLP:journals/sttt/Frehse08,DBLP:conf/cade/PlatzerQ08,DBLP:conf/icfem/RenshawLP11,DBLP:conf/cav/FrehseGDCRLRGDM11}; see \rref{ch:RelatedWork} for an overview.
Each verification approach has benefits and tradeoffs. It is promising  to combine ideas from approaches rooted in different traditions to leverage the specific advantages of each.
For instance, fixpoint loops, which are a driving force behind  model checking \cite{ClarkeGrumberg_MC_1999,BaierKL08}, have been used as a proof strategy to find deductive proofs in the proof calculus of differential dynamic logic \cite{DBLP:journals/jar/Platzer08}. Both can be used to compute invariants and differential invariants of the system \cite{DBLP:conf/cav/PlatzerC08,DBLP:journals/fmsd/PlatzerC09}.
The study of the logic of dynamical systems combines many areas of science, including mathematical logic, automated theorem proving, proof theory, model checking, and decision procedures, as well as differential algebra, computer algebra, algebraic geometry, analysis, stochastic calculus, and numerical approximation.

We see a number of advantages of the approach we focus on here, which make it attractive for research and applications, with the most important being soundness, completeness, compositionality, and extendability.
Because dynamical systems can capture very complex behavior, their analysis can become very challenging and it is surprisingly difficult to get the reasoning sound \cite{DBLP:journals/mst/Collins07,DBLP:conf/hybrid/PlatzerC07}.
In logic, \emph{soundness} is easier to achieve, because we just check a small number of elementary proof rules for soundness once and for all. Then everything that can be derived from those simple rules, no matter how complicated, is going to be correct.
Soundness (everything we prove is true) and completeness (we can prove everything that is true) are separated by design.
In logic, \emph{completeness} is a meaningful question to ask, not just in practice but also in theory, and has been answered in detail for logic of dynamical systems (\rref{sec:dL-complete} and \cite{DBLP:conf/lics/Platzer12b}).
More generally, theoretical questions and logically foundational questions, including relative completeness \cite{DBLP:journals/jar/Platzer08,DBLP:conf/lics/Platzer12b,DBLP:journals/lmcs/Platzer12b} and relative deductive power \cite{DBLP:journals/logcom/Platzer10,DBLP:journals/lmcs/Platzer12}, become meaningful in a logical setting.

The logics and proof systems we consider are \emph{compositional}. That is, the logics have a perfectly compositional, denotational semantics, in which the semantics of a model and the meaning of a formula are simple functions of the respective semantics of their parts.
Furthermore, the proof systems are compositional, i.e., they exploit this compositional semantics and systematically reduce a property of a complex systems to a number of properties about simpler systems by structural decomposition.
This makes it possible to understand complex dynamical systems in terms of their parts, which are often much easier than the full system.
In fact, completeness results prove that decomposition is always successful.
This result translates into practice, where systems that are designed according to good engineering practice adhering to modularity principles are easier to verify than those that are not.
Smart decompositions can have a tremendous impact on the practical verification complexity and improve scalability \cite{DBLP:conf/fm/LoosPN11}.

Another beneficial phenomenon in logics of dynamical systems is that they are easy to \emph{extend}.
Verification is based on a proof calculus, which is a collection of simple proof rules (and axioms).
In order to verify a feature in a different way, we can simply add new proof rules, which will improve the verification since the previous proof rules are kept as alternatives.
We will exercise this a number of times in this article, particularly when we are adding more and more proof rules to handle various sophisticated aspects of differential equations.
We start with simple rules using solutions of differential equations, then study differential invariants \cite{DBLP:journals/logcom/Platzer10}, an induction principle for differential equations, then differential cuts \cite{DBLP:journals/logcom/Platzer10,DBLP:journals/lmcs/Platzer12}, a logical cut principle for differential equations, and finally differential auxiliaries \cite{DBLP:journals/lmcs/Platzer12}.
Differential refinement and differential transformation rules are further extensions \cite{DBLP:journals/logcom/Platzer10,Platzer10}, but beyond the scope of this article.
Temporal logic extensions \cite{Platzer10,DBLP:conf/lfcs/Platzer07} and extensions to differential-algebraic hybrid systems \cite{DBLP:journals/logcom/Platzer10,Platzer10} are other illustrations of how the logic and proof calculus can be extended easily just by adding rules to cover more advanced temporal properties and systems with more complex dynamics.

In this article we focus on the logic of hybrid systems and we illustrate two more invasive extensions that change the logic of dynamical systems in fundamental ways by changing the characteristic of relevant dynamical aspects.
In \rref{ch:QdL}, we consider the logic of distributed hybrid systems \cite{DBLP:conf/csl/Platzer10,DBLP:journals/lmcs/Platzer12b}, which changes the state space in fundamental ways from fixed finite-dimensional state spaces to evolving and infinite-dimensional state spaces of arbitrarily many hybrid system agents interacting with each other through remote communication and physical interaction.
This extension is as radical as that from propositional logic to first-order logic, except that it happens in the dynamics, not just the propositions.
In \rref{ch:SdL}, we consider the logic of stochastic hybrid systems \cite{DBLP:conf/cade/Platzer11}, which changes the dynamics in fundamental ways to incorporate discrete and continuous stochastic effects changing the semantics from deterministic boolean truth to the randomness of stochastic processes.
While both extensions are radical, catapulting us into fundamentally different classes of dynamical systems, we will see that the changes in the proof calculi are surprisingly moderate additions of proof rules for new dynamical features and refinements, e.g., to adapt to a stochastic semantics.

Another helpful aspect of logic is that it produces proofs that can serve as readable evidence for the correctness of a system for certification purposes.
Concerns that are sometimes voiced in the context of classical discrete systems about theorem proving compared to model checking involve the degree of automation and the ability to find counterexamples.
They are less relevant for general dynamical systems.
Even the verification of very simple classes of hybrid systems is neither semidecidable nor co-semidecidable \cite{DBLP:journals/jcss/AsarinM98,DBLP:conf/lics/Henzinger96,DBLP:conf/concur/CassezL00}.
Consequently, quite unlike in finite-state systems and timed automata \cite{ClarkeGrumberg_MC_1999,BaierKL08}, exhaustive exploration of all states, even in bisimulation quotients, does not terminate in general, so that approximations and abstractions have to be used during the reachability analysis, and counterexamples are no longer reliable (see \cite{Clarke_2003b} for counterexample-guided abstraction refinement techniques).
Some nontrivial applications \cite{DBLP:conf/cav/PlatzerC08,DBLP:journals/fmsd/PlatzerC09,DBLP:conf/fm/PlatzerC09,DBLP:conf/icfem/PlatzerQ09,Platzer10} have been proved fully automatically with the approach we survey here.
Improving automation and scalability is, nevertheless, a permanently promising challenge in verification.
For complex systems, we find it advantageous that proving is amenable to human guidance, because the designer can specify the critical invariants of his system design, which helps finding proofs when current automation techniques fail.
In this article, we take a view that we call \emph{multi-dynamical systems}, i.e., the principle to understand complex systems as a combination of multiple elementary dynamical aspects.
This approach helps us tame the complexity of complex systems by understanding that their complexity just comes from combining lots of simple dynamical aspects with one another.
The overall system itself is still as complicated as the whole application.
But since differential dynamic logics and proofs are compositional, we can leverage the fact that the individual parts of a system are simpler than the whole, and we can prove correctness properties about the whole system by reduction to simpler proofs about their parts.
This approach demonstrates that the whole can be greater than the sum of all parts. The whole system is complicated, but we can still tame its complexity by an analysis of its parts, which are simpler.
Completeness results are the theoretical justification why this multi-dynamical systems principle works.

The results reported in this paper are based on previous research on logics of dynamical systems \cite{DBLP:conf/tableaux/Platzer07,DBLP:journals/jar/Platzer08,DBLP:journals/logcom/Platzer10,Platzer08,Platzer10,DBLP:conf/csl/Platzer10,DBLP:conf/hybrid/Platzer11,DBLP:conf/cade/Platzer11,DBLP:conf/lics/Platzer12b,DBLP:journals/lmcs/Platzer12,DBLP:journals/lmcs/Platzer12b}.
The results presented here are new in that we show significantly simplified Hilbert-type axiomatizations and, consequently, simplified semantics in comparison to the earlier presentations, which were more tuned for automation.
This setting enables us to identify connections between the approaches for the different classes of dynamical systems.
We provide an overview of the approach of logic of dynamical systems here, but it is, by no means, possible to handle all material comprehensively in this survey.
A more comprehensive source on logic of hybrid systems is a book \cite{Platzer10} and subsequent extensions \cite{DBLP:journals/lmcs/Platzer12,DBLP:conf/lics/Platzer12b}. Details about the logic of distributed hybrid systems \cite{DBLP:conf/csl/Platzer10,DBLP:conf/hybrid/Platzer11,DBLP:journals/lmcs/Platzer12b} and about logic of stochastic hybrid systems \cite{DBLP:conf/cade/Platzer11} can be found in previous work.
More information about algorithmic aspects can be found in related papers \cite{DBLP:conf/verify/Platzer07,DBLP:conf/cav/PlatzerC08,DBLP:journals/fmsd/PlatzerC09,DBLP:conf/cade/PlatzerQ08} and applications \cite{DBLP:conf/fm/PlatzerC09,DBLP:conf/icfem/PlatzerQ09,DBLP:conf/fm/LoosPN11,DBLP:conf/itsc/LoosP11,DBLP:conf/icfem/RenshawLP11,DBLP:conf/iccps/MitschLP12}.
Complementary extensions to differential temporal dynamic logic \cite{Platzer10,DBLP:conf/lfcs/Platzer07} and extensions to differential-algebraic hybrid systems with complex dynamics \cite{DBLP:journals/logcom/Platzer10,Platzer10} are very useful, but beyond the scope of this article.

In \rref{ch:dynamical-systems}, we briefly summarize the dynamical aspects of various classes of dynamical systems before we study their models, logics, and proofs in more detail in subsequent sections.
In \rref{ch:dL}, we study the logic of hybrid systems, which includes the logic of discrete dynamical systems and the logic of continuous dynamical systems as fragments.
In \rref{ch:QdL}, we study the logic of distributed hybrid systems, extending the results from \rref{ch:dL} to multi-agent scenarios.
We study the logic of stochastic hybrid systems in \rref{ch:SdL}.
We discuss related work in \rref{ch:RelatedWork} and give pointers to the literature.
Section~\ref{ch:Conclusions} concludes with a summary and an outlook for future research opportunities.

\section{Dynamical Systems} \label{ch:dynamical-systems}

In this section, we briefly recall the basic principles behind a number of classes of dynamical systems, for which we study models, logics, and proof approaches in subsequent sections.
We also illustrate our multi-dynamical systems view on these dynamical systems, which we detail in subsequent sections.

Formally, a \emph{dynamical system} is an action of a monoid $T$ (time) on a state space $\mathcal{X}$.
That is a dynamical system is described by a function $\varphi$, whose value $\varphi_t(x)\in\mathcal{X}$ at time $t\in T$ denotes the state that the system has at time $t$, provided that it started in the initial state $x\in\mathcal{X}$ at time 0.
It starts at \m{\varphi_0(x)=x} and the evolution can proceed in stages, i.e., \m{\varphi_{t+s}(x)=\varphi_s(\varphi_t(x))} for all $s,t\in T$ and $x\in\mathcal{X}$.
That is, if the dynamical system evolves for time $t$ and, from the state $\varphi_t(x)$ that it reached then, for time $s$, then it reaches the same state by simply evolving for time $t+s$ starting from $x$ right away.
For different choices of $T$ and $\mathcal{X}$, we get different classes of dynamical systems.
For computational analysis purposes, it is also crucial to choose a sufficiently computational description of the dynamical system~$\varphi$.

\subsection{Discrete Dynamical Systems} \label{sec:discrete-dynamical-system}

Discrete dynamical systems have an integer notion of time (e.g., $T=\naturals$ or $T=\integers$) so that the state evolves in discrete steps, one step at a time, as typically described by a difference equation or discrete state transition function.
The \emph{discrete dynamical system}
\begin{equation}
\varphi_{n+1}(x) = f(\varphi_n(x)) \quad(n\in\naturals)
\label{eq:discrete-dynamical}
\end{equation}
is fully described by its \emph{generator} \m{f:\mathcal{X}\to\mathcal{X}} or transition function, where $x\in\mathcal{X}$ is the initial state.
Equivalently, when defining \m{h(x):=f(x)-x} the discrete dynamical system \rref{eq:discrete-dynamical} can be described by the \emph{difference equation}
\[
\varphi_{n+1}(x) - \varphi_n(x) = h(\varphi_n(x)) \quad(n\in\naturals)
\]
Computation processes can be described by discrete dynamical systems, for example.
The system starts in an initial state $\varphi_0(x)=x$ at a time 0, performs a transition to a new state \m{\varphi_1(x)=f(x)} at a time 1, then another transition to a state \m{\varphi_2(x)=f(f(x))} at time 2, etc. until the computation terminates at a state $\varphi_n(x)$ at some time $n$. The scaling unit of these integer time steps is not relevant, but could be chosen, e.g., as the cycle time of a processor or discrete controller.
Program models and automata models have been used to describe discrete dynamical systems and have been used very successfully in verification \cite{ClarkeGrumberg_MC_1999,BaierKL08,AptdeBoerOlderog10}.
The behavior of systems with a discrete state transition relation \m{R\subseteq\mathcal{X}\times\mathcal{X}} is nondeterministic, but can still be captured as a discrete dynamical system using the powerset $2^{\mathcal{X}}$ as the state space instead of $\mathcal{X}$:
\[
\varphi_{n+1}(X) = \{f(x) \with x\in \varphi_n(X)\} \quad(n\in\naturals)
\]
when starting from a set $X\subseteq\mathcal{X}$ of initial states.

Discrete dynamical systems cannot, however, describe continuous processes, except as approximations at discrete points in time, e.g., with a uniform discretization grid $\frac{1}{n}$ at the discrete points in time $\frac{0}{n}, \frac{1}{n}, \frac{2}{n},\dots,\frac{n}{n}$.
Discrete-time approximations give limited information about the behavior in between the $\frac{i}{n}$, which causes fundamental differences \cite{DBLP:conf/hybrid/PlatzerC07} and similarities \cite{DBLP:conf/lics/Platzer12b}.

\subsection{Continuous Dynamical Systems}

Continuous dynamical systems have a real continuous notion of time (e.g. $T=\reals_{\geq0}$ or $T=\reals$) so that the state evolves continuously along a function of real time, typically described by a differential equation.
The state of the system $\varphi_t(x)$ then is a function of continuous time $t$.
The \emph{continuous dynamical system}
\begin{align*}
  \D[t]{\varphi_t(x)} &= f(\varphi_t(x)) \quad (t\in\reals)\\
  \varphi_0(x) &= x
\end{align*}
is fully described by its \emph{generator} \m{f:\mathcal{X}\to\mathcal{X}}, where $x\in\mathcal{X}$ is the initial state.
Depending on the duration of the solution of the above differential equation, the continuous system may only be defined on a subinterval of $\reals$.
The time-derivative $\D[t]{}$ is only well-defined under additional assumptions, e.g., that $\mathcal{X}$ is a Euclidean space $\reals^n$ or a differentiable manifold \cite{HirschSmaleDevaney,PerkoDEDS06}.

Many physical processes are continuous dynamical systems described by differential equations.
The movement of the longitudinal position of a car of velocity $v$ down a straight road from initial position $p_0$, for example, can be described by the differential equation \m{\D{p}(t)=v} with initial value \m{p(0)=p_0}.
The state of the dynamical system at time $t$ then is the solution \m{\varphi_t(p_0)=p_0+t v}, which is defined at all times $t\in\reals$.
We refer to the literature for more details and many more examples of continuous dynamical systems \cite{HirschSmaleDevaney,PerkoDEDS06}.
Continuous dynamical systems cannot represent discrete transitions easily; see, however, \rref{sec:dL-complete}.
Discrete transitions lead to discontinuities, which lead to interesting but very complicated generalized notions of solutions, including Carath\'eodory solutions \cite{Walter:ODE} or Filippov solutions \cite{AubinCellina84}.

\subsection{Hybrid Systems}

\emph{Hybrid dynamical systems} alias \emph{hybrid systems} \cite{DBLP:journals/tcs/AlurCHHHNOSY95,DBLP:conf/hybrid/Branicky95,DBLP:conf/lics/Henzinger96,DBLP:journals/tac/BranickyBM98,DavorenNerode_2000,DBLP:journals/ieee/AlurHLP00,DBLP:journals/jar/Platzer08,DBLP:journals/logcom/Platzer10,Platzer08,Platzer10,DBLP:conf/lics/Platzer12b} are dynamical systems that combine discrete dynamical systems and continuous dynamical systems.
Discrete and continuous dynamical systems are not just combined side by side to form hybrid systems, but they can interact in interesting ways.
Part of the system can be described by discrete dynamics (e.g., decisions of a discrete-time controller), other parts are described by continuous dynamics (e.g., movement of a physical process), and both kinds of dynamics interact freely in a hybrid system (e.g., when the discrete controller changes control variables of the continuous side by appropriate actuators, e.g., when changing acceleration, or when the continuous dynamics determines the values of sensor readings for the discrete decisions, e.g., the velocity).
Embedded systems and cyber-physical systems are often modeled as hybrid systems, because they involve both discrete control and physical effects.

\begin{wrapfigure}{r}{37mm}
\vspace*{-\baselineskip}
\includegraphics[width=39mm]{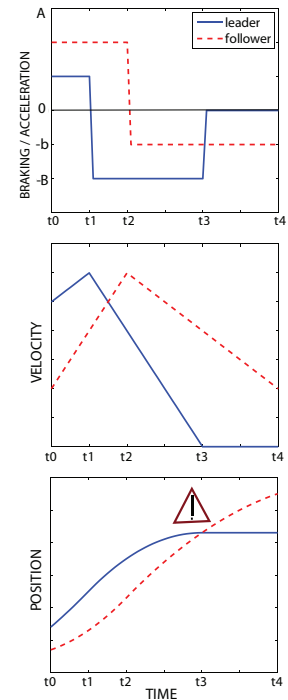}
\caption{Local car crash} 
\label{fig:local-lane}
\end{wrapfigure}
A typical example is a car that drives on a road according to a differential equation for the physical movement, but is subject to discrete control decisions where discrete controllers change the acceleration and braking of the wheels, e.g., when the adaptive cruise control or the electronic stability program takes effect.
Figure~\ref{fig:local-lane} shows how the acceleration of a car changes instantaneously by discrete control decisions (top), and how the velocity and position evolve continuously over time (middle and bottom).
The situation in \rref{fig:local-lane} illustrates a bad control choice, where the follower car brakes too late (at time $t_2$) and then crashes into the leader car at time $t_3$.
In particular, the follower car made a bad decision to keep on accelerating at some point before time $t_2$, when it should have activated the brakes instead, because, at time $t_2$, no control choice (within the physical acceleration limits $-b$ to $A$) could prevent the crash.
This is one illustration of the phenomenon that bad control choices in the past cause unsafety in the future and that we need to verify our control choices now by considering their possible dynamical effects in the future.

Notice that the state space $\mathcal{X}$ has no bearing on whether a system is a hybrid system or not. It is the notion of time and dynamics that determines hybrid systems.
For example, a system that has both discrete-valued state variables from a discrete set $\{1,2,3,4\}$ and continuous-valued state variables from a continuous set like $\reals$ is still a discrete dynamical system if all its variables only change in discrete steps (\rref{sec:discrete-dynamical-system}).

In hybrid systems, we follow our multi-dynamical systems philosophy and model each part of the system by the most appropriate dynamics, whether discrete or continuous, instead of having to model everything discrete, uniformly, for the whole system as in discrete dynamical systems or to model everything continuous, uniformly, as in continuous dynamical systems.
The overall system behavior can still be very complicated, if the system under investigation is complex, but at least each part of the system has an easier, more natural model.

For example, when using hybrid systems, there neither is a need to use unnatural discretizations for continuous phenomena, because full continuous dynamics is allowed in hybrid systems. Nor is there a need to represent the system dynamics with the interesting but  complicated discontinuous Carath\'eodory \cite{Walter:ODE} or Filippov solutions \cite{AubinCellina84} to understand jumps in continuous processes, because discrete jumps are allowed directly as separate elements in hybrid systems.
The overall system behavior can still be as complicated, and, in fact, a study of some behaviors in terms of Carath\'eodory and Filippov solutions can be insightful. But the individual parts of the hybrid system have a simpler behavior that can be understood and analyzed by easier means.
In our model for hybrid systems, the dynamical affects have separate atomic programs that can be combined in flexible ways by program combinators (\rref{ch:dL}).

We exploit the multi-dynamical systems philosophy in our analysis approach, because the logics we explain in the subsequent sections of this article have a fully compositional semantics and fully compositional proof principles.
Thus, since our proof approach works by reasoning by parts, all the individual reasoning steps get easier, because hybrid systems combine many but simpler dynamical aspects instead of requiring a single inscrutable effect.
Consequently, we can reason separately about the individual parts of the hybrid systems.

\subsection{Distributed Hybrid Systems}

Distributed hybrid systems \cite{DBLP:conf/hybrid/DeshpandeGV96,DBLP:conf/hybrid/Rounds04,DBLP:conf/hybrid/KratzSPL06,DBLP:conf/hybrid/MeseguerS06,DBLP:journals/taas/GilbertLMN09,DBLP:conf/csl/Platzer10,DBLP:journals/lmcs/Platzer12b,DBLP:conf/fmoods/JohnsonM12} are dynamical systems that combine distributed systems \cite{Lynch,DBLP:conf/concur/AttieL01,AptdeBoerOlderog10} with hybrid systems (and their discrete and continuous dynamics).
Again, they are not just combined side by side, but can interact.

Distributed systems are systems consisting of multiple computers that interact through a communication network.
They feature both (discrete) local computation and remote communication.
Distributed hybrid systems, instead, consist of multiple hybrid systems that interact through a communication network, but may also interact through physical interactions.
Distributed hybrid systems include multi-agent hybrid systems and hybrid systems where the number of agents involved in the system evolves over time.
A typical example is a distributed car control scenario (see \rref{fig:distributed-car-control-new}),
\begin{wrapfigure}{r}{0.5\textwidth}
  \includegraphics[width=0.5\textwidth]{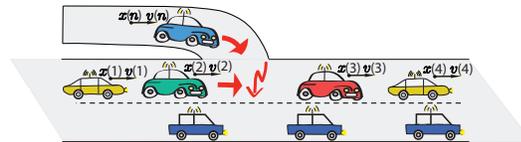}%
  \caption{Distributed car control}
  \label{fig:distributed-car-control-new}
\end{wrapfigure}
in which multiple cars drive on a road and use sensing and/or communication to inform each other of their respective positions and velocities and control intentions in order to coordinate their actions to prevent collisions.
Distributed hybrid systems become crucial, e.g., when we do not know how many agents are going to be involved exactly, or when there are more agents than hybrid systems analysis could handle.
Consequently, unlike in classical hybrid systems, the state space of distributed hybrid systems is usually an infinite-dimensional vector space $\mathcal{X}$.
Because of their importance in practical applications, many modeling approaches have been pursued for distributed hybrid systems \cite{DBLP:conf/hybrid/DeshpandeGV96,DBLP:conf/hybrid/Rounds04,DBLP:conf/hybrid/KratzSPL06,DBLP:conf/hybrid/MeseguerS06}, including SHIFT \cite{DBLP:conf/hybrid/DeshpandeGV96}, R-Charon \cite{DBLP:conf/hybrid/KratzSPL06}, and the process algebra $\chi$ \cite{DBLP:journals/jlp/BeekMRRS06}.

In distributed hybrid systems, we follow our multi-dynamical systems philosophy and model each part of the system by the most appropriate dynamical aspect, whether discrete or continuous or structural (e.g., changes in the communication topology or changes in the physical configuration) or dimensional (e.g., appearance or disappearance of cars on the street).
In our model for distributed hybrid systems, the dynamical affects have separate atomic programs that can be combined in flexible ways by program combinators (\rref{ch:QdL}).
We exploit the multi-dynamical systems philosophy in our analysis, logic, and proofs, so that we can reason separately about the individual parts of a distributed hybrid system.

\subsection{Stochastic Hybrid Systems}

Stochastic hybrid systems \cite{DBLP:journals/roystats/Davis84,DBLP:journals/jcopt/GhoshAM97,DBLP:conf/hybrid/HuLS00,BujorianuL06,Cassandras2006,DBLP:conf/hybrid/MeseguerS06,DBLP:journals/tsmc/KoutsoukosR08,DBLP:journals/jlp/FranzleTE10,DBLP:conf/cade/Platzer11} are dynamical systems that combine the dynamics of stochastic processes \cite{KaratzasShreve,Oksendal07,KloedenPlaten2010} with hybrid systems.
Again, they are not just combined side by side, but can interact.

\begin{wrapfigure}{r}{0.5\columnwidth}
  \vspace{-\baselineskip}
  \includegraphics[width=0.5\columnwidth]{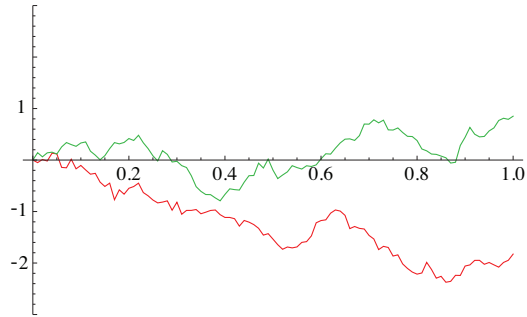}%
  \caption{Two samples from a switched continuous stochastic process}
  \label{fig:brownian-motion3}
\end{wrapfigure}
There is more than one way in which stochasticity has been added into hybrid systems models; see, e.g., \rref{fig:brownian-motion3}.
Stochasticity might be restricted to the discrete dynamics, as in piecewise deterministic Markov decision processes \cite{DBLP:journals/roystats/Davis84}, restricted to the continuous and switching behavior as in switching diffusion processes \cite{DBLP:journals/jcopt/GhoshAM97}, or allowed in many parts as in so-called General Stochastic Hybrid Systems; see \cite{BujorianuL06,Cassandras2006} for an overview.
Stochastic hybrid systems models have the desire in common to add stochastic information about uncertainties into the system dynamics.
Hybrid systems and distributed hybrid systems are limited to nondeterministic views and can only encode simple probabilistic effects in their hybrid dynamics.
For stochastic hybrid systems, the state space is more complicated, because it has to be rich enough to define stochastic process transitions. But the time domain is still such that some transitions are in continuous time, others are discrete steps in time.

In stochastic hybrid systems, we follow our multi-dynamical systems philosophy and model each part of the system by the most appropriate dynamical aspect, whether discrete or continuous, whether stochastic or not.
In particular, there is no need to represent the system dynamics with interesting but complicated concepts like semimartingales \cite{KaratzasShreve,Protter}.
The overall system behavior can still be as complicated, and a study of some behaviors in terms of semimartingales can be insightful. But the individual parts of the stochastic hybrid system have a simpler behavior that can be understood and analyzed by easier means.
In our model for stochastic hybrid systems, the dynamical effects have separate atomic programs that can be combined by program combinators (\rref{ch:SdL}).
We exploit the multi-dynamical systems philosophy in our analysis, logic, and proofs, so that we can reason separately about the individual parts of a stochastic hybrid system.

\section{Differential Dynamic Logic for Hybrid Systems} \label{ch:dL}

In this section, we study \emph{differential dynamic logic} \dL \cite{DBLP:conf/tableaux/Platzer07,DBLP:journals/jar/Platzer08,DBLP:conf/lics/Platzer12b}, the \emph{logic of hybrid systems}, i.e., systems with interacting discrete and continuous dynamics.

Hybrid systems \cite{DBLP:journals/tcs/AlurCHHHNOSY95,DBLP:conf/hybrid/Branicky95,DBLP:conf/lics/Henzinger96,DBLP:journals/tac/BranickyBM98,DavorenNerode_2000,DBLP:journals/ieee/AlurHLP00,DBLP:journals/jar/Platzer08,DBLP:journals/logcom/Platzer10,Platzer08,Platzer10,DBLP:conf/lics/Platzer12b} are a fusion of continuous dynamical systems and discrete dynamical systems.
They freely combine dynamical features from both worlds and play an important role, e.g., in modeling systems that use computers to control physical systems.
Hybrid systems feature (iterated) difference equations for discrete dynamics and differential equations for continuous dynamics. They, further, combine conditional switching, nondeterminism, and repetition.

As a specification and verification language for hybrid systems, we have introduced \dfn[logic!differential~dynamic]{differential dynamic logic} \dfn[\dL]{{\dL}} \cite{DBLP:conf/tableaux/Platzer07,DBLP:journals/jar/Platzer08,Platzer08,Platzer10,DBLP:conf/lics/Platzer12b}.
The logic \dL is based on first-order modal logic \cite{DBLP:journals/jsyml/Carnap46,HughesCresswell96} and dynamic logic \cite{DBLP:conf/focs/Pratt76,Harel_et_al_2000} and internalizes operational models of hybrid systems as first-class citizens, so that correctness statements about the transition behavior of hybrid systems can be expressed as logical formulas.
In addition to all operators of first-order real arithmetic, the logic \dL provides parametrized modal operators~$\dbox{\alpha}{}$ and~$\ddiamond{\alpha}{}$ that refer to the states reachable by hybrid system~$\alpha$ and can be placed in front of any formula.
The \dL formula~\m{\dbox{\alpha}{\phi}} expresses
that all states reachable by hybrid system~$\alpha$ satisfy formula~$\phi$. Likewise,~\m{\ddiamond{\alpha}{\phi}} expresses
that there is at least one state reachable by~$\alpha$ for
which~$\phi$ holds.
These modalities can be used to express necessary or possible properties of the transition behavior of~$\alpha$.

We first explain the system model of hybrid programs that \dL provides for modeling hybrid systems (\rref{sec:HP}).
Then we explain the logical formulas that \dL provides for specification and verification purposes (\rref{sec:dL-formula}).
For reference, we provide a short exposition of hybrid automata (\rref{sec:hybrid-automata}) and relate them to hybrid programs.
Then, we explain reasoning principles, axioms, and proof rules for verifying \dL formulas (\rref{sec:dL-calculus}).
We subsequently show soundness and relative completeness theorems (\rref{sec:dL-complete}) and investigate stronger proof rules for differential equations (\rref{sec:diffind}--\ref{sec:diffaux}).
Finally, we briefly discuss an implementation in the theorem prover \KeYmaera and applications (\rref{sec:KeYmaera}).

\subsection{Regular Hybrid Programs} \label{sec:HP}
Differential dynamic logic uses (regular) \emph{hybrid programs} (HP) \cite{DBLP:conf/tableaux/Platzer07,DBLP:journals/jar/Platzer08,Platzer10,DBLP:conf/lics/Platzer12b} as hybrid system models.
HPs are a program notation for hybrid systems and combine differential equations with conventional program constructs and discrete assignments.
HPs form a Kleene algebra with tests \cite{DBLP:journals/toplas/Kozen97}.
Atomic HPs are instantaneous discrete jump \emph{assignments} \m{\pupdate{\pumod{x}{\theta}}},
\emph{tests} $\ptest{\ivr}$ of a first-order formula\footnote{
The test $\ptest{\ivr}$ means ``if $\ivr$ then \textit{skip} else \textit{abort}''. Our results generalize to rich-test \dL, where \m{\ptest{\ivr}} is a HP for any \dL formula $\ivr$ (\rref{sec:dL-formula}).} $\ivr$ of real arithmetic,
and \emph{differential equation (systems)} \m{\hevolvein{\D{x}=\genDE{x}}{\ivr}} for a continuous evolution restricted to the domain of evolution $\ivr$, where $\D{x}$ denotes the time-derivative of $x$.
Compound HPs are generated from atomic HPs by nondeterministic choice ($\cup$), sequential composition ($;$), and Kleene's nondeterministic repetition ($\prepeat{}$).
We use polynomials with rational coefficients as terms here, but divisions can be allowed as well when guarding against singularities of divisions by zero; see \cite{DBLP:journals/jar/Platzer08,Platzer10} for details.
\begin{definition}[Hybrid program]
HPs are defined by the following grammar ($\alpha,\beta$ are HPs, $x$ a variable, $\theta$ a term possibly containing $x$, and $\ivr$ a formula of first-order logic of real arithmetic):
\[
  \alpha,\beta ~\bebecomes~
  \pupdate{\pumod{x}{\theta}}
  \alternative
  \ptest{\ivr}
  \alternative
  \hevolvein{\D{x}=\genDE{x}}{\ivr}
  \alternative
  \alpha\cup\beta
  \alternative
  \alpha;\beta
  \alternative
  \prepeat{\alpha}
\]
\end{definition}
The first three cases are called atomic HPs, the last three compound.
The \dfn{test} action~\m{\ptest{\ivr}} is used to define conditions. Its effect is that of a \textit{no-op} if the formula~$\ivr$ is true in the current state; otherwise, like \textit{abort}, it allows no transitions.
That is, if the test succeeds because formula~$\ivr$ holds in the current state, then the state does not change, and the system execution continues normally.
If the test fails because formula~$\ivr$ does not hold in the current state, then the system execution cannot continue, is cut off, and not considered any further.

Nondeterministic choice~\m{\pchoice{\alpha}{\beta}}, sequential composition~\m{\alpha;\beta}, and non\-de\-ter\-min\-is\-tic repetition~\m{\prepeat{\alpha}} of programs are as in regular expressions but generalized to a semantics in hybrid systems.
\dfn[nondeterministic!choice]{Nondeterministic choice} \m{\pchoice{\alpha}{\beta}} expresses behavioral alternatives between the runs of~$\alpha$ and~$\beta$.
That is, the HP~\m{\pchoice{\alpha}{\beta}} can choose nondeterministically to follow the runs of HP~$\alpha$, or, instead, to follow the runs of HP~$\beta$.
The \dfn[composition!sequential]{sequential composition}~\m{\alpha;\beta} models that the HP~$\beta$ starts running after HP~$\alpha$ has finished ($\beta$ never starts if~$\alpha$ does not terminate).
In~\m{\alpha;\beta}, the runs of~$\alpha$ take effect first, until~$\alpha$ terminates (if it does), and then~$\beta$ continues.
Observe that, like repetitions, continuous evolutions within~$\alpha$ can take more or less time, which causes uncountable nondeterminism.
This nondeterminism occurs in hybrid systems, because they can operate in so many different ways, which is as such reflected in HPs.
\dfn[nondeterministic!repetition]{Nondeterministic repetition}~\m{\prepeat{\alpha}} is used to express that the HP~$\alpha$ repeats any number of times, including zero times.
When following~\m{\prepeat{\alpha}}, the runs of HP~$\alpha$ can be repeated over and over again, any nondeterministic number of times (\m{{\geq}0}).

These operations can define all classical WHILE programming constructs and all hybrid systems \cite{Platzer10}.
We, e.g., write \m{\hevolve{\D{x}=\genDE{x}}} for the unrestricted differential equation \m{\hevolvein{\D{x}=\genDE{x}}{\ltrue}}.
We allow differential equation systems and use vectorial notation.
Vectorial assignments are definable from scalar assignments and $;$ using auxiliary variables.\footnote{\newcommand{\old}[1]{\grave{#1}}A vectorial assignment \m{\pupdate{\pumod{x_1}{\theta_1}\syssep\dots\syssep\pumod{x_n}{\theta_n}}} is definable by \m{\pupdate{\pumod{\old{x}_1}{x_1}};\dots;\pupdate{\pumod{\old{x}_n}{x_n}};\pupdate{\pumod{x_1}{\old{\theta}_1}};\dots;\pupdate{\pumod{x_n}{\old{\theta}_n}}} where $\old{\theta}_i$ is $\theta_i$ with $x_j$ replaced by $\old{x}_j$ for all $j$.}
Other program constructs can be defined easily \cite{Platzer10}.
For example,  nondeterministic assignments of any real value to $x$, if-then-else statements, and while loops can be defined as follows:
\begin{equation}
\begin{aligned}
  \prandom{x}
  &\equiv \pchoice{\pevolve{\D{x}=1}}{\pevolve{\D{x}=-1}}\\
  \text{if}~ (\ivr) ~\text{then}~ \alpha ~\text{else}~\beta~\text{fi}
  &\equiv \pchoice{(\ptest{\ivr};\alpha)}{(\ptest{\lnot\ivr};\beta)}\\
  \text{if}~ (\ivr) ~\text{then}~ \alpha 
  &\equiv \pchoice{(\ptest{\ivr};\alpha)}{\ptest{\lnot\ivr}}\\
  \pwhile{\ivr}{\alpha} &\equiv \prepeat{(\ptest{\ivr}; \alpha)}; \ptest{\lnot\ivr}
\end{aligned}
\label{eq:HP-defined}
\end{equation}

HPs have a compositional semantics.
We define their semantics by a reachability relation and refer to previous work for their trace semantics \cite{DBLP:conf/lfcs/Platzer07,Platzer10}.
A \emph{state} $\iget[state]{\I}$ is a mapping from variables to $\reals$.
The set of states is denoted $\linterpretations{\Sigma}{V}$.
We denote the value of term $\theta$ in $\iget[state]{\I}$ by \m{\ivaluation{\I}{\theta}}.
The state \m{\iget[state]{\imodif[state]{\I}{x}{d}}} agrees with~$\iget[state]{\I}$ except for the interpretation of variable~$x$, which is changed to~\m{d\in\reals}.
We write $\imodels{\I}{\chi}$ iff first-order formula $\chi$ is true in state $\iportray{\I}$ (defined in \rref{sec:dL-formula}).
\begin{definition}[Transition semantics of HPs] \label{def:HP-transition}
Each HP $\alpha$ is interpreted semantically as a binary reachability relation \m{\iaccess[\alpha]{\I}\subseteq\linterpretations{\Sigma}{V}\times\linterpretations{\Sigma}{V}} over states, defined inductively by
\begin{itemize}
\item \m{\iaccess[\pupdate{\pumod{x}{\theta}}]{\I} = \{(\iget[state]{\I},\iget[state]{\It}) \with \iget[state]{\It}=\iget[state]{\I}~\text{except that}~\ignore{\iget[state]{\It}(x)=}\ivaluation{\It}{x}=\ivaluation{\I}{\theta}\}}
\item \m{\iaccess[\ptest{\ivr}]{\I} = \{(\iget[state]{\I},\iget[state]{\I}) \with \imodels{\I}{\ivr}\}}
\item
\newcommand{\Ift}{\DALint[state=\varphi(t)]}%
\newcommand{\Ifz}{\DALint[state=\varphi(\zeta)]}%
  \m{\iaccess[\hevolvein{\D{x}=\genDE{x}}{\ivr}]{\I} = \{({\iget[flow]{\If}(0)},{\iget[flow]{\If}(r)}) ~\with~ 
        \imodels{\Ift}{\hevolve{\D{x}=\genDE{x}}}} 
        and \m{\imodels{\Ift}{\ivr}} for all \m{0\leq t\leq r} for a solution \m{\iget[flow]{\If}:[0,r]\to\linterpretations{\Sigma}{V}} of any duration \m{r\}};
        i.e., with \m{\iget[state]{\Ift}(\D{x}) \mdefeq \D[\zeta]{\iget[state]{\Ifz}(x)}(t)}, $\iget[flow]{\If}$ solves the differential equation and satisfies $\ivr$ at all times \cite{DBLP:journals/jar/Platzer08}
\item \m{\iaccess[\pchoice{\alpha}{\beta}]{\I} = \iaccess[\alpha]{\I} \cup \iaccess[\beta]{\I}}
\newcommand{\Iz}{\dLint[state=\mu]}
\item \m{\iaccess[\alpha;\beta]{\I} = \iaccess[\beta]{\I} \compose\iaccess[\alpha]{\I}}
\(= \{(\iget[state]{\I},\iget[state]{\It}) : (\iget[state]{\I},\iget[state]{\Iz}) \in \iaccess[\alpha]{\I},  (\iget[state]{\Iz},\iget[state]{\It}) \in \iaccess[\beta]{\I}\}\)
\item \m{\iaccess[\prepeat{\alpha}]{\I} = \displaystyle\cupfold_{n\in\naturals}\iaccess[{\prepeat[n]{\alpha}}]{\I}} 
with \m{\prepeat[n+1]{\alpha} \mequiv \prepeat[n]{\alpha};\alpha} and \m{\prepeat[0]{\alpha}\mequiv\,\ptest{\ltrue}}.
\end{itemize}
\end{definition}
We refer to our book \cite{Platzer10} for a comprehensive background and for an elaboration how the case \(r=0\) (in which the only condition is \m{\iget[flow]{\If}(0)\models\ivr}) is captured by the above definition.
Time itself is not special but implicit. If a clock variable $t$ is needed in a HP, it can be axiomatized by \m{\hevolve{\D{t}=1}}.

\begin{example}[Single car] \label{ex:HP}
As an example, consider a simple car control scenario.
We denote the position of a car by $x$, its velocity by $v$, and its acceleration by $a$.
From Newton's laws of mechanics, we obtain a simple kinematic model for the longitudinal motion of the car on a straight road, which  can be described by the differential equation \m{\D{x}=v\syssep\D{v}=a}.
That is, the time-derivative of position is velocity (\m{\D{x}=v}) and, simultaneously, the derivative of velocity is acceleration (\m{\D{v}=a}).
We restrict the car to never drive backwards by specifying the evolution domain constraint $v\geq0$ and obtain the continuous dynamical system \m{\hevolvein{\D{x}=v\syssep\D{v}=a}{v\geq0}}.
In addition, suppose the car controller can decide to accelerate (represented by \m{\pupdate{\pumod{a}{A}}}) or brake (\m{\pupdate{\pumod{a}{-b}}}), where $A\geq0$ is a symbolic parameter for the maximum acceleration and $b>0$ a symbolic parameter describing the brakes.
The HP \m{\pchoice{\pupdate{\pumod{a}{A}}}{\pupdate{\pumod{a}{-b}}}} describes a controller that can choose nondeterministically to accelerate or brake.
Accelerating will only sometimes be a safe control decision, so the discrete controller in the following HP requires a test $\ptest{\ivr}$ to be passed in the acceleration choice:
\begin{equation}
  \textit{car}_s \mequiv
  \prepeat{\big((\pchoice{(\ptest{\ivr};\pupdate{\pumod{a}{A}})}{\pupdate{\pumod{a}{-b}}});~
  \hevolvein{\D{x}=v\syssep\D{v}=a}{v\geq0}\big)}
  \label{eq:ex-HP}
\end{equation}
This HP, which we abbreviate by $\textit{car}_s$, first allows a nondeterministic choice of acceleration (if the test $\ivr$ succeeds) or braking, and then follows the differential equation for an arbitrary period of time (that does not cause $v$ to enter $v<0$).
The HP repeats nondeterministically as indicated by the $\prepeat{}$ repetition operator.
Note that the nondeterministic choice ($\cup$) in \rref{eq:ex-HP} can nondeterministically select to proceed with \m{\ptest{\ivr};\pupdate{\pumod{a}{A}}} or with \m{\pupdate{\pumod{a}{-b}}}.
Yet the first choice can only continue if, indeed, formula $\ivr$ is true about the current state (then both choices are possible).
Otherwise only the braking choice will run successfully.
With this principle, HPs elegantly separate the fundamental principles of (nondeterministic) choice from conditional execution (tests).

Which formula is suitable for $\ivr$ depends on the control objective or property we care about.
A simple guess for $\ivr$ like $v\leq20$ has the effect that the controller can only choose to accelerate at lower speeds.
This condition alone is insufficient for most control purposes.
We will refine $\ivr$ in \rref{ex:dL}.

HPs are a program notation for hybrid systems.
Hybrid automata \cite{DBLP:journals/tcs/AlurCHHHNOSY95,DBLP:conf/lics/Henzinger96} are an automaton notation for hybrid systems.
Hybrid automata correspond to finite automata with guards and reset relations annotated at edges and with differential equations and evolution domain constraints annotated at nodes (defined in detail in \rref{sec:hybrid-automata}).
The car system in \rref{eq:ex-HP} can be represented by the hybrid automaton in \rref{fig:hybridAutomaton}.
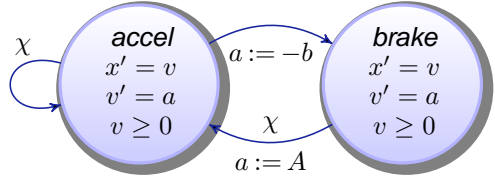
\begin{figure}[tbh]
  \centering
  \begin{tikzpicture}
    \renewcommand{\syssep}{\\}%
    \tikzstyle{every initial by arrow}=[trans]
    \tikzstyle{my loop}=[->,to path={ .. controls +(165:2) and +(195:2) .. (\tikztotarget) \tikztonodes}]
    \node[state,initial text=] (accel) at (0,0)
      {\drawstate[0.9cm]{$\textsl{accel}$}
        {\D{x}=v\syssep\D{v}=a}%
        {v\geq0}
      };
    \node[state] (brake) at (3.5,0)
      {\drawstate[0.9cm]{$\textsl{brake}$}
        {\D{x}=v\syssep\D{v}=a}%
        {v\geq0}
      };
    \draw[trans] (accel) to[bend left=30]
        node[above] {}
        node[below=3pt]{$\pumod{a}{-b}$}
        (brake);
    \draw[trans] (brake) to[bend left=30]
        node[above] {$\ivr$}
        node[below]{$\pumod{a}{A}$}
        (accel);
      \draw[trans] (accel) to [my loop] node[near end,below] {} node[near start,above] {$\ivr$}  (accel);
  \end{tikzpicture}
  \caption{Hybrid automaton for a simple car}
  \label{fig:hybridAutomaton}
\end{figure}
All hybrid automata can be represented as HPs \cite{Platzer10} just like finite automata can be implemented in classical WHILE programs (\rref{sec:hybrid-automata}).

An important phenomenon is that the evolution domain constraint in \rref{eq:ex-HP} and \rref{fig:hybridAutomaton} is too lax for many purposes.
It does not specify when the continuous evolution stops.
Many systems are unsafe if the continuous evolution evolves forever without giving the controller a chance to react.
To model \emph{event-triggered systems}, we would augment the evolution domain constraint with a formula that prevents the continuous evolution from missing important events. 
For example, we could add the evolution domain constraint $v\leq22$ into the differential equation in \rref{eq:ex-HP} to ensure that the continuous evolutions stop and the discrete controllers will react before the velocity increases beyond $22$:
\begin{equation*}
  \prepeat{\big((\pchoice{(\ptest{\ivr};\pupdate{\pumod{a}{A}})}{\pupdate{\pumod{a}{-b}}});~
  \hevolvein{\D{x}=v\syssep\D{v}=a}{v\geq0\land v\leq22}\big)}
  \label{eq:ex-HP-event}
\end{equation*}
In \emph{time-triggered systems}, we would, instead, replace the continuous evolution in \rref{eq:ex-HP} by
\m{
  \pupdate{\pumod{t}{0}};~\hevolvein{\D{x}=v\syssep\D{v}=a\syssep\D{t}=1}{v\geq0\land t\leq\varepsilon}
}
with a clock $t$ with slope \m{\D{t}=1} that is reset by a discrete assignment (\m{\pupdate{\pumod{t}{0}}}) before the continuous evolution and whose value is bounded ($t\leq\varepsilon$ in the evolution domain constraint) by a symbolic parameter for the maximum reaction time $\varepsilon>0$.
Then, the continuous evolution stops at the latest after $\varepsilon$ time units so that the discrete controllers have a chance to react to situation changes.
Without such a bound on the reaction time, systems are rarely safe.
The time-triggered version of \rref{eq:ex-HP} is the following HP, which we abbreviate by $\textit{car}_\varepsilon$:
\begin{equation}
\textit{car}_\varepsilon \mequiv
  \big((\pchoice{(\ptest{\ivr};\pupdate{\pumod{a}{A}})}{\pupdate{\pumod{a}{-b}}});~
  \pupdate{\pumod{t}{0}};~\hevolvein{\D{x}=v\syssep\D{v}=a\syssep\D{t}=1}{v\geq0\land t\leq\varepsilon}
  \prepeat{\big)}
  \label{eq:ex-HPeps}
\end{equation}
Time-triggered models are closer to the implementation, because event-triggered models require permanent sensing.
Event-triggered models are usually easier to verify but time-triggered models are easier to implement and reveal important timing effects.
\end{example}

Observe that, at this point, we could try to investigate the reachability question whether from a given state $\iget[state]{\I}$ we can reach a state $\iget[state]{\It}$ along car model $\textit{car}_s$ from \rref{eq:ex-HP}, i.e., \m{\iaccessible[\textit{car}_s]{\I}{\It}}, at which \m{\iget[state]{\It}(x)} is at a certain goal position.
We could also study the safety question whether for all states $\iget[state]{\It}$ with \m{\iaccessible[\textit{car}_s]{\I}{\It}} it is the case that \m{\iget[state]{\It}(v)<10} is true.
Instead of studying each of those questions with one ad-hoc notion for each question, we follow a more principled approach and define a logic in which those and many more general properties of hybrid systems can be expressed and verified.
We first discuss another instructive example, however.

\begin{example}[Bouncing ball] \label{ex:bouncing-ball}
Another intuitive example of a hybrid system is the bouncing ball \cite{EgerstedtJSL99}; see \rref{fig:bouncingball-simple}.
\begin{figure}[bth]
  \centering
  \begin{minipage}[T]{3cm}
    \begin{tabbing}
      $\bigl(\,$\= $\hevolvein{\D{h}=v\syssep\D{v}=-g}{h\geq 0}$;\\
      \>$\text{if }$\=$(h=0)$ $\text{then}$\\
      \> \> $v:= -cv$ \\
      \> $\text{fi} \bigr)^*$
    \end{tabbing}
  \end{minipage}%
  \qquad
  \begin{minipage}[T]{1.9cm}
  \includegraphics[width=1.9cm]{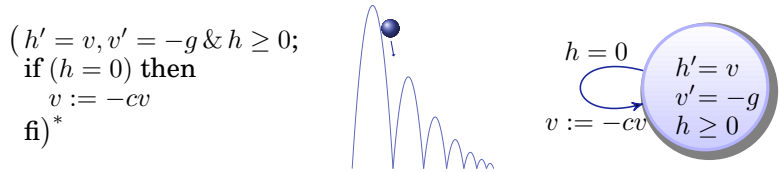}
  \end{minipage}%
  \qquad
  \begin{minipage}[T]{3.1cm}
    \begin{tikzpicture}
    \useasboundingbox (-2.1,-1) rectangle (1,1);
    \tikzstyle{every initial by arrow}=[trans]
    \tikzstyle{my loop}=[->,to path={ .. controls +(165:2) and +(195:2) .. (\tikztotarget) \tikztonodes}]
    \draw(0,0) node[state,text width=0.8cm,inner sep=0pt] (q) {\begin{tabbing}$\D{h}$\=$=v$\\$\D{v}$\>$=-g$\\$h$\>$\geq 0$\end{tabbing}};
      \draw[trans] (q) to [my loop] node[near end,below] {\m{v:=-cv}} node[near start,above] {$h=0$}  (q);
    \end{tikzpicture}
  \end{minipage}%
  \caption{Hybrid program, plot, and hybrid automaton of a bouncing ball}
  \label{fig:bouncingball-simple}
\end{figure}
The bouncing ball is let loose in the air and is falling towards the ground.
When it hits the ground, the ball bounces back up and climbs until gravity wins and it starts to fall again.
The bouncing ball follows the continuous dynamics of physical movement by gravity.
It can be understood naturally as a hybrid system, because its continuous movement switches from falling to climbing by reversing its velocity whenever the ball hits the ground and bounces back.
Let us denote the height of the ball by~$h$ and the current velocity of the ball by~$v$.
The bouncing ball is affected by gravity of force~\m{g>0}, so its height follows the differential equation \m{\D[2]{h}=-g}, i.e., the second time derivative of height equals the negative gravity force.
The ball bounces back from the ground (which is at height \m{h=0}) after an elastic deformation.
At every bounce, the ball loses energy according to a damping factor~\m{0\leq c < 1}.
Figure~\ref{fig:bouncingball-simple} depicts a HP, an illustration of the system dynamics, and a representation as a hybrid automaton.

The first line of the HP describes the continuous dynamics along the differential equation \m{\hevolve{\D{h}=v\syssep\D{v}=-g}} (which is equivalent to \m{\D[2]{h}=-g}) restricted to (written $\&$) the evolution domain \m{h\geq0} above the floor.
In particular, the bouncing ball never falls through the floor.
After the sequential composition ($;$), an if-then statement resets velocity~$v$ to \m{-cv} by assignment \m{\pupdate{\pumod{v}{-cv}}} if \m{h=0} holds at the current state.
This assignment will change the direction from falling (the velocity~$v$ was negative before) to climbing (the velocity \m{-cv} is nonnegative again) after dampening the velocity~$v$ by~$c$.
Recall \rref{eq:HP-defined} for how if-then is defined.
Finally, the sequence of continuous and discrete statements can be repeated arbitrarily often, as indicated by the regular-expression-style repetition operator ($\prepeat{}$) at the end.

The hybrid automaton on the right of \rref{fig:bouncingball-simple} represents the same system as a HP.
It has one node: falling along the differential equation system \m{\D{h}=v,\D{v}=-g} restricted to  evolution domain \m{h\geq0}, above the floor.
The hybrid automaton has one jump edge: on the ground (\m{h=0}), it can reset the velocity~$v$ to \m{-cv} and continue in the same node.

Note one strange phenomenon in the bouncing ball.
It seems like the bouncing ball will bounce over and over again, switching its direction in shorter and shorter periods of time as indicated in \rref{fig:bouncingball-simple} (unless \m{c=0}, which means that the ball will just lie flat right away).
Even worse, the ball will end up switching directions infinitely often in a short amount of time.
This controversial phenomenon is called \emph{Zeno behavior}.

  In reality, the ball bounces a couple of times and can then come to a standstill when its remaining kinetic energy is insufficient.
  To model this phenomenon without the need to have a precise physical model for all physical forces and frictions, we can allow for the damping factor~$c$ to change at each bounce by adding \m{\prandom{c};~ \ptest{(0\leq c<1)}} before \m{\pupdate{\pumod{v}{-cv}}}.
  HP \m{\prandom{c}} represents an uncountably infinite nondeterministic choice for~$c$ as a nondeterministic assignment.
  Recall \rref{eq:HP-defined} for its definition.
  The subsequent test \m{\ptest{(0\leq c<1)}} restricts the arbitrary choices for~$c$ to choices in the half-open interval \m{\interval{[0,1)}} and discards all other choices.
  Now the bouncing ball can stop.
  This particular model still allows a Zeno execution when each choice of $c$ is $c>0$, which can be removed by imposing additional restrictions on the permitted choices of $c$.
\end{example}

To avoid technicalities, we consider only polynomial differential equations here and refer to previous work \cite{DBLP:journals/logcom/Platzer10,Platzer10} for how to handle hybrid systems with more general differential equations, including differential equations with fractions, differential inequalities \cite{Walter:ODE}, differential-algebraic equations \cite{KunkelM06}, and differential-algebraic constraints with disturbances.
Those more general hybrid systems can be modeled by differential-algebraic programs, for which there is an extension of \dL called \emph{differential-algebraic dynamic logic} \DAL \cite{DBLP:journals/logcom/Platzer10,Platzer10}.
There also is an extension of \dL to temporal properties that gives hybrid programs a trace semantics.
This extension is called \emph{differential temporal dynamic logic} \dTL \cite{Platzer10,DBLP:conf/lfcs/Platzer07}. We refer to \cite{Platzer10} for details.

\subsection{\dLbf Formulas} \label{sec:dL-formula}

Differential dynamic logic \dL \cite{DBLP:conf/tableaux/Platzer07,DBLP:journals/jar/Platzer08,Platzer10,DBLP:conf/lics/Platzer12b} is a dynamic logic \cite{DBLP:conf/focs/Pratt76} for hybrid systems. It combines first-order real arithmetic \cite{tarski_decisionalgebra51} with first-order modal logic \cite{DBLP:journals/jsyml/Carnap46,HughesCresswell96} and dynamic logic \cite{DBLP:conf/focs/Pratt76} generalized to hybrid systems.
(Nonlinear) real arithmetic is necessary for describing concepts like safe regions of the state space and real-valued quantifiers are for quantifying over the possible values of system parameters.
The modal operators~$\dbox{\alpha}{}$ and~$\ddiamond{\alpha}{}$ refer to all (modal operator $\dbox{\alpha}{}$) or some (modal operator $\ddiamond{\alpha}{}$) state reachable by following HP~$\alpha$.

\begin{definition}[\dL formula]
The \emph{formulas of differential dynamic logic} ({\dL}) are defined by the grammar
(where $\phi,\psi$ are \dL formulas, $\theta_1,\theta_2$ terms, $x$ a variable, $\alpha$ a HP):
  \[
  \phi,\psi ~\bebecomes~
  \theta_1=\theta_2 \alternative
  \theta_1\geq\theta_2 \alternative
  \lnot \phi \alternative
  \phi \land \psi \alternative
  \lforall{x}{\phi} \alternative 
  \dbox{\alpha}{\phi}
  \]
\end{definition}
The operator $\ddiamond{\alpha}{}$ dual to $\dbox{\alpha}{}$ is defined by \m{\ddiamond{\alpha}{\phi} \mequiv \lnot\dbox{\alpha}{\lnot\phi}}.
Operators $>,\leq,<,\lor,\limply,\lbisubjunct,\exists{x}{}$ can be defined as usual, e.g., \m{\lexists{x}{\phi} \mequiv \lnot\lforall{x}{\lnot\phi}}.
We use the notational convention that quantifiers and modal operators bind strong, i.e., their scope only extends to the formula immediately after. Thus, \(\dbox{\alpha}{\phi}\land\psi \mequiv (\dbox{\alpha}{\phi})\land\psi\) and \(\lforall{x}{\phi}\land\psi \mequiv (\lforall{x}{\phi})\land\psi\).
In our notation, we also let $\lnot$ bind stronger than $\land$, which binds stronger than $\lor$, which binds stronger than $\limply,\lbisubjunct$.
Thus, \(\lnot A\land B\lor C\limply D\lor E\land F \mequiv (((\lnot A)\land B)\lor C) \limply (D\lor (E\land F))\).

\begin{definition}[\dL semantics]
\newcommand{\Id}{\imodif[state]{\I}{x}{d}}%
The \dfn[satisfaction]{satisfaction relation} \m{\imodels{\I}{\phi}} for \dL formula $\phi$ in state $\iget[state]{\I}$ is defined inductively and as usual in first-order modal logic (of real arithmetic):
  \begin{itemize}
  \item \(\imodels{\I}{(\theta_1=\theta_2)}\)
    iff \(\ivaluation{\I}{\theta_1} = \ivaluation{\I}{\theta_2}\).
  \item \(\imodels{\I}{(\theta_1\geq\theta_2)}\)
    iff \(\ivaluation{\I}{\theta_1} \geq \ivaluation{\I}{\theta_2}\).
  \item \(\imodels{\I}{\lnot\phi}\) iff
    it is not the case that \(\imodels{\I}{\phi}\).
  \item \(\imodels{\I}{\phi \land \psi}\) iff
    \(\imodels{\I}{\phi}\) and \(\imodels{\I}{\psi}\).
  \item \(\imodels{\I}{\lforall{x}{\phi}}\)
    iff
    \(\imodels{\Id}{\phi}\)
    for all \m{d\in\reals}.
    \index{$\lforall{}{}$}
  \item \(\imodels{\I}{\lexists{x}{\phi}}\)
    iff
    \(\imodels{\Id}{\phi}\)
    for some \m{d\in\reals}.
    \index{$\lexists{}{}$}
  \item
    \(\imodels{\I}{\dbox{\alpha}{\phi}}\)
      iff
      \m{\imodels{\It}{\phi}} for all $\iget[state]{\It}$ with \m{\iaccessible[\alpha]{\I}{\It}}.
      \index{$\dbox{\alpha}{}$}
     \item \(\imodels{\I}{\ddiamond{\alpha}{\phi}}\)
       iff
       \(\imodels{\It}{\phi}\)
       for some $\iget[state]{\It}$ with
       \m{\iaccessible[\alpha]{\I}{\It}}.
       \index{$\ddiamond{\alpha}{}$}
     \end{itemize}
If \m{\imodels{\I}{\phi}}, then we say that $\phi$ is true at $\iportray{\I}$.
A \dL formula $\phi$ is \emph{valid}, written \m{\entails\phi}, iff \m{\imodels{\I}{\phi}} for all states $\iportray{\I}$.
\end{definition}

A \dL formula of the form \m{A\limply\dbox{\alpha}{B}} corresponds to a Hoare triple \cite{Floyd67,DBLP:journals/cacm/Hoare69}, but for hybrid systems.
It is valid if, for all states: if \dL formula $A$ holds (in the initial state), then \dL formula $B$ holds for all states reachable by following HP $\alpha$.
That is, \m{A\limply\dbox{\alpha}{B}} is valid if $B$ holds in all states reachable by HP $\alpha$ from initial states satisfying $A$.
\begin{example}[Single car] \label{ex:dL}
First, consider a very simple \dL formula:
\[
v\geq0 \land A\geq0 \limply \dbox{\pupdate{\pumod{a}{A}};\, \pevolve{\D{x}=v\syssep\D{v}=a}}{v\geq0}
\]
This \dL formula expresses that, when, initially, the velocity $v$ and maximal acceleration $A$ are nonnegative, then all states reachable by the HP in the $\dbox{\cdot}{}$ modality have a nonnegative velocity ($v\geq0$).
The HP first performs a discrete assignment \m{\pupdate{\pumod{a}{A}}} setting the acceleration $a$ to maximal acceleration $A$, and then, after the sequential composition ($;$), follows the differential equation \m{\pevolve{\D{x}=v\syssep\D{v}=a}} where the derivative of the position $x$ is the velocity (\m{\D{x}=v}) and the derivative of the velocity is the chosen acceleration $a$ (\m{\D{v}=a}).
This \dL formula is valid, because the velocity will never become negative when accelerating. It could, however, become negative when choosing a negative acceleration $a<0$, which is what this simple \dL formula does not allow.

Next, consider the following \dL formula, where $\textit{car}_s$ denotes the HP from \rref{eq:ex-HP} in \rref{ex:HP} that always allows braking but acceleration only when \(\chi\mequiv v\leq20\) holds:
\[
v\geq0 \land A\geq0\land b>0 \limply \dbox{\textit{car}_s}{v\geq0}
\]
This \dL formula is trivially valid, simply because the postcondition $v\geq0$ is implied by both the precondition and by the evolution domain constraint of \rref{eq:ex-HP}.
Because it is implied by the precondition, $v\geq0$ holds initially.
It is also implied by the evolution domain constraint and the system has no runs that leave the evolution domain constraint.
Note that this \dL formula would not be valid, however, if we removed the evolution domain constraint, because the controller would then be allowed nondeterministically to choose a negative acceleration ($\pupdate{\pumod{a}{-b}}$) and stay in the continuous evolution arbitrarily long.

A more interesting valid \dL formula is the following, where $\textit{car}_\varepsilon$ denotes the time-triggered HP from \rref{eq:ex-HPeps} with the choice
\(\chi\mequiv 2b(x-m)\geq v^2+\big(A+b\big)\big(A\varepsilon^2+2\varepsilon v\big)\)
as acceleration constraint:
\begin{equation}
v^2\leq2b(m-x) \land A\geq0\land b>0 \limply \dbox{\textit{car}_\varepsilon}{x\leq m}
\label{eq:car-single-essentials}
\end{equation}
This \dL formula expresses that if, initially, the velocity is not too large (\m{v^2\leq2b(m-x)}) compared to the braking $b$ and remaining distance $m-x$ to a stoplight $m$, and if \(A\geq0\,\land\, b>0\), then all states reachable by following HP $\textit{car}_\varepsilon$ satisfy the postcondition \m{x\leq m}, i.e., the car never passes the stoplight at position $m$.

The \dL formula \rref{eq:car-single-essentials} expresses a safety property, because it says that $\textit{car}_\varepsilon$ always remains safely before the stoplight.
But that would be the case for car controllers that never move.
Yet, we can also express and prove liveness in \dL by showing that the car is able to (note the $\ddiamond{\cdot}{}$ modality) pass every point~$p$ by an appropriate choice of the stoplight $m$:
\begin{equation}
  \varepsilon>0\land-b>0\land A\geq0 \limply
  \lforall{p}{\lexists{m}{\ddiamond{\textit{car}_\varepsilon}{\,x\geq p}}}
  \label{eq:car-single-live}
\end{equation}
Statements of this type give alternations of quantifiers and of modalities.
See \cite[Sections 2.9 and 7.3]{Platzer10} and \cite{DBLP:conf/itsc/LoosP11,DBLP:conf/iccps/MitschLP12} for details about models in which $m$ changes dynamically as permission to move changes over time and for details on the proof of \dL formula \rref{eq:car-single-live}.
\end{example}

\begin{example}[Single car, multiple modalities]
The fact that \dL formula \rref{eq:car-single-essentials} is valid shows that its assumption about the initial state is sufficient for safety.
The logic \dL can be used to state and prove constraints that are both necessary and sufficient for dynamical properties \cite[Chapter 7]{Platzer10}.
First, we consider what the proper assumptions about the initial state should be for car control.
The HP $\textit{car}_\varepsilon$ in \rref{eq:car-single-essentials}, which originates from \rref{eq:ex-HPeps}, is a specific car control model deciding under which circumstance to choose which control action.
It would not make sense to require a controller to remain safe even in circumstances where no safe control choice is left, e.g., when not even immediate braking would be safe anymore.
In \dL, we can easily state that we want $\textit{car}_\varepsilon$ to always remain safe at least from those states where braking remains safe:
\begin{equation}
v\geq0 \land A\geq0\land b>0  \land \dbox{\pevolve{\D{x}=v\syssep\D{v}=-b}}{x\leq m} \limply \dbox{\textit{car}_\varepsilon}{x\leq m}
\label{eq:car-single-essentials-relative}
\end{equation}
This valid (and provable) \dL formula says that if, initially, \(v\geq0 \land A\geq0\land b>0\) holds and if $x\leq m$ would always hold if the decision were to brake immediately (i.e. \m{\dbox{\pevolve{\D{x}=v\syssep\D{v}=-b}}{x\leq m}}), then the more permissive control model $\textit{car}_\varepsilon$ also always remains safe (i.e. \m{\dbox{\textit{car}_\varepsilon}{x\leq m}}), because it may accelerate instead, but, due to the choice of constraint $\ivr$, will start braking in due time before $m$.
This principle of using modal formulas about simpler dynamical systems to describe states can be very useful for systematically designing controllers.
Formula \rref{eq:car-single-essentials-relative} is very intuitive: if braking would be safe, then $\textit{car}_\varepsilon$ will be safe, because it will notice in time when acceleration would not be safe any longer.

The same principle can be used to design how to choose the constraint $\ivr$ in $\textit{car}_\varepsilon$.
Constraint $\ivr$ is a design choice that determines under which circumstance the car controller is allowed to choose to accelerate.
Since, according to \rref{eq:ex-HPeps}, the car controller may possibly not have a chance to react again for up to $\varepsilon$ time units, the car controller can only choose to accelerate, if it would be safe to accelerate for $\varepsilon$ time units, and, after that, the car still has enough distance to brake (from its then faster velocity) before reaching the stoplight $m$.
This behavior can be expressed by the following \dL formula with two nested modalities, the first one for the acceleration for up to time $\varepsilon$, the second one for subsequent braking:
\[
  \dbox{\pupdate{\pumod{t}{0}};~\hevolvein{\D{x}=v\syssep\D{v}=a\syssep\D{t}=1}{v\geq0\land t\leq\varepsilon}}
  {\dbox{\pevolve{\D{x}=v\syssep\D{v}=-b}}{x\leq m}}
\]
In rich-test \dL, we can directly use \dL formulas with such modalities as tests inside HPs.
These are useful and instructive for designing systems, but harder to implement, because they refer to future states reached when following a dynamics.
This is perfect for model-predictive control, but tests on static quantifier-free first-order arithmetic formulas without modalities are easier to implement by simple arithmetic checks for the concrete values of the current state.

The following equivalence shows that the assumption about the initial state in \rref{eq:car-single-essentials} is necessary and sufficient and also explains how \dL formulas \rref{eq:car-single-essentials} and \rref{eq:car-single-essentials-relative} are related:
\[
v\geq0 \land b>0 \limply
  \big(v^2\leq2b(m-x) \lbisubjunct \dbox{\pevolve{\D{x}=v\syssep\D{v}=-b}}{x\leq m}\big)
\]
This valid \dL formula expresses that, if the initial velocity is nonnegative and the braking constant is $b>0$, then the car will always remain before the stoplight when braking if and only if \m{v^2\leq2b(m-x)} holds for the initial state.
Note that this \dL formula relates a dynamic statement (\m{\dbox{\pevolve{\D{x}=v\syssep\D{v}=-b}}{x\leq m}}) about the behavior at all future states of a dynamical system to a static statement about its present state.
Because the \dL formula is an equivalence ($\lbisubjunct$), it characterizes all states from which the car can be controlled to remain safely before the stoplight.
We refer to \cite[Chapter 7]{Platzer10} for more details on equivalent characterizations of dynamical constraints and how they can be used for systematic design.
\end{example}

\begin{example}[Bouncing ball]
Consider the bouncing ball (with or without variable damping coefficient $c$) from \rref{ex:bouncing-ball} and denote this HP by \textit{ball}.
The intuitive property that, if gravity $g$ is positive, the bouncing ball never bounces higher than its initial height~$H$ is expressed by the following \dL formula:
  \begin{equation*}
    h=H \land h\geq 0 \land g > 0 %
    \limply
    \dbox{\textit{ball}}{(0\leq h \leq H)}
  \end{equation*}
This \dL formula may be intuitive, but it is not valid, because the postcondition \m{0\leq h\leq H} will be violated if the initial velocity is positive (climbing).
Assuming $v\leq0$ holds initially,
  \begin{equation*}
    v\leq 0 \land h=H \land h\geq 0 \land g > 0 %
    \limply
    \dbox{\textit{ball}}{(0\leq h \leq H)}
  \end{equation*}
will, however, still not lead to a valid formula, because the ball would then start falling, but, if it is initially falling very fast (e.g., when dribbling a basket ball), then it will jump back higher than its initial height, despite the damping coefficient $c$.
We refer to \cite{Platzer10} for more details and for properties of bouncing balls that are valid (and provable).
\end{example}

The logic \dL also supports more complicated nested properties and quantifiers like \m{\lexists{p}{\dbox{\alpha}{\ddiamond{\beta}{\phi}}}} which says that there is a choice of parameter~$p$ (expressed by \m{\exists{p}{}}) such that for all behaviors of HP~$\alpha$  (expressed by \m{\dbox{\alpha}{}}) there is a reaction of HP~$\beta$ (i.e., \m{\ddiamond{\beta}{}}) that ensures that~$\phi$ holds in the resulting state.
Likewise,~\m{\lexists{p}{(\dbox{\alpha}{\phi}\land\dbox{\beta}{\psi})}} says that there is a choice of parameter~$p$ that makes both~\m{\dbox{\alpha}{\phi}} and~\m{\dbox{\beta}{\psi}} true, simultaneously, i.e., that makes the conjunction \m{\dbox{\alpha}{\phi}\land\dbox{\beta}{\psi}} true, saying that formula $\phi$ holds for all states reachable by runs of HP $\alpha$ and, independently, $\psi$ holds after all runs of HP $\beta$.
This results in a very flexible logic for specifying and verifying even sophisticated properties of hybrid systems, including the ability to refer to multiple hybrid systems at once in a single formula.
This flexibility is useful for computing invariants and differential invariants \cite{DBLP:conf/cav/PlatzerC08,DBLP:journals/fmsd/PlatzerC09,Platzer10}.

\subsection{Hybrid Automata} \label{sec:hybrid-automata}
{%
\newcommand{\flowc}[1]{\textit{flow}_{#1}}%
\newcommand{\jump}[1]{\textit{jump}_{#1}}%
\newcommand{\invariant}[1]{\textit{dom}_{#1}}%
\newcommand{\initc}[1]{\textit{init}_{#1}}%
\newcommand{\guard}[1]{\textit{guard}_{#1}}%
\newcommand{\reset}[1]{\textit{reset}_{#1}}%
\newcommand{\poststate}[1]{{#1}^+}%
\newcommand{\trans}[1][]{\stackrel{#1}{\curvearrowright}}%
\newcommand{\reach}{\closureTransitive{\trans}}%

In this subsection, we discuss hybrid automata and show their close relation to hybrid programs.
Besides hybrid programs, hybrid automata \cite{DBLP:journals/tcs/AlurCHHHNOSY95,DBLP:conf/lics/Henzinger96} are another popular notation for hybrid systems and there is a close connection between both models, which we have seen in Examples~\ref{ex:HP} and~\ref{ex:bouncing-ball}.
There are numerous slightly different notions of hybrid automata or automata-based models for hybrid systems~\cite{Tavernini87,DBLP:conf/hybrid/AlurCHH92,DBLP:conf/hybrid/NicollinOSY92,DBLP:journals/tcs/AlurCHHHNOSY95,DBLP:conf/hybrid/Branicky95,DBLP:conf/lics/Henzinger96,DBLP:journals/tse/AlurHH96,DBLP:journals/tac/BranickyBM98,DBLP:conf/hybrid/LafferrierePY99,DavorenNerode_2000,DBLP:conf/cav/PiazzaAMPWM05,DammHO06}.
We review a hybrid automata model close to Henzinger's \cite{DBLP:conf/lics/Henzinger96}, yet with polynomial differential equations even if the theoretical model allows others and even if verification tools for hybrid automata focus on subclasses of hybrid automata, e.g., constant \cite{DBLP:journals/sttt/HenzingerHW97,DBLP:journals/sttt/Frehse08} or linear dynamics.

Hybrid automata are graph models with two kinds of transitions: discrete jumps in the state space caused by mode switches (edges in the graph), and continuous evolution along continuous flows within a mode (vertices in the graph).
Recall the automata in \rref{fig:hybridAutomaton} from \rref{ex:HP} and in \rref{fig:bouncingball-simple} from \rref{ex:bouncing-ball} as typical examples.
\begin{definition}[Hybrid automaton] \label{def:hybridautomaton}
  A \dfn[automaton!hybrid]{hybrid automaton}~$A$ consists of
  \index{hybrid!automaton|see{automaton, hybrid}}%
  \begin{itemize}
  \item a finite set \(X=\{x_1,\dots,x_n\}\) of real-valued state variables, where $n\in\naturals$ is the dimension of $A$;
  \item a finite directed multigraph (i.e., graph that may have multiple edges between the same pair of vertices) with vertices~$Q$ (as \emph{modes}\index{mode}) and edges~$E$ (as \emph{control switches}\index{control!switch});
  \item flow conditions~\m{\flowc{q}} in mode $q\in Q$, i.e., differential equations \m{\hevolve{\D{x_1}=\theta_1\syssep\dots\syssep\D{x_n}=\theta_n}} that determine the relationship of the continuous state variables~$x_i$ and their time derivative~$\D{x_i}$ during continuous evolution in mode~\m{q\in Q};\index{mode}
  \item evolution domain constraints~\m{\invariant{q}}, which are first-order real arithmetic formulas over $X$ that have to be true of the continuous state while in mode~\m{q\in Q};
  \item initial conditions~\m{\initc{q}}, which are first-order real arithmetic formulas over $X$ that are true of the continuous state if the system starts in mode~\m{q\in Q};
  \item guard\index{guard} conditions~\m{\guard{e}}, i.e., which are first-order real arithmetic formulas over $X$ that determine whether the automaton can follow edge $e\in E$ depending on whether \m{\guard{e}} is true of the current state value;
  \item resets\index{reset}~\m{\reset{e}} along edge $e\in E$, which are lists of equalities \m{\poststate{x}_1=\theta_1,\dots,\poststate{x}_n=\theta_n} where $\theta_i$ is a term over $X$ that determines $\poststate{x}_i$, which denotes the new value of the continuous state variable~$x_i$ after following edge~$e\in E$;
  \end{itemize}
\end{definition}
It is crucial to work with a computational representation \cite{DBLP:conf/lics/Henzinger96}, e.g., as first-order real arithmetic formulas, instead of just arbitrary initial set of states \m{\initc{q}\subseteq\reals^n} and an arbitrary relation \m{\flowc{q}\subseteq\reals^n\times\reals^n} of variables and their derivatives to describe the dynamics.
Otherwise computational analysis becomes impossible.
If $\initc{q}$ is an undecidable set, it may already be undecidable whether 0 is an initial state ($0\in\initc{q}$).

\begin{definition}[Transition semantics of hybrid automata]
\newcommand{\Ix}{\dLint[state=x]}
\newcommand{\Ifz}{\dLint[state=\varphi(\zeta)]}
  The \emph{transition system} of a hybrid automaton~$A$ is a transition\index{transition} relation~$\trans$ defined as follows
  \begin{itemize}
  \item \m{S \eqdef \{(q,x) \in Q\times\reals^n \with x\models\invariant{q}\}} is the state space;
  \item \m{S_0 \eqdef \{(q,x) \in S \with x\models\initc{q}\}} is the set of initial states;
  \item ${\trans}\, \subseteq S\times S$ is the transition relation defined as the union
    \m{
    \cupfold_{e\in E}\trans[e] ~\cup~ \cupfold_{q\in Q} \trans[q]
    }
    where
    \begin{enumerate}
    \item \m{\relatedi{\trans[e]}{(q,x)}{(\tilde{q},\tilde{x})}} iff
      $e\in E$ is an edge from~\m{q\in Q} to~\m{\tilde{q}\in Q} in the hybrid automaton $A$ and \(x\models\guard{e}\) and, further, \m{\tilde{x}_i=\ivaluation{\Ix}{\theta_i}} for $i=1,\dots,n$. (\emph{discrete transition}).\index{discrete!transition}
    \item \m{\relatedi{\trans[q]}{(q,x)}{(q,\tilde{x})}} iff
      \m{q\in Q} and
      there is a function~\m{\varphi:\interval{[0,r]}\to\reals^n} that has a time derivative \m{\D{\varphi}:\interval{(0,r)}\to\reals^n}
      such that
      \(\varphi(0)=x,\varphi(r)=\tilde{x}\)
      and such that
      \m{\D{\varphi_i}(\zeta) = \ivaluation{\Ifz}{\theta_i}}
      at each \m{\zeta\in\interval{(0,r)}} and for $i=1,\dots,n$, where $\varphi_i$ is the projection of $\varphi$ to the $i$-th component.
      Further,
      \m{\varphi(\zeta)\models\invariant{q}} has to hold
      for each \m{\zeta \in \interval{[0,r]}}. (\emph{continuous transition}).\index{continuous!transition}
    \end{enumerate}
  \end{itemize}
  State~\m{\sigma\in S} is \emph{reachable} from state~\m{\sigma_0\in S_0}, denoted by~\m{\relatedi{\reach}{\sigma_0}{\sigma}}, iff,
  for some~\m{n\in\naturals}, there is a sequence of states
  \m{\sigma_1,\sigma_2,\dots,\sigma_n=\sigma\in S}
  such that
  \m{\relatedi{\trans}{\sigma_{i-1}}{\sigma_i}} for \m{1\leq i\leq n}.
\end{definition}

Just like for classical discrete systems, where every finite automaton can be implemented as a WHILE program, every hybrid automaton can be represented as a hybrid program with a similar construction \cite{Platzer10}.
Na\"I've compilation introduces additional coding variables, however, which may make verification unnecessarily tedious compared to a direct natural representation as a hybrid program.
The following HP has been compiled from the hybrid automaton in \rref{fig:hybridAutomaton}:
  \[
  \begin{array}{l}
    \hupdate{\humod{q}{\textsl{accel}}};  \qquad\textit{/* initial mode is node \textsl{accel} */}
    \\
    \big(
        \phantom{\hchoice{}{}\!}
        (\htest{q=\textsl{accel}};~~
        \hevolvein{\D{x}=v\syssep\D{v}=a}{v\geq0})
        \\\hchoice{}{}~
        (\htest{q=\textsl{accel}};~~\internal{\htest{inv(\textsl{accel})};}
        \hupdate{\humod{a}{-b}};~~
        \hupdate{\humod{q}{\textsl{brake}}};~~\htest{v\geq0})
        \\\hchoice{}{}~
        (\htest{q=\textsl{accel}\land \ivr};~~\internal{\htest{inv(\textsl{accel})};}
        \hupdate{\humod{q}{\textsl{accel}}};~~\htest{v\geq0})
        \\\hchoice{}{}~
        (\htest{q=\textsl{brake}};~~
        \hevolvein{\D{x}=v\syssep\D{v}=a}{v\geq0})
        \\\hchoice{}{}~
        (\htest{q=\textsl{brake}\land \ivr};~~\internal{\htest{inv(\textsl{accel})};}
        \hupdate{\humod{a}{A}};~~
        \hupdate{\humod{q}{\textsl{accel}}};~~\htest{v\geq0})
      \big)
      ^{*}
  \end{array}
  \]
Note the difference of this HP compared to the natural HP in \rref{ex:HP}.
Line~1 represents that, in the beginning, the current node~$q$ of the system is the initial node \textsl{accel}.
The HP represents each discrete and continuous transition of the automaton as a sequence of statements with a nondeterministic choice ($\cup$) between these transitions.
Line~2 represents a continuous transition of the automaton. It tests if the current node~$q$ is \textsl{accel}, and then (i.e., if the test was successful) follows the differential equation system \m{\hevolve{\D{x}=v\syssep\D{v}=a}} restricted to the evolution domain~\m{v\geq0}.
Line~3 characterizes a discrete transition of the automaton.
It tests whether the automaton is in node \textsl{accel}, resets \m{\pupdate{\pumod{a}{-b}}} and then switches~$q$ to node \textsl{brake}.
By the semantics of hybrid automata, an automaton in node \textsl{accel} is only allowed to make a transition to node \textsl{brake} if the evolution domain restriction of \textsl{brake} is true when entering the node, which is expressed by the additional test \m{\ptest{v\geq0}} at the end of line~3.
Observe that this test of the evolution domain restriction generally needs to be checked as the last operation after the guard and reset, because a reset like \m{\pumod{v}{v-1}} could affect the outcome of the evolution domain region test.
In order to obtain a fully compositional model, HPs make all these implicit side conditions explicit.
Line~4 represents the discrete transition for the self-loop at \textsl{accel} of the automaton.
It tests the guard \m{\ivr} when in node \textsl{accel}, and, if successful, switches~$q$ back to node \textsl{accel}, and checks the evolution domain constraint \m{v\geq0} of \textsl{accel}.
Line~5 represents the continuous transition when staying in node \textsl{brake} and following the differential equation system \m{\hevolve{\D{x}=v\syssep\D{v}=a}} restricted to the evolution domain~\m{v\geq0}.
Line~6 represents the discrete transition from node \textsl{brake} of the automaton to node \textsl{accel}, again testing the guard in the beginning and testing the evolution domain constraint of \textsl{accel} at the end.

Lines~2--6 cannot run unless their tests succeed.
In particular, at any state, the nondeterministic choice ($\cup$) among lines~2--6 reduces de facto to a nondeterministic choice between either lines~2--4 or between lines~5--6.
At any state,~$q$ can have value either \textsl{accel} or \textsl{brake} (assuming these are different constants), not both.
Consequently, when \m{q=\textsl{brake}}, a nondeterministic choice of lines~2--4 would immediately fail the tests in the beginning and not run any further.
The only remaining choices that have a chance to succeed are lines~5--6 then.
In fact, only the single successful choice of line~5 would remain if the second conjunct \m{\ivr} of the test in line~6 does not hold for the current state.
Note that, still, all four choices in lines~2--6 are available, but at least two of these nondeterministic choices will always be unsuccessful.
Note that executions of line~3,4, or 6 would fail if the respective test at the end of those lines fails.
Since $v\geq0$ is in the evolution domain constraints of all nodes, however, the system gets stuck if $v<0$, which can only happen initially in this system.
Finally, the repetition operator ($\prepeat{}$) at the end of the HP expresses that the transitions of a hybrid automaton, as represented by lines~2--6, can repeat arbitrarily often, possibly taking different nondeterministic choices between lines~2--6 at every repetition.
}%

We could have defined differential dynamic logic for hybrid automata instead of for hybrid programs, because hybrid automata can be compiled to hybrid programs.
The primary reason why we chose a hybrid program representation instead of an automata representation for our logic is because our verification works by structural decomposition and hybrid programs have a perfectly compositional semantics, which enables us to use perfectly compositional proof rules (\rref{sec:dL-calculus}).

\subsection{Axiomatization} \label{sec:dL-calculus}

We do not only use \dL for specification purposes but also for verification of hybrid systems.
That is, we use \dL formulas to specify what properties of hybrid systems we are interested in, and then use \dL proof rules to verify them.
The axioms and proof rules of \dL are syntactic, which means that we can use them to verify properties of hybrid systems without having to recourse to their mathematical semantics.
In \rref{sec:dL-complete}, we show that the semantics and proof rules of \dL match completely, so we are not losing anything by taking on a syntactic perspective on verification.
Syntactic proof rules are crucial, because they can be implemented and used computationally in a computer (\rref{sec:KeYmaera}).

Our axiomatization of \dL is shown in \rref{fig:dL}.
To highlight the logical essentials, we use our axiomatization from our recent result \cite{DBLP:conf/lics/Platzer12b} that is simplified compared to our earlier work \cite{DBLP:journals/jar/Platzer08}, which was tuned for automation.
The axiomatization we use here is closer to that of Pratt's dynamic logic for conventional discrete programs \cite{DBLP:conf/focs/Pratt76,DBLP:conf/stoc/HarelMP77}.
We use the first-order Hilbert calculus (modus ponens \irref{MP} and $\forall$-generalization rule \irref{gena}) as a basis and allow all instances of valid formulas of first-order real arithmetic as axioms.
The first-order theory of real-closed fields is decidable \cite{tarski_decisionalgebra51} by quantifier elimination.
\begin{figure}[tbh]
  \renewcommand*{\irrulename}[1]{\text{#1}}%
  \newdimen\linferenceRulehskipamount%
  \linferenceRulehskipamount=1mm%
  \newdimen\lcalculuscollectionvskipamount%
  \lcalculuscollectionvskipamount=0.1em%
  \begin{calculuscollections}{\columnwidth}
    \begin{calculus}
      \cinferenceRule[assignb|$\dibox{:=}$]{assignment / substitution axiom}
      {\linferenceRule[equiv]
        {\mapply[x]{\phi}{\theta}}
        {\dbox{\pupdate{\umod{x}{\theta}}}{\mapply[x]{\phi}{x}}}
      }
      {}%
      \cinferenceRule[testb|$\dibox{?}$]{test}
      {\linferenceRule[equiv]
        {(\ivr \limply \phi)}
        {\dbox{\ptest{\ivr}}{\phi}}
      }{}
      \cinferenceRule[evolveb|$\dibox{'}$]{evolve}
      {\linferenceRule[equiv]
        {\lforall{t{\geq}0}{\dbox{\pupdate{\pumod{x}{\solf(t)}}}{\phi}}}
        {\dbox{\hevolve{\D{x}=\genDE{x}}}{\phi}}
      }{\m{\D{\solf}(t)=\genDE{\solf}}}%
    \cinferenceRule[evolveinb|${[\&]}$]{evolution domain restriction} %
      {\linferenceRule[equiv]
        {\lforall{t_0{=}\stime}{\dbox{\hevolve{\D{\sol}=\genDE{\sol}}}{}}
        {{\big(\dbox{\hevolve{\D{\sol}=-\genDE{\sol}}}{(\stime\geq t_0\limply\ivr)} \limply \phi\big)}}}
        {\dbox{\hevolvein{\D{\sol}=\genDE{\sol}}{\ivr}}{\phi}}
      }{}%
      \cinferenceRule[choiceb|$\dibox{\cup}$]{axiom of nondeterministic choice}
      {\linferenceRule[equiv]
        {\dbox{\alpha}{\phi} \land \dbox{\beta}{\phi}}
        {\dbox{\pchoice{\alpha}{\beta}}{\phi}}
      }{}
      \cinferenceRule[composeb|$\dibox{{;}}$]{composition} %
      {\linferenceRule[equiv]
        {\dbox{\alpha}{\dbox{\beta}{\phi}}}
        {\dbox{\alpha;\beta}{\phi}}
      }{}
      \cinferenceRule[iterateb|$\dibox{{}^*}$]{iteration/repeat unwind} %
      {\linferenceRule[equiv]
        {\phi \land \dbox{\alpha}{\dbox{\prepeat{\alpha}}{\phi}}}
        {\dbox{\prepeat{\alpha}}{\phi}}
      }{}
      \cinferenceRule[K|K]{K axiom / modal modus ponens} %
      {\linferenceRule[impl]
        {\dbox{\alpha}{(\phi\limply\psi)}}
        {(\dbox{\alpha}{\phi}\limply\dbox{\alpha}{\psi})}
      }{}
      \cinferenceRule[I|I]{loop induction}
      {\linferenceRule[impl]
        {\dbox{\prepeat{\alpha}}{(\inv\limply\dbox{\alpha}{\inv})}}
        {(\inv\limply\dbox{\prepeat{\alpha}}{\inv})}
      }{}
      \cinferenceRule[C|C]{loop convergence}
      {\linferenceRule[impl]
        {\dbox{\prepeat{\alpha}}{\lforall{v{>}0}{(\mapply{\var}{v}\limply\ddiamond{\alpha}{\mapply{\var}{v-1}})}}}
        {\lforall{v}{(\mapply{\var}{v} \limply
            \ddiamond{\prepeat{\alpha}}{\lexists{v{\leq}0}{\mapply{\var}{v}}})}\qquad}
      }{\m{v\not\in\alpha}}%
      \cinferenceRule[B|B]{Barcan$\dbox{}{}\forall{}$} %
      {\linferenceRule[impl]
        {\lforall{x}{\dbox{\alpha}{\phi}}}
        {\dbox{\alpha}{\lforall{x}{\phi}}}
      }{\m{x\not\in\alpha}}
      \cinferenceRule[V|V]{vacuous $\dbox{}{}$}
      {\linferenceRule[impl]
        {\phi}
        {\dbox{\alpha}{\phi}}
      }{\m{FV(\phi)\cap BV(\alpha)=\emptyset}}%
      \cinferenceRule[G|G]{$\dbox{}{}$ generalization} %
      {\linferenceRule[formula]
        {\phi}
        {\dbox{\alpha}{\phi}}
      }{}
      \cinferenceRule[MP|MP]{modus ponens}
      {\linferenceRule[formula]
        {\phi\limply\psi \quad \phi}
        {\psi}
      }{}%
      \cinferenceRule[gena|$\forall$]{$\forall{}$ generalization}
      {\linferenceRule[formula]
        {\phi}
        {\lforall{x}{\phi}}
      }{}%
    \end{calculus}%
  \end{calculuscollections}
  \caption{Differential dynamic logic axiomatization}
  \label{fig:dL}
\end{figure}
We write \m{\infers \phi} iff \dL formula $\phi$ can be \emph{proved} with \dL rules from \dL axioms (including first-order rules and axioms); see \rref{fig:dL}.
That is, a \dL formula is inductively defined to be \dfn{provable} in the \dL calculus if it is an instance of a \dL axiom or if it is the conclusion (below the rule bar) of an instance of one of the \dL proof rules \irref{G}, \irref{MP}, \irref{gena}, whose premises (above the rule bar) are all provable.
Our axiomatization in \rref{fig:dL} is phrased in terms of $\dbox{\cdot}{}$.
Corresponding axioms hold for $\ddiamond{\cdot}{}$ by the defined duality \m{\ddiamond{\alpha}{\phi}\mequiv\lnot\dbox{\alpha}{\lnot\phi}}; see \cite{Platzer10} for explicit $\ddiamond{\cdot}{}$ rules.

Axiom \irref{assignb} is Hoare's assignment rule.
It uses substitutions to axiomatize discrete assignments.
To show that~$\mapply[x]{\phi}{x}$ is true after a discrete assignment, axiom \irref{assignb} shows that it has been true before, when substituting the affected variable~$x$ with its new value~$\theta$.
Formula $\mapply[x]{\phi}{\theta}$ is obtained from $\mapply[x]{\phi}{x}$ by \emph{substituting} $\theta$ for $x$,
provided $x$ does not occur in the scope of a quantifier or modality binding $x$ or a variable of $\theta$.
All substitutions in this paper require this admissibility condition.
A modality $\dbox{\alpha}{}$ containing \m{\pupdate{\pumod{z}{}}} or $\D{z}$ \emph{binds} $z$ (written $z\in BV(\alpha)$ for bound variable).
Only variables that are bound by HP $\alpha$ can possibly be changed when running $\alpha$.

Tests are proven by assuming
that the test succeeds with an implication in axiom \irref{testb}, because test~$\ptest{\ivr}$ can only make a transition when condition~$\ivr$ actually holds true.
From left to right, axiom \irref{testb} for \dL formula \m{\dbox{\ptest{\ivr}}{\phi}} assumes that formula $\ivr$ holds true (otherwise there is no transition and thus nothing to show) and shows that $\phi$ holds after the resulting no-op.
The converse implication from right to left is by case distinction.
Either $\ivr$ is false, then $\ptest{\ivr}$ cannot make a transition and there is nothing to show.
Or $\ivr$ is true, but then also $\phi$ is true.

In axiom \irref{evolveb}, $\solf(\cdot)$ is the (unique \cite[Theorem~10.VI]{Walter:ODE}) solution of the symbolic initial-value problem \m{\D{\solf}(t)=\genDE{\solf},\solf(0)=x}.
Given such a solution $\solf(\cdot)$, continuous evolution along differential equation \m{\hevolve{\D{x}=\genDE{x}}} can be replaced by a discrete assignment \m{\pupdate{\pumod{x}{\solf(t)}}} with an additional quantifier for the evolution time~$t$.
It goes without saying that variables like $t$ are fresh in \rref{fig:dL}.
Notice that conventional initial-value problems are numerical with concrete numbers $x\in\reals^n$ as initial values, not symbols $x$ \cite{Walter:ODE}.
This would not be enough for our purpose, because we need to consider all states in which the system could start, which may be uncountably many.
That is why axiom \irref{evolveb} solves one symbolic initial-value problem, because we could hardly solve uncountable many numerical initial-value problems.

Nondeterministic choices split into their alternatives in axiom \irref{choiceb}.
From right to left: If all $\alpha$ runs lead to states satisfying~$\phi$ (i.e., \m{\dbox{\alpha}{\phi}} holds) and all $\beta$ runs lead to states satisfying $\phi$ (i.e., \m{\dbox{\beta}{\phi}} holds), then all runs of HP \m{\pchoice{\alpha}{\beta}}, which may choose between following $\alpha$ and following $\beta$, also lead to states satisfying $\phi$ (i.e., \m{\dbox{\pchoice{\alpha}{\beta}}{\phi}} holds).
The converse implication from left to right holds, because \m{\pchoice{\alpha}{\beta}} can run all runs of $\alpha$ and all runs of $\beta$. 
A general principle behind the \dL axioms is most noticeable in axiom \irref{choiceb}.
The equivalence axioms of \dL are primarily intended to be used by reducing the formula on the left to the (structurally simpler) formula on the right.
With such a reduction, we symbolically decompose a property of a more complicated system into separate properties of easier fragments~$\alpha$ and~$\beta$. This decomposition makes the problem tractable and is good for scalability purposes.
For these symbolic structural decompositions, it is very helpful that \dL is a full logic that is closed under all logical operators, including disjunction and conjunction, for then both sides in \irref{choiceb} are \dL formulas again (unlike in Hoare logic \cite{DBLP:journals/cacm/Hoare69}).
This is also an advantage for computing invariants \cite{DBLP:conf/cav/PlatzerC08,DBLP:journals/fmsd/PlatzerC09,Platzer10}.

Sequential compositions are proven using nested modalities in axiom \irref{composeb}.
From right to left: If, after all $\alpha$-runs, all $\beta$-runs lead to states satisfying~$\phi$ (i.e., \m{\dbox{\alpha}{\dbox{\beta}{\phi}}} holds), then all runs of the sequential composition \m{\alpha;\beta} lead to states satisfying $\phi$ (i.e., \m{\dbox{\alpha;\beta}{\phi}} holds).
The converse implication uses the fact that if after all $\alpha$-runs all $\beta$-runs lead to $\phi$ (i.e., \m{\dbox{\alpha}{\dbox{\beta}{\phi}}}), then all runs of \m{\alpha;\beta} lead to $\phi$ (that is, \m{\dbox{\alpha;\beta}{\phi}}), because the runs of \m{\alpha;\beta} are exactly those that first do any $\alpha$-run, followed by any $\beta$-run.
Again, it is crucial that \dL is a full logic that considers reachability statements as modal operators, which can be nested, for then both sides in \irref{composeb} are \dL formulas (unlike in Hoare logic \cite{DBLP:journals/cacm/Hoare69}, where intermediate assertions need to be guessed or computed as weakest preconditions for $\beta$ and $\phi$).
Note that \dL can directly express weakest preconditions, because the \dL formula \m{\dbox{\beta}{\phi}} or any formula equivalent to it already is the weakest precondition for $\beta$ and $\phi$.
Strongest postconditions are expressible in \dL as well.

Axiom \irref{iterateb} is the iteration axiom, which partially unwinds loops.
It uses the fact that $\phi$ always holds after repeating $\alpha$ (i.e., \m{\dbox{\prepeat{\alpha}}{\phi}}), if $\phi$ holds at the beginning (for $\phi$ holds after zero repetitions then), and if, after one run of $\alpha$, $\phi$ holds after every number of repetitions of $\alpha$, including zero repetitions (i.e., \m{\dbox{\alpha}{\dbox{\prepeat{\alpha}}}{\phi}}).
So axiom \irref{iterateb} expresses that \m{\dbox{\prepeat{\alpha}}{\phi}} holds iff $\phi$ holds immediately and after one or more repetitions of $\alpha$.
Bounded model checking corresponds to unwinding loops $N$ times by axiom \irref{iterateb} and simplifying the resulting formula in the \dL calculus.
If the formula is invalid, a bug has been found, otherwise $N$ increases.
We use induction axioms \irref{I} and \irref{C} for proving formulas with unbounded repetitions of loops.

Axiom \irref{K} is the modal modus ponens from modal logic \cite{DBLP:journals/jsyml/Kripke59,Kripke63,HughesCresswell96}.
It expresses that, if an implication \m{\phi\limply\psi} holds after all runs of $\alpha$ (i.e., \m{\dbox{\alpha}{(\phi\limply\psi)}}) and $\phi$ holds after all runs of $\alpha$ (i.e., \m{\dbox{\alpha}{\phi}}), then $\psi$ holds after all runs of $\alpha$ (i.e., \m{\dbox{\alpha}{\psi}}), because $\psi$ is a consequence in each state reachable by $\alpha$.

Axiom \irref{I} is an induction schema for repetitions.
Axiom \irref{I} says that, if, after any number of repetitions of $\alpha$, invariant $\inv$ remains true after one (more) iteration of $\alpha$ (i.e., \m{\dbox{\prepeat{\alpha}}{(\inv\limply\dbox{\alpha}{\inv})}}), then $\inv$ holds after any number of repetitions of~$\alpha$ (i.e., \m{\dbox{\prepeat{\alpha}}{\inv}}) if $\inv$ holds initially.
That is, if~$\inv$ is true after running~$\alpha$ whenever~$\inv$ has been true before, then, if~$\inv$ holds in the beginning,~$\inv$ will continue to hold, no matter how often we repeat~$\alpha$ in \m{\dbox{\prepeat{\alpha}}{\inv}}.

Axiom \irref{C}, in which $v$ does not occur in $\alpha$ (written $v\not\in\alpha$), is a variation of Harel's convergence rule, suitably adapted to hybrid systems over $\reals$.
Axiom \irref{C} expresses that, if, after any number of repetitions of $\alpha$, $\var(v)$ can decrease after some run of~$\alpha$ by~1 (or another positive real constant) when $v>0$, then, if $\var(v)$ holds for any $v$, then the variant~$\var(v)$ holds for some real number~\m{v\leq0} after repeating~$\alpha$ sufficiently often (i.e., \m{\ddiamond{\prepeat{\alpha}}{\lexists{v{\leq}0}{\mapply{\var}{v}}}}). This axiom shows that positive progress with respect to~$\var(v)$ can be achieved by running~$\alpha$.
Note that positive progress is only sufficient if it is bounded from below, otherwise progress could converge to zero before reaching the destination.

Axiom \irref{B} is the Barcan formula of first-order modal logic, characterizing anti-monotonic domains \cite{HughesCresswell96}. In order for it to be sound for \dL, $x$ must not occur in $\alpha$.
It expresses that, if, from all initial values of $x$, all runs of $\alpha$ lead to states satisfying $\phi$, then, after all runs of $\alpha$, $\phi$ holds for all values of $x$, because the value of $x$ cannot affect the runs of $\alpha$, nor can $x$ change during runs of $\alpha$, since $x\not\in\alpha$.
The converse of \irref{B} is provable\footnote{
From \(\lforall{x}{\phi}\limply\phi\), derive \(\dbox{\alpha}{(\lforall{x}{\phi}\limply\phi)}\) by \irref{G}, from which \irref{K} and propositional logic derive \(\dbox{\alpha}{\lforall{x}{\phi}}\limply\dbox{\alpha}{\phi}\).
Then, first-order logic derives \(\dbox{\alpha}{\lforall{x}{\phi}}\limply\lforall{x}{\dbox{\alpha}{\phi}}\), as $x$ is not free in the antecedent.
}
\cite[BFC\,p. 245]{HughesCresswell96} and called \irref{B}.

Axiom \irref{V} is for vacuous modalities and requires that no free variable of $\phi$ (written $FV(\phi)$) is bound by $\alpha$, because $\alpha$ then cannot change any of the free variables of $\phi$.
It expresses that, if $\phi$ holds in a state, then it holds after all runs of $\alpha$, because, by \m{FV(\phi)\cap BV(\alpha)=\emptyset}, no variable that $\alpha$ can change occurs free in $\phi$.
The converse of \irref{V} holds, but we do not need it.
Note that, unlike the other axioms, \irref{B}, \irref{V}, and \irref{iterateb} are not strictly required for proving \dL formulas.

Rule \irref{G} is G\"odel's necessitation rule for modal logic \cite{HughesCresswell96}.
It expresses that, if $\phi$ is valid, i.e., true in all states, then $\dbox{\alpha}{\phi}$ is valid.
Note that, quite unlike rule \irref{G}, axiom \irref{V} crucially requires the variable condition that ensures that the value of $\phi$ is not affected by running $\alpha$ \cite{DBLP:conf/lics/Platzer12b}.

Rules \irref{MP} and \irref{gena} are as in first-order logic.
Modus ponens (\irref{MP}) expresses that if we know that both $\phi\limply\psi$ and $\phi$ are valid, then $\psi$ is a valid consequence.
The $\forall$-generalization rule (\irref{gena}) expresses that if $\phi$ is valid, then so is $\lforall{x}{\phi}$.

The \dL axiomatization in \rref{fig:dL} uses a modular \dL axiom \irref{evolveinb} that reduces differential equations with evolution domain constraints to differential equations without them by checking the evolution domain constraint backwards along the reverse flow \cite{DBLP:conf/lics/Platzer12b}.
It checks $\ivr$ backwards from the end of the evolution up to the initial time $t_0$, using that \m{\hevolve{\D{x}=-\genDE{x}}} follows the same flow as \m{\hevolve{\D{x}=\genDE{x}}}, but backwards.
See \rref{fig:backflow}%
\begin{wrapfigure}{r}{7cm}
  \newcommand{\mtime}{t}%
  \begin{minipage}[b]{7cm}
    \begin{tikzpicture}[scale=1.5]
  \newcommand{\ws}{\nu}\newcommand{\wt}{}%
  \renewcommand{\I}{\iconcat[state=\ws]{\stdI}}%
  \renewcommand{\It}{\iconcat[state=\wt]{\stdI}}%
  \def\vec#1{#1}%
  \tikzstyle{axes}=[]
  \tikzstyle{mode switch}=[black!70,thin,dotted]
      \begin{scope}[style=axes]
        \draw[->] (-0.1,0) -- (2.4,0) node[right] {$\mtime$} coordinate(t axis);
        \draw[->] (0,-0.1) -- (0,1.2) node[above] {$\vec{x}$} coordinate(x axis);
      \end{scope}
      {
        \draw[draw=vgreen,fill=vgreen!5] (1.1,0.8) ellipse (0.9cm and 0.4cm);
        \node[color=vgreen!140] at (1.6,0.6) {$\ivr$};
      }
      \newcommand{\breakp}{1.8}
      \begin{scope}[xshift=0.7cm,yshift=-0.1cm]
        {
          \draw[thick,domain=-0.7:0.6,smooth,xshift=.5cm]
            plot
            (\x,{exp(-1.5*\x)+1.2*(1-exp(-1.5*\x))})
          node (flowend) {}
          node[above] {$\iget[state]{\It}$};
        }
        \node (flowstart) at (-0.7,0.62847) {};
        \draw[thick,vred,domain=-0.7:0.6,smooth,xshift=0.5cm,yshift=-0.11cm]
          plot[mark=triangle*,mark options={vred},mark phase=5,mark repeat=8]
          (\x,{exp(-1.5*\x)+1.2*(1-exp(-1.5*\x))})
          node (backstart) {};
        \draw[thick,dotted,vred,domain=-1:-0.7,smooth,xshift=0.5cm,yshift=-0.11cm] plot (\x,{exp(-1.5*\x)+1.2*(1-exp(-1.5*\x))});
        \tikzstyle{my loop}=[->,to path={
          .. controls +(10:0.5) and +(-10:0.5) .. (\tikztotarget) \tikztonodes}]
        \draw[thick,my loop] (flowend) to (backstart);
        \node[right,text width=4cm] at (1.5,1) {\footnotesize revert flow and time $\stime$;\\check~$\ivr$ backwards};
      \end{scope}
      \node[above=-1pt,rotate=5] at (1,0.9) {$\hevolve{\D{x}={\genDE{x}}}$};
      \draw[mode switch] (0.5,0) node[below,black] {$t_0=\stime$} -- ++(0,1.1);
      \draw[mode switch] (\breakp,0) node[below,black] {$r$} -- ++(0,1.1);
      \path (0.7,0.4) -- node[below,vred] {$\hevolve{\D{x}={-\genDE{x}}}$} (\breakp,0.4);
    \end{tikzpicture}
  \end{minipage}
  \caption[``There and back again'' axiom checks evolution domain along backwards flow over time]{``There and back again'' axiom\,\irref{evolveinb} checks evolution domain along backwards flow over time}
  \label{fig:backflow}
\end{wrapfigure}
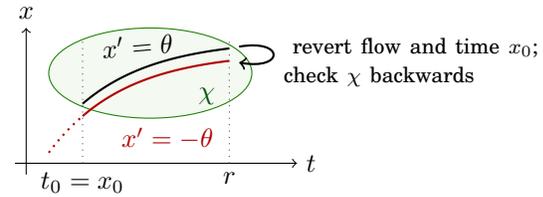
for an illustration.
To simplify notation, we assume that the (vector) differential equation \m{\hevolve{\D{x}=\genDE{x}}} in axiom \irref{evolveinb} already includes a clock \m{\hevolve{\D{\stime}=1}} for tracking time.
The idea behind axiom \irref{evolveinb} is that the fresh variable $t_0$ remembers the initial time $\stime$, then $x$ evolves forward along \m{\pevolve{\D{x}=\genDE{x}}} for any amount of time.
Afterwards, $\phi$ has to hold if, for all ways of evolving backwards along \m{\pevolve{\D{x}=-\genDE{x}}} for any amount of time, \m{\stime\geq t_0\limply\ivr} holds, i.e., $\ivr$ holds at all previous times that are later than the initial time $t_0$.
Thus, $\phi$ is not required to hold after a forward evolution if the evolution domain constraint $\ivr$ can be left by evolving backwards for less time than the forward evolution took.

The following loop invariant rule \irref{invind} derives from \irref{G} and \irref{I}.
Convergence rule \irref{con} derives from $\forall$-generalization, \irref{G}, and \irref{C} (like in \irref{C}, $v$ does not occur in $\alpha$):
\[
      \dinferenceRule[invind|$ind$]{inductive invariant}
      {\linferenceRule[formula]
        {\inv\limply\dbox{\alpha}{\inv}}
        {\inv\limply\dbox{\prepeat{\alpha}}{\inv}}
      }{}
      \qquad\qquad
      \dinferenceRule[con|$con$]{loop convergence right} %
      {\linferenceRule[formula]
        {\mapply{\var}{v}\land v>0\limply\ddiamond{\alpha}{\mapply{\var}{v-1}}}
        {\mapply{\var}{v} \limply
            \ddiamond{\prepeat{\alpha}}{\lexists{v{\leq}0}{\mapply{\var}{v}}}}
      }{}%
\]
While this is not the focus of this paper, we note that we have successfully used a refined sequent calculus variant \cite{DBLP:journals/jar/Platzer08} of the Hilbert calculus in \rref{fig:dL} for automatic verification of hybrid systems, including trains, cars, and aircraft; see \rref{sec:KeYmaera}.
Several different verification paradigms can be formulated for the \dL calculus by choosing in which order to use the axioms, including proving by symbolic execution, proving by forward image computation, proving by backward image computation, proving by fixpoint loops, and full deduction \cite{Platzer10}.

Uses of real arithmetic, which, we denote by \irref{qear}, are decidable by quantifier elimination in real-closed fields \cite{tarski_decisionalgebra51}.
\begin{definition}[Quantifier elimination] \label{def:qelim}
  A first-order theory admits
  \dfn[quantifier~elimination]{quantifier elimination} if, with each formula~$\phi$, a quantifier-free formula $\qelim{\phi}$\indexn[_QE]{\QE}\index{QE~(quantifier~elimination)@$\QE$~(quantifier~elimination)} can be associated effectively that is equivalent (i.e., \(\phi\lbisubjunct\qelim{\phi}\) is valid) and has no additional free variables\index{symbol!free} or function symbols. The operation $\qelim{}$ is further assumed to evaluate formulas without variables,
  yielding a decision procedure for closed formulas\index{formula!closed}\index{closed!formula|see{formula, closed}} of this theory (i.e., formulas without free variables)\index{free~variable}.
\end{definition}
Quantifier elimination is decidable in the first-order logic of real-closed fields \cite{tarski_decisionalgebra51}.
It exploits the special structure of real arithmetic to express quantified arithmetic formulas equivalently without quantifiers.
\begin{example}[Quantifier elimination]
\qelim{} yields the equivalence:
\[
 \qelim{\lexists{x}{(ax^2+bx+c=0)}} \mequiv
 (a\neq0 \land b^2-4ac\geq0) \lor (a=0 \land (b=0\limply c=0))
\]
In this particular case, the equivalence can be found by using the generic condition for solvability of quadratic equations over the reals plus special cases when coefficients are zero.
\end{example}
Quantifier elimination gives a decision procedure for real arithmetic \cite{tarski_decisionalgebra51}. Implementations use partial cylindrical algebraic decomposition \cite{DBLP:journals/jsc/CollinsH91}, virtual substitution \cite{DBLP:journals/aaecc/Weispfenning97}, semidefinite programming relaxations \cite{DBLP:journals/mp/Parrilo03,DBLP:conf/tphol/Harrison07} for Stengle's Positivstellensatz \cite{DBLP:journals/mathann/Stengle74}, or Gr\"obner bases for the real Nullstellensatz \cite{DBLP:conf/cade/PlatzerQR09}, which combine Gr\"obner bases \cite{Buchberger65} with Stengle's real Nullstellensatz \cite{DBLP:journals/mathann/Stengle74} and semidefinite programming \cite{BoydV04}.

For the purposes of this survey, we denote the use of decidable real arithmetic and quantifier elimination in proofs by \irref{qear}.
More constructive deduction modulo proof rules, which can be used to combine first-order real arithmetic with the proof calculus presented here and that are suitable for automation, have been reported in previous work \cite{DBLP:journals/jar/Platzer08,DBLP:journals/logcom/Platzer10,Platzer10}.
Those are based on real-valued free variables, Skolemization, Deskolemization, and the following lifting of quantifier elimination in real-closed fields \cite{tarski_decisionalgebra51,DBLP:conf/automata/Collins75}.
\begin{lemma}[Quantifier elimination lifting \cite{DBLP:journals/jar/Platzer08}] \label{lem:qelim-lift}
  Quantifier elimination can be lifted to instances of formulas of first-order theories that admit quantifier elimination,
  i.e., to formulas that result from the base theory by substitution.
  \index{quantifier~elimination!lifting}
  \index{lifting|see{quantifier~elimination, lifting}}
\end{lemma}

\begin{example}[Single car] \label{ex:dL-proof}
     \def\MA{m}%
     \def\prem{v\geq0\land x\leq m}%
In order to illustrate how the \dL calculus can be used to prove \dL formulas and identify parameter constraints required for them to be valid, we consider a \dL formula for the braking case of HP \rref{eq:ex-HP}:
\begin{equation}
v\geq0\land x\leq m \limply \dbox{\pupdate{\pumod{a}{-b}};\hevolve{\D{x}=v\syssep\D{v}=a}}{x\leq m}
\label{eq:car-single-braking-discovery}
\end{equation}
Formula \rref{eq:car-single-braking-discovery} claims a car would never run a stoplight if it starts before the stoplight (\m{x\leq m}) and is applying the brakes.
Since braking is the safest operation for cars, we might think that car control would always be safe in this most conservative scenario.
But that is not the case.
If the car starts off too fast compared to the remaining distance to the stoplight, then not even braking can prevent a crash.
We can easily find out, however, under which circumstance the \dL formula \rref{eq:car-single-braking-discovery} is valid by applying \dL axioms to it.
The following \dL proof reveals that \rref{eq:car-single-braking-discovery} is valid if \m{v^2 \leq 2\itimes b\itimes(m-x)} holds initially:
     \renewcommand{\arraystretch}{1.3}%
      \begin{sequentdeduction}[array]
      \linfer[qear]
      {\lsequent{\prem} {v^2 \leq 2\itimes b\itimes(\MA-x)}}
      {\linfer[assignb]
        {\lsequent{\prem}{\lforall{t{\geq}0}{(\frac{-b}{2}\itimes t^2+v\itimes t + x\leq\MA)}}}
        {\linfer[assignb]
            {\lsequent{\prem}{\dbox{\pupdate{\pumod{a}{-b}}}{\lforall{t{\geq}0}{(\frac{a}{2}\itimes t^2+v\itimes t + x\leq\MA)}}}}
          {\linfer[evolveb]
              {\lsequent{\prem}{\dbox{\pupdate{\pumod{a}{-b}}}{\lforall{t{\geq}0}{\dbox{\hupdate{\humod{x}{\frac{a}{2}\itimes t^2+v\itimes t + x}}}{x\leq\MA}}}}}
            {\linfer[composeb]
              {\lsequent{\prem}{\dbox{\pupdate{\pumod{a}{-b}}}{\dbox{\hevolve{\D{x}=v\syssep\D{v}=a}}{x\leq\MA}}}}
              {\lsequent{\prem}{\dbox{\pupdate{\pumod{a}{-b}};\hevolve{\D{x}=v\syssep\D{v}=a}}{x\leq\MA}}}
            }
          }
        }
      }
      \end{sequentdeduction}
We follow the conventions in sequent calculus and read proofs bottom-up, from the desired conclusion at the bottom to the premises at the top; see previous work \cite{DBLP:journals/jar/Platzer08} for more details on a sequent calculus for \dL.
Here, we first apply the axiom \irref{composeb} to reduce the sequential composition equivalently to a nested modality.
Then we use axiom \irref{evolveb} to reduce the differential equation to an assignment with the solution and a quantifier $\forall t$ for its duration.
Even though it is a quantifier over a real variable, we cannot use the decision procedure of quantifier elimination for real-closed fields \cite{tarski_decisionalgebra51} to handle it, because we do not have a formula of first-order real arithmetic, but still a \dL formula with a modality expressing a property of all reachable states.
Instead, we first use axiom \irref{assignb} twice to equivalently substitute in the effect of the assignments.
Finally, we use equivalences of real arithmetic (using quantifier elimination, denoted \irref{qear}) to discover the constraint \m{v^2 \leq 2\itimes b\itimes(m-x)}, which is required to make \rref{eq:car-single-braking-discovery} valid.
Indeed, if we add this so-discovered constraint about the initial state, the following \dL formula is provable in the \dL calculus by a minor variation of the above \dL proof:
\[
v^2 \leq 2\itimes b\itimes(m-x)\land
v\geq0\land x\leq m \limply 
\dbox{\pupdate{\pumod{a}{-b}};\hevolve{\D{x}=v\syssep\D{v}=a}}{x\leq m}
\]
This construction explains why \m{v^2 \leq 2\itimes b\itimes(m-x)} has to be assumed in \dL formula \rref{eq:car-single-essentials}, because braking is one of the choices for its HP $\textit{car}_\varepsilon$.
Constructions based on this principle turn out to be very effective for discovering invariants \cite{DBLP:conf/cav/PlatzerC08,DBLP:journals/fmsd/PlatzerC09,Platzer10} and constraints on the free parameters of the system \cite{DBLP:journals/jar/Platzer08,DBLP:conf/icfem/PlatzerQ09,Platzer10}, or design constraints for closed-loop properties \cite{DBLP:conf/acc/ArechigaLPK12}.
Another interesting observation is that parameter constraints discovered in this way from precise proofs about simplified models are useful for deriving design decisions about the full system, even for aspects that have not been modeled, like camera resolutions and frame rates for video-based car safety technology \cite{DBLP:conf/iccps/MitschLP12,Platzer10}.

\begin{figure*}[tb]
\footnotesize
\def\prem{C}%
\def\concl{E}%
\[\hspace*{-0.5cm}
\begin{array}{@{}l@{}l@{~}l@{}l@{}}
 & \prem \limply \dbox{\pupdate{\pumod{a}{-b}}}{\dbox{\hevolve{\D[2]{x}=a}}{\concl}} &&\text{1}\\
 & \dbox{\pupdate{\pumod{a}{-b}};\hevolve{\D[2]{x}=a}}{\concl}\lbisubjunct\dbox{\pupdate{\pumod{a}{-b}}}{\dbox{\hevolve{\D[2]{x}=a}}{\concl}} &&\text{2: by \irref{composeb}}\\
 & (\dbox{\pupdate{\pumod{a}{-b}};\hevolve{\D[2]{x}=a}}{\concl}\lbisubjunct\dbox{\pupdate{\pumod{a}{-b}}}{\dbox{\hevolve{\D[2]{x}=a}}{\concl}}) \limply ((\prem \limply \dbox{\pupdate{\pumod{a}{-b}}}{\dbox{\hevolve{\D[2]{x}=a}}{\concl}})\limply (\prem \limply \dbox{\pupdate{\pumod{a}{-b}};\hevolve{\D[2]{x}=a}}{\concl})) && \text{3: by \irref{qear}}\\
 &(\prem \limply \dbox{\pupdate{\pumod{a}{-b}}}{\dbox{\hevolve{\D[2]{x}=a}}{\concl}})\limply (\prem \limply \dbox{\pupdate{\pumod{a}{-b}};\hevolve{\D[2]{x}=a}}{\concl}) && \text{4: \irref{MP}(2,3)}\\
 &\prem \limply \dbox{\pupdate{\pumod{a}{-b}};\hevolve{\D[2]{x}=a}}{\concl} && \text{5:  \irref{MP}(4,1)}
\end{array}
\]
  \caption{\dL Hilbert proof corresponding to bottom proof step in \rref{ex:dL-proof}, abbreviating $v\geq0 \land x\leq\MA$ by $\prem$, $x\leq\MA$ by $\concl$, and \m{\hevolve{\D{x}=v\syssep\D{v}=a}} by \m{\hevolve{\D[2]{x}=a}}}
  \label{fig:ex-dL-Hilbert-proof}
\end{figure*}

The reader should note that we use an abbreviated notation for proofs here.
Without abbreviations, the bottom-most proof step \irref{assignb} in the above proof expands to the Hilbert-style \dL proof shown in \rref{fig:ex-dL-Hilbert-proof}, in which each line corresponds to a \dL axiom or the result of a \dL proof rule applied to previous lines as indicated in the column on the right.
In this paper, we abbreviate Hilbert proofs by rewriting formulas using the \dL axioms directly, e.g., by replacing instances of the left-hand side of an equivalence in a \dL axiom by the corresponding (structurally) simpler right-hand side.
This abbreviated style can be understood systematically in a sequent calculus formulation \cite{DBLP:journals/jar/Platzer08} of the \dL proof calculus.

For typesetting purposes for a \dL proof for \dL formula \rref{eq:car-single-essentials}, let us abbreviate \m{\hevolvein{\D{x}=v\syssep\D{v}=a\syssep\D{t}=1}{v\geq0\land t\leq\varepsilon}} in the HP $\textit{car}_\varepsilon$ from \rref{eq:car-single-essentials} by \m{\hevolvein{\D[2]{x}=a}{t\leq\varepsilon}}.
Furthermore, we abbreviate \m{v^2\leq2b(m-x) \land A\geq0\land b>0} by $\inv$.
\begin{sequentdeduction}[array]
\linfer[testb]
  {\lsequent{\inv} {
    (\ivr\limply\dbox{\pupdate{\pumod{a}{A}}}{\dbox{\pupdate{\pumod{t}{0}}; \hevolvein{\D[2]{x}=a}{t\leq\varepsilon}}{\inv}})
    \land \dbox{\pupdate{\pumod{a}{-b}}}{\dbox{\pupdate{\pumod{t}{0}}; \hevolvein{\D[2]{x}=a}{t\leq\varepsilon}}{\inv}}}}
  {\linfer[composeb]
    {\lsequent{\inv} {
      \dbox{\ptest{\ivr}}{\dbox{\pupdate{\pumod{a}{A}}}{\dbox{\pupdate{\pumod{t}{0}}; \hevolvein{\D[2]{x}=a}{t\leq\varepsilon}}{\inv}}}
      \land \dbox{\pupdate{\pumod{a}{-b}}}{\dbox{\hevolve{\pupdate{\pumod{t}{0}}; \D[2]{x}=a}}{\inv}}}}
    {\linfer[choiceb]
      {\lsequent{\inv} {
        \dbox{\ptest{\ivr};\pupdate{\pumod{a}{A}}}{\dbox{\pupdate{\pumod{t}{0}}; \hevolvein{\D[2]{x}=a}{t\leq\varepsilon}}{\inv}}
        \land \dbox{\pupdate{\pumod{a}{-b}}}{\dbox{\pupdate{\pumod{t}{0}}; \hevolvein{\D[2]{x}=a}{t\leq\varepsilon}}{\inv}}}}
      {\linfer[composeb]
        {\lsequent{\inv} {\dbox{\pchoice{(\ptest{\ivr};\pupdate{\pumod{a}{A}})}{\pupdate{\pumod{a}{-b}}}}{\dbox{\pupdate{\pumod{t}{0}}; \hevolvein{\D[2]{x}=a}{t\leq\varepsilon}}{\inv}}}}
        {\linfer[invind]
          {\lsequent{\inv} {\dbox{(\pchoice{(\ptest{\ivr};\pupdate{\pumod{a}{A}})}{\pupdate{\pumod{a}{-b}}}); \pupdate{\pumod{t}{0}}; \hevolvein{\D[2]{x}=a}{t\leq\varepsilon}}{\inv}}}
          {\linfer[K+G]
            {\lsequent{\inv} {\dbox{\prepeat{((\pchoice{(\ptest{\ivr};\pupdate{\pumod{a}{A}})}{\pupdate{\pumod{a}{-b}}}); \pupdate{\pumod{t}{0}}; \hevolvein{\D[2]{x}=a}{t\leq\varepsilon})}}{\inv}}}
            {\lsequent{\inv} {\dbox{\prepeat{((\pchoice{(\ptest{\ivr};\pupdate{\pumod{a}{A}})}{\pupdate{\pumod{a}{-b}}}); \pupdate{\pumod{t}{0}}; \hevolvein{\D[2]{x}=a}{t\leq\varepsilon})}}{x\leq m}}}
          }
        }
    }
  }
}
\end{sequentdeduction}
The first (bottom-most) step uses that \m{\phi\limply x\leq m} is provable in arithmetic.
From this, \irref{G} proves \(\dbox{\textit{car}_\varepsilon}{(\phi\limply x\leq m)}\), and, thus, \irref{K} proves \(\dbox{\textit{car}_\varepsilon}{\phi} \limply \dbox{\textit{car}_\varepsilon}{x\leq m}\).
The right conjunct in the top-most premise can be proven by a minor variation of the \dL proof for \rref{eq:car-single-braking-discovery}.
The \dL proof for the left conjunct in the premise works similarly, but requires slightly more involved arithmetic and the choice
\(\chi\mequiv 2b(x-m)\geq v^2+\big(A+b\big)\big(A\varepsilon^2+2\varepsilon v\big)\).
A full proof about a similar system can be found in our book \cite[Section 2.9]{Platzer10}.
\end{example}

\subsection{Soundness and Completeness} \label{sec:dL-complete}

The \dL calculus is \emph{sound} \cite{DBLP:journals/jar/Platzer08,DBLP:conf/lics/Platzer12b}, that is, every formula that is provable using the \dL axioms and proof rules is valid, i.e., true in all states.
That is, for all \dL formulas $\phi$:
\begin{equation}
  \infers\phi ~\text{implies}~ \entails\phi
  \label{eq:sound}
\end{equation}
Soundness should be \textit{sine qua non} for formal verification, but, for fundamental reasons \cite{DBLP:conf/hybrid/PlatzerC07,DBLP:journals/mst/Collins07}, is so complex for hybrid systems that it is sometimes inadvertently forsaken.
In logic, we ensure soundness easily just by checking it locally once for each axiom and proof rule.
Thus, no matter how complicated a proof, the proven \dL formula is valid, because it is a (complicated) consequence of lots simple valid proof steps.

More intriguingly, however, our logical setting also enables us to ask the converse: is the \dL proof calculus \emph{complete}, i.e., can it prove all that is true? That is, does the converse of \rref{eq:sound} hold?
A simple corollary to G\"odel's incompleteness theorem shows that already the fragments for discrete dynamical systems and for continuous dynamical systems are incomplete.
\begin{theorem}[Incompleteness \cite{DBLP:journals/jar/Platzer08}] \label{thm:dL-incomplete}
  \index{incomplete!dL@\dL}
  Both the discrete fragment and the continuous fragment of \dL are \dfn[axiomatize]{not effectively axiomatizable}, i.e., they have no sound and complete effective calculus, because natural numbers are definable\index{definable} in both fragments.
  \index{integer!arithmetic}
  \index{natural!number!definable}
\end{theorem}
\begin{proof}
  G\"odel's incompleteness theorem \cite{Goedel_1931} applies to the discrete fragment of \dL, because natural numbers are definable in that fragment by repeated addition:
  \[
  \textit{nat}(n) ~\lbisubjunct~ \ddiamond{\pupdate{\umod{x}{0}};\prepeat{(\pupdate{\umod{x}{x+1}})}}{~x=n}
  \index{_nat_@$\textit{nat}$}
  \]
  \begin{figure}[tbh]
    \centering
    \begin{tikzpicture}
      \begin{scope}
        \draw[->] (-0.1,0) -- (8.3,0) node[right] {$\tau$} coordinate(t axis);
        \draw[->] (0,-1.2) -- (0,1.2) node[above] {$s$} coordinate(x axis);
      \end{scope}
      \fill[draw=blue,fill=blue!20,domain=0:7.8539815,smooth] plot[id=sin2x] function{sin(2*x)} |- (0,0);
      \foreach \A in {0,1.5707963,3.1415926,...,9}
      {
        \fill[vred] (\A,0) circle (2pt);
      }
      \draw (1.5707963-0.15,0) node[below] {$\phantom{1}\pi$} -- ++(0,0);
      \foreach \A in {3,5}
      {
        \draw (\A*1.5707963-0.15,0) node[below] {$\A\pi$} -- ++(0,0);
      }
      \foreach \A in {2,4}
      {
        \draw (\A*1.5707963+0.15,0) node[below] {$\A\pi$} -- ++(0,0);
      }
    \end{tikzpicture}
    \caption{Characterization of~$\naturals$ as zeros of solutions of differential equations}
    \label{fig:dL-incomplete-sin}
  \end{figure}
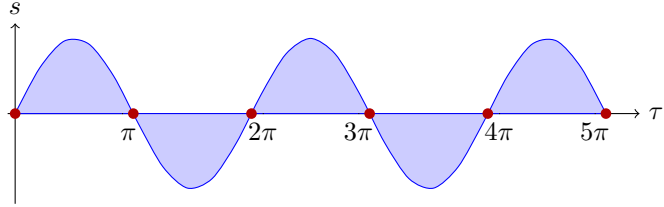
  G\"odel's incompleteness theorem \cite{Goedel_1931} applies to the continuous fragment of \dL, because an isomorphic copy of the natural numbers is definable using linear differential equations, which characterize solutions $\sin$ and $\cos$, whose zeros (see \rref{fig:dL-incomplete-sin}), as detected by $\tau$, correspond to natural numbers, scaled by $\pi$:
  \[
  \textit{nat}(n) ~\lbisubjunct~
  \lexists{s}{\lexists{c}{\lexists{\tau}{(s=0\land c=1\land\tau=0 \land
 \ddiamond{\hevolve{\D{s}={c}\syssep\D{c}={-s}\syssep\D{\htime}={1}}}{(s=0 \land \htime=n)})}}}
  \]
  \qedhere
\end{proof}
Incompleteness is not unexpected since hybrid systems contain a computationally complete sublanguage and because reachability of hybrid systems is not semidecidable \cite{DBLP:conf/lics/Henzinger96}.
Yet, it is instructive to understand the above simple proof based on a classical result about logic.

In logic, the suitability of an axiomatization can still be established by showing completeness relative to a fragment \cite{DBLP:journals/siamcomp/Cook78,DBLP:conf/stoc/HarelMP77}.
This \emph{relative completeness}, in which we assume we were able to prove valid formulas in a fragment and prove that we can then prove all others, tells us how subproblems are related computationally.
It tells us whether one subproblem dominates the others.
Standard relative completeness \cite{DBLP:journals/siamcomp/Cook78,DBLP:conf/stoc/HarelMP77}, however, which works relative to the data logic, is inadequate for hybrid systems, whose complexity comes from the dynamics, not the data logic, first-order real arithmetic, which is perfectly decidable \cite{tarski_decisionalgebra51}.

We have shown that both the original \dL sequent calculus \cite{DBLP:journals/jar/Platzer08} and the Hilbert-type calculus in \rref{fig:dL} \cite{DBLP:conf/lics/Platzer12b} are sound and complete axiomatizations of \dL relative to the continuous fragment (\FOD).
\FOD is the \emph{first-order logic of differential equations}, i.e., first-order real arithmetic augmented with formulas expressing properties of differential equations, that is, \dL formulas of the form \m{\dbox{\hevolve{\D{x}=\genDE{\theta}}}{F}} with a first-order formula~$F$.
Note that axioms \irref{B} and \irref{V} are not needed for the proof of the following theorem.

\begin{theorem}[Relative completeness of \dL \cite{DBLP:journals/jar/Platzer08,DBLP:conf/lics/Platzer12b}] \label{thm:dL-complete}
  \index{complete!relatively!dL@\dL}
  The \dL calculus is a \emph{sound and complete axiomatization} of hybrid systems relative to \FOD, i.e.,
  every valid \dL formula can be derived from \FOD tautologies:
  \[
  \entails\phi ~~\text{iff}~~ \text{\upshape Taut}_{\FOD} \infers \phi
  \]
\end{theorem}

This central result shows that we can prove properties of hybrid systems in the \dL calculus exactly as good as properties of differential equations can be proved.
One direction is obvious, because differential equations are part of hybrid systems, so we can only understand hybrid systems to the extent that we can reason about their differential equations.
We have shown the other direction by proving that all true properties of hybrid systems can be reduced effectively in the \dL calculus to elementary properties of differential equations.
Moreover, the \dL proof calculus for hybrid systems can perform this reduction constructively and, vice versa, provides a provably perfect lifting of every approach for differential equations to hybrid systems.

Another important consequence of this result is that decomposition can be successful in taming the complexity of hybrid systems.
The \dL proof calculus is strictly compositional.
All proof rules prove logical formulas or properties of HPs by reducing them to structurally simpler \dL formulas.
As soon as we understand that the hybrid systems complexity comes from a combination of several simpler aspects, we can, hence, tame the system complexity by reducing it to analyzing the dynamical effects of simpler parts.
This decomposition principle makes it possible for \dL proofs to scale to interesting systems in practice.
\rref{thm:dL-complete} gives the theoretical evidence why this principle works in general, not just in the case studies we have considered so far.
This is a good illustration of our principle of multi-dynamical systems and even a proof that the decompositions behind the multi-dynamical systems approach are successful.
Note that, even though \rref{thm:dL-complete} proves (constructively) that every true property of hybrid systems can be proved in the \dL calculus by decomposition from elementary properties of differential equations, it is still an interesting question which decompositions are most efficient.

For an even more surprising ``converse'' result proving a sound and complete axiomatization of \dL relative to the discrete fragment of \dL, we refer to recent work \cite{DBLP:conf/lics/Platzer12b}.
That proof is again a constructive reduction, proving that hybrid dynamics, continuous dynamics, and discrete dynamics are proof-theoretically equivalently reducible in the \dL calculus. Even though the nature of each kind of dynamics is fundamentally different, they still enjoy a perfect proof-theoretical correspondence.
In a nutshell, we have shown that we can proof-theoretically equate:
\[
\text{``}\textit{hybrid} = \textit{continuous} = \textit{discrete}\text{''}
\]
A discussion of this fundamental result about the nature of hybridness is beyond the scope of this paper; we refer to previous work \cite{DBLP:conf/lics/Platzer12b}.

\subsection{Differential Invariants} \label{sec:diffind}
{\def\inv{F}%

The \dL axiomatization in \rref{fig:dL} is sound and complete relative to \FOD.
But \rref{fig:dL} only has a very simple proof rule for differential equations (\irref{evolveb}) based on computing a solution of the differential equation; we refer to previous work for discretization techniques \cite{DBLP:conf/lics/Platzer12b}.
For proving more complicated differential equations by induction, \dL provides \emph{differential invariants} and \emph{differential variants} \cite{DBLP:journals/logcom/Platzer10}, which have been introduced in 2008 \cite{DBLP:journals/logcom/Platzer10} and later refined to a procedure that computes differential invariants in a fixed-point loop \cite{DBLP:conf/cav/PlatzerC08,DBLP:journals/fmsd/PlatzerC09}.
All premier proof principles for discrete loops are based on some form of induction.
\rref{thm:dL-complete} and its discrete converse \cite{DBLP:conf/lics/Platzer12b} prove that verification techniques that are successful for discrete systems generalize to continuous and hybrid systems and vice versa.
Differential invariants and differential variants can be considered as one (of many possible) constructive and practical consequences of this result.
Differential induction defines induction for differential equations.
It resembles induction for discrete loops (rule \irref{invind}) but works for differential equations instead and uses a \emph{differential formula} (\m{\subst[\D{\inv}]{\D{x}}{\theta}}, which we develop below) for the induction step.

\begin{center}
\begin{calculus}
  \cinferenceRule[diffind|DI]{differential invariant}
  {\linferenceRule[sequent]
    {\lsequent{\ivr}{\subst[\D{\inv}]{\D{x}}{\theta}}}
    {\lsequent{\inv}{\dbox{\hevolvein{\D{x}=\theta}{\ivr}}{\inv}}}
  }{}
\end{calculus}
\end{center}
\noindent
This \emph{differential induction} rule is a natural induction principle for differential equations.
The difference compared to ordinary induction for discrete loops is that the evolution domain constraint~$\ivr$ is assumed in the premise (because the continuous evolution is not allowed to leave its evolution domain constraint) and that the induction step uses the differential formula \m{\subst[\D{\inv}]{\D{x}}{\theta}} corresponding to formula $\inv$ and the differential equation \m{\hevolve{\D{x}=\theta}} instead of a statement that the loop body preserves the invariant.
Intuitively, the \emph{differential formula} \m{\subst[\D{\inv}]{\D{x}}{\theta}} captures the infinitesimal change of formula $\inv$ over time along \m{\hevolve{\D{x}=\theta}}, and expresses the fact that $\inv$ is only getting more true when following the differential equation \m{\hevolve{\D{x}=\theta}}.
The semantics of differential equations is defined in a mathematically precise but computationally intractable way using analytic differentiation and limit processes at infinitely many points in time.
The key point about differential invariants is that they replace this precise but computationally intractable semantics with a computationally effective, algebraic and syntactic total derivative \m{\D{\inv}} along with simple substitution of differential equations.
The valuation of the resulting computable formula \m{\subst[\D{\inv}]{\D{x}}{\theta}} along differential equations coincides with analytic differentiation.

\begin{definition}[Derivation] \label{def:derivation}
  The operator~\m{\DD{}\ignore{:\lterms{\Sigma{\cup}\D{\Sigma}}{V}\to\lterms{\Sigma{\cup}\D{\Sigma}}{V}}} that is defined as follows on terms is called \emph{syntactic (total) derivation}\index{derivation!syntactic}:
  \begin{subequations}
  \begin{align}
    \der{r} & = 0
      \hspace{2.1cm}\text{for numbers}~r\in\rationals
    \label{eq:Dconstant}\\
    \der{x} & =  \D{x}
      \hspace{2cm}\text{for variable}~x\label{eq:Dpolynomial}\\
    \der{a+b} & = \der{a} + \der{b}
    \label{eq:Dadditive}\\
    \der{a-b} & = \der{a} - \der{b}
    \label{eq:Dsubtractive}\\
    \der{a\cdot b} & = \der{a}\cdot b + a\cdot\der{b}
    \label{eq:DLeibniz}\\
    \der{a / b} & = (\der{a}\cdot b - a\cdot\der{b}) / b^2
    \label{eq:Dquotient}
  \end{align}
  \label{eq:Dterm}
  \end{subequations}
  \index{differential!symbol}%
  We extend it to (quantifier-free) first-order real-arithmetic formulas $F$ as follows:
  \begin{subequations}
  \begin{align}
    \der{F\land G} &\,\mequiv\, \der{F} \land \der{G}\\
    \der{F\lor G} &\,\mequiv\, \der{F} \land \der{G}
    \\
    \der{a\geq b} &\,\mequiv\, \der{a} \geq \der{b}
    \hspace{1cm}\text{accordingly for \m{<,>,\leq,=}}
  \end{align}
  \label{eq:Dformula}
  \end{subequations}
We abbreviate \m{\der{\inv}} by \m{\D{\inv}} and define \m{\subst[\D{\inv}]{\D{x}}{\theta}} as the result of substituting $\theta$ for $\D{x}$ in \m{\D{\inv}}, which is a Lie-type operator \cite{Lie93}.
\end{definition}
The conditions~\rref{eq:Dterm} define a derivation operator on terms that~\rref{eq:Dformula} lifts conjunctively to logical formulas.
It is important for the soundness of \irref{diffind} to define \m{\D{(F\lor G)}} as \m{\D{F} \land \D{G}}, because both subformulas need to satisfy the induction step, it is not enough if $F$ satisfies the induction step $\D{F}$ and $G$ holds initially; see \cite{DBLP:journals/logcom/Platzer10,Platzer10} for details and alternatives. 
We assume for simplicity that formulas use dualities like \m{\lnot(a\geq b) \mequiv a<b} to avoid negations; see \cite{DBLP:journals/logcom/Platzer10,Platzer10} for a discussion of this and the $\neq$ operator.
For a discussion why this definition of differential invariants gives a sound approach and many other attempts would be unsound, we refer to previous work \cite{DBLP:journals/logcom/Platzer10,Platzer10}.
We also refer to previous work \cite{Platzer10} for discussions about which weaker conditions like \m{\der{a>b} \,\mequiv\, \der{a}\geq\der{b}} are sound.
It is crucial, however, to realize that it would generally be unsound to assume $\inv$ or the boundary of $\inv$ in the premise.
Otherwise, we could draw the counterfactual conclusion that \m{-(x-y)^2\geq0} is an invariant of \m{\hevolve{\D{x}=1\syssep\D{y}=y}}.
See previous work \cite{DBLP:journals/logcom/Platzer10,Platzer10,DBLP:journals/lmcs/Platzer12} for an explanation under which circumstances this assumption would be sound.

\begin{wrapfigure}{r}{2.8cm}
  \vspace*{-\baselineskip}
  \includegraphics[width=2.6cm]{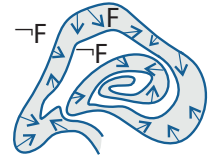} 
  \caption{Differential invariant $F$ for safety}
  \label{fig:diffind}
\end{wrapfigure}
The basic idea behind rule \irref{diffind} is that the premise of \irref{diffind} shows that the total derivative~\m{\D{\inv}} holds within evolution domain~$\ivr$ when substituting the differential equations \m{\D{x}=\theta} into~$\D{\inv}$.
If~$\inv$ holds initially (antecedent of conclusion), then~$\inv$ itself always stays true (succedent of conclusion).
Intuitively, the premise gives a condition showing that, within~$\ivr$, the total derivative~$\D{\inv}$ along the differential constraints is pointing inwards or transversally to~$\inv$ but never outwards to~$\lnot \inv$; see \rref{fig:diffind} for an illustration.
Hence, if we start in~$\inv$ and, as indicated by~$\D{\inv}$, the local dynamics never points outside~$\inv$, then the system always stays in~$\inv$ when following the dynamics.
Observe that, unlike \m{\D{\inv}}, the premise of \irref{diffind} is a well-formed formula, because all differential expressions are replaced by non-differential terms when forming~\m{\subst[\D{\inv}]{\D{x}}{\theta}}.
It is possible and insightful to give a meaning to the differential formula \m{\D{\inv}} itself in differential states \cite{DBLP:journals/logcom/Platzer10}.
Crucial for soundness is the result that the valuation of syntactic derivatives along differential equations coincides with analytic differentiation.
This derivation lemma plays a role similar to the substitution lemma in classical logic, but for derivations and differential equations.
For simplicity, we report a derivation lemma that already combines differential substitution \cite[Lemma 2]{DBLP:journals/logcom/Platzer10} with the derivation lemma \cite[Lemma 1]{DBLP:journals/logcom/Platzer10}.
These results form the basis for more general differential transformations \cite[Sect. 3.5]{Platzer10}.

{\newcommand{\crf}{c}%
\begin{lemma}[{Derivation lemma \cite{DBLP:journals/logcom/Platzer10}}] \label{lem:derivationLemma}
  \newcommand{\Iff}{\iconcat[state=\varphi(t)]{\I}}%
  \newcommand{\Ifz}{\iconcat[state=\varphi(\zeta)]{\I}}%
  Let~$\hevolvein{\D{x}=\theta}{\ivr}$ be a differential equation with evolution domain constraint~$\ivr$ and
  let~$\iget[flow]{\If}:[0,r]\to(V\to\reals)$ be a corresponding solution of duration~$r>0$, where $V$ is the set of variables.
  Then for all terms~$\crf$ and all~$\zeta\in[0,r]:$
  \begin{displaymath}
    \D[t]{\,{\ivaluation{\Iff}{\crf}}}(\zeta) = \ivaluation{\Ifz}{\subst[\D{\crf}]{\D{x}}{\theta}}
    \enspace.
  \end{displaymath}
  In particular,~\m{\ivaluation{\Iff}{\crf}} is continuously differentiable.
\end{lemma}
}
\begin{example}[Rotational dynamics] \label{ex:rotational}
The rotational dynamics \m{\hevolve{\D{x}=y\syssep\D{y}=-x}} is complicated to the extent that the solution involves trigonometric functions, which are generally outside decidable classes of arithmetic (see proof of \rref{thm:dL-incomplete}).
Yet, we can easily prove properties about the solution using \irref{diffind} and decidable polynomial arithmetic.
As a simple example, we can prove that \m{x^2+y^2\geq p^2} is a differential invariant of the dynamics using the following \dL proof:
\renewcommand{\arraystretch}{1.2}%
\begin{sequentdeduction}[array]
 \linfer[diffind]
 {\linfer
   {\linfer[qear]
     {\lclose}
     {\lsequent{}{2xy+2y(-x)\geq0}}
   }
   {\lsequent{}{\subst[(2x\D{x}+2y\D{y}\geq0)]{\D{x}}{y}\subst[\,]{\D{y}}{-x}}}
 }
 {\lsequent{x^2+y^2\geq p^2}{\dbox{\hevolve{\D{x}=y\syssep\D{y}=-x}}{x^2+y^2\geq p^2}}}
\end{sequentdeduction}
\end{example}
\begin{example}[Quartic dynamics] \label{ex:quartic}
The following simple \dL proof uses \irref{diffind} to prove an invariant of a quartic dynamics.
\renewcommand{\arraystretch}{1.6}%
  \begin{sequentdeduction}[array]
    \linfer[diffind]
    {\linfer
      {\linfer[qear]
        {\lclose}
        {\lsequent{a\geq0}{2x^2((x-3)^4+a)\geq0}}
      }
      {\lsequent{a\geq0}{\subst[(2x^2\D{x}\geq0)]{\D{x}}{(x-3)^4+a}}}
    }
    {\lsequent{x^3\geq-1}{\dbox{\hevolvein{\D{x}=(x-3)^4+a}{a\geq0}}{x^3\geq-1}}}
  \end{sequentdeduction}
  Observe that rule \irref{diffind} directly makes the evolution domain constraint \m{a\geq0} available as an assumption in the premise, because the continuous evolution is never allowed to leave it.
  This is useful if we have a strong evolution domain constraint or can make it strong during the proof, which is what we consider in \rref{sec:diffcut}.
\end{example}

\begin{counterexample}[Negative equations] \label{cex:diffind-neq}
  It is crucial for soundness that we do not define \m{\der{a\neq b}} to be \m{\der{a}\neq\der{b}}.
  Otherwise we could draw the wrong conclusion that \m{x\neq0} is an invariant of \m{\hevolve{\D{x}=1}}
  from the fact that \m{\der{x}\neq0}, i.e., \m{1\neq0}.
  This would be counterfactual, because variable~$x$ can reach \m{x=0} without its derivative ever being~$0$.
  In fact, $x$ can only reach \m{x=0} from an initial state \m{x\neq0} if its derivative is nonzero at some point.
  Instead, we could define \m{\der{a\neq b} \,\mequiv\, \der{a}=\der{b}} if needed.
  Intuitively, if $a$ and $b$ have the same derivative, then, $a\neq b$ is an invariant if it holds initially.
  This definition is already included, because we can equivalently rewrite \m{a\neq b} to \m{a>b \lor a<b} and use the weaker condition \m{\der{a>b} \,\mequiv\, \der{a}\geq\der{b}} to obtain \m{\der{a>b \lor a<b} \,\mequiv\, \der{a}\geq \der{b} \land \der{a}\leq \der{b}},
  which is equivalent to \m{\der{a}=\der{b}}.
  Sometimes, it can be more efficient to split the reasoning into two cases, the case where $a>b$ and the case where $a<b$ and consider them separately to prove that $a>b$ is an invariant if it was true initially and that $a<b$ is an invariant if that was true initially, instead of proving the stronger condition \m{\der{a\neq b} \,\mequiv\, \der{a}=\der{b}}.
\end{counterexample}

More advanced uses of differential invariants can be found in previous work \cite{DBLP:journals/logcom/Platzer10,DBLP:conf/cav/PlatzerC08,Platzer10,DBLP:conf/fm/PlatzerC09,DBLP:conf/icfem/PlatzerQ09,DBLP:journals/lmcs/Platzer12}.
Differential dynamic logic proofs with differential invariants have been instrumental in enabling the verification of more complicated hybrid systems, including separation properties in complex curved flight collision avoidance maneuvers for air traffic control \cite{Platzer10,DBLP:conf/fm/PlatzerC09}, advanced safety, reactivity and controllability properties of train control systems with disturbance and PI controllers \cite{Platzer10,DBLP:conf/icfem/PlatzerQ09}, and properties of electrical circuits \cite{Platzer10}.
Differential invariants are also the proof technique of choice for differential inequalities, differential-algebraic equations, and differential equations with disturbances \cite{DBLP:journals/logcom/Platzer10,Platzer10}.

Differential invariants enjoy closure properties, e.g.:
\begin{lemma}[Closure properties of differential invariants \cite{DBLP:journals/logcom/Platzer10}]
  Differential invariants are closed under conjunction, differentiation, and propositional equivalences.
\end{lemma}

\begin{wrapfigure}{r}{0.5\textwidth}
  \includegraphics[width=0.5\textwidth]{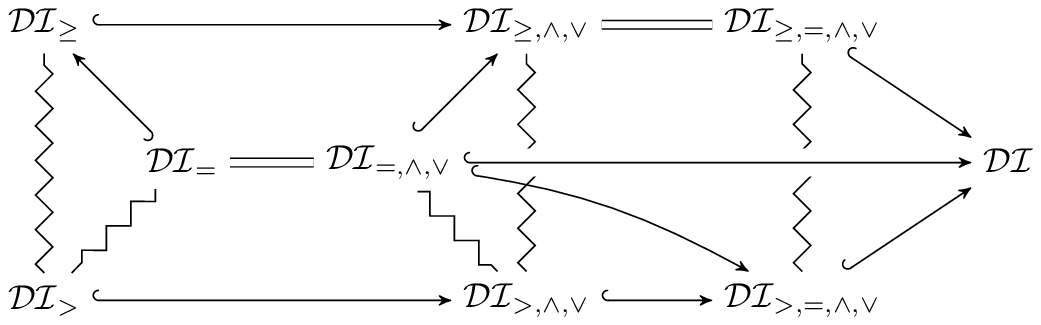}
  \caption{Differential invariance chart}
  \label{fig:diffind-chart}
\end{wrapfigure}
We refer to previous work \cite{DBLP:journals/lmcs/Platzer12,DBLP:journals/logcom/Platzer10} for more details on the structure of differential invariants and a complete investigation of the relative deductive power of several classes of differential invariants; see \rref{fig:diffind-chart} for an overview of classes of differential invariants restricted to the operators as indicated, where strict inclusions of the deductive power are indicated by $\hookrightarrow$, equivalences of deductive power are indicated by $=$, and incomparable deductive powers are indicated by \tikz{\draw[decorate,decoration=zigzag] (0,0) -- + (0.6,0);}.

We also refer to previous work \cite{DBLP:journals/logcom/Platzer10,Platzer10} for the technique of \emph{differential axiomatization}, which is useful for transforming sophisticated non-polynomial differential equations into polynomial differential equations by introducing new variables.
This is beneficial because, even though the solutions of the resulting polynomial differential equations are still equally complicated, we never need the solutions when working with differential invariants.
Differential invariants depend on the right-hand side of the differential equations, which is then polynomial and, thus, leads to decidable arithmetic.

\subsection{Differential Variants} \label{sec:difffin}

\newcommand{\wnot}[1]{{\sim}#1}%
\newcommand{\crf}{c}%
\emph{Differential variants} \cite{DBLP:journals/logcom/Platzer10} use ideas similar to those behind differential invariants, except that they use progress arguments so that differential variants can be used to prove formulas of the form \m{\ddiamond{\hevolvein{\D{x}=\genDE{x}}{\ivr}}{\inv}}.
That is, differential variants prove that the system can make progress along \m{\D{x}=\genDE{x}} to finally reach $\inv$ without having left~$\ivr$ before; see the right side of \rref{fig:diffind} for an illustration.
\begin{center}
\begin{calculus}
  \cinferenceRule[difffin|DV]{differential variant}
  {\linferenceRule[sequent]
        {\lsequent%
         {} {\lexists{\varepsilon{>}0}{\lforall{x}{(\lnot F\land\ivr \limply
            \subst[(\D{F}\geq\varepsilon)]{\D{x}}{\theta})}}}}
        {\lsequent{\dbox{\hevolvein{\D{x}=\genDE{x}}{\wnot{F}}}{\ivr}}
          {
            \ddiamond{\hevolvein{\D{x}=\genDE{x}}{\ivr}}{F}}}
  }{}
\end{calculus}
\end{center}

\begin{wrapfigure}{r}{2.8cm}
  \includegraphics[width=2.6cm]{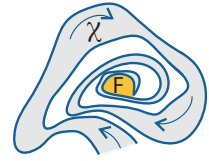} 
  \caption{Differential variant for progress}
  \label{fig:difffin}
\end{wrapfigure}
In rule \irref{difffin}, \m{\D{F}\geq\varepsilon} is a mnemonic notation for replacing all occurrences of inequalities~\m{a\geq b} in~$\D{F}$ with~\m{a\geq b+\varepsilon} and~\m{a>b} by~\m{a>b+\varepsilon}\internal{a\geq b+\varepsilon would be sufficient in both cases}
(accordingly for~\m{\leq,>,<}).
Intuitively, the premise expresses that, wherever~$\ivr$ holds but~$F$ does not yet hold, the total derivative is pointing towards~$F$; see right side of \rref{fig:difffin}.
Especially,~\m{\D{F}\geq\varepsilon} guarantees a minimum progress rate of~$\varepsilon$ towards~$F$ along the dynamics.
To further ensure that the continuous evolution towards~$F$ remains within~$\ivr$, the antecedent of the conclusion shows that~$\ivr$ holds \emph{until}~$F$ is attained, which can again be proven using \irref{diffind}.
Overall, the premise of rule \irref{difffin} shows that the dynamics makes progress (at least some~$\varepsilon$) toward~$F$, and the antecedent shows that the dynamics does not leave the evolution domain restriction~$\ivr$ on the way to~$F$.
In this context,~\m{\wnot{F}}\index{__~F@${\sim}F$} is a shorthand notation for \emph{weak negation}\index{weak~negation}, i.e., the operation that behaves like~$\lnot$, except that~\m{\wnot{(a\geq b)} \mequiv a\leq b} and~\m{\wnot{(a>b)} \mequiv a\leq b}.
Unlike negation $\lnot F$, weak negation $\wnot{F}$ retains the boundary of~$F$, which is required in \irref{difffin} as~$\ivr$ needs to continue to hold (including the boundary of~$F$) until~$F$ is reached.
Especially, for rule \irref{difffin}, invariant~$\ivr$ is not required to hold after~$F$ has been reached successfully.
The operations~\m{\D{F}\geq\varepsilon} and~$\wnot{F}$ are defined accordingly for other inequalities (in rule~\irref{difffin}, we do not permit~$F$ to contain equalities, because they could lead to unsoundness).
We refer to previous work \cite{DBLP:journals/logcom/Platzer10,Platzer10} for details.
Note that the order of quantifiers in \irref{difffin} is crucial for soundness \cite{DBLP:journals/logcom/Platzer10,Platzer10} to avoid Zeno progress that never reaches $F$ despite always getting closer.

\begin{example}[Progress discovery]
  Consider the simple property \m{\ddiamond{\hevolve{\D{x}=a}}{x\geq b}}, i.e., that we can finally reach region \m{x\geq b}, when following the dynamics \m{\hevolve{\D{x}=a}} long enough.
  We analyze this \dL formula in the following \dL proof:
     \renewcommand{\arraystretch}{1.2}%
  \begin{sequentdeduction}[array]
  \linfer[difffin]
    {\linfer
      {\linfer[qeer]
        {\lsequent{} {a>0}}
        {\lsequent{} {\lexists{\varepsilon{>}0}{\lforall{x}{(x\leq b \limply a\geq\varepsilon)}}}}
      }
      {\lsequent{} {\lexists{\varepsilon{>}0}{\lforall{x}{(x\leq b \limply \subst[(\D{x}\geq\varepsilon)]{\D{x}}{a})}}}}
    }
    {\lsequent{} {\ddiamond{\hevolve{\D{x}=a}}{x\geq b}}}
  \end{sequentdeduction}
  As the \dL proof reveals, the \dL formula is valid if \m{a>0}. This makes sense, because the system dynamics is then evolving towards \m{x\geq b}; otherwise it is evolving away from \m{x\geq b} (if \m{a<0}) or is constant (\m{a=0}).
  This proof constructs the answer that \m{a>0} is the required condition, which illustrates how constraints on parameters can be found by \dL proofs.
  For the above proof, we do not need to solve the differential equations.
  Solving the differential equation would be trivial here for a constant~$a$, but is more involved when~$a$ is an arbitrary term.
  With \irref{difffin}, we just form the total differential of
  \m{F \mequiv x\geq b},
  which gives
  \m{\D{F} = \D{x}\geq \D{b}}.
  When we substitute in the differential equations \m{\hevolve{\D{x}=a}}, we obtain
  \m{\subst[\D{F}]{\D{x}}{a} \mequiv a\geq0}.
  Consequently,
  \m{\subst[(\D{F}\geq\varepsilon)]{\D{x}}{a}}
  gives \m{a\geq\varepsilon}.
  If we prove \m{a\geq\varepsilon} holds for one common minimum progress \m{\varepsilon>0}, then the system makes some minimum progress towards the goal and will reach it in finite time.
  This even holds if we restrict the progress condition to all~$x$ that have not yet reached \m{x\geq b} or are on the boundary of \m{x\geq b}, which is the assumption \m{x\leq b} in the premise.
\end{example}

\subsection{Differential Cuts} \label{sec:diffcut}

\emph{Differential cuts} \cite{DBLP:journals/logcom/Platzer10} are another fundamental proof principle for differential equations. They can be used to strengthen assumptions in a sound way:

\begin{center}
\begin{calculus}
  \cinferenceRule[diffcut|DC]{differential cut}
  {\linferenceRule[sequent]
    {\lsequent{\inv}{\dbox{\hevolvein{\D{x}=\theta}{\ivr}}{C}}
    &&\lsequent{\inv}{\dbox{\hevolvein{\D{x}=\theta}{(\ivr\land C)}}{\inv}}}
    {\lsequent{\inv}{\dbox{\hevolvein{\D{x}=\theta}{\ivr}}{\inv}}}
  }{}
\end{calculus}
\end{center}
\noindent
The differential cut rule works like a cut, but for differential equations.
In the right premise, rule \irref{diffcut} restricts the system evolution to the subdomain \m{\ivr \land C} of $\ivr$, which restricts the system dynamics to a subdomain but this change is a pseudo-restriction, because the left premise proves that the extra restriction $C$ on the system evolution is an invariant anyhow (e.g. using rule \irref{diffind}).
Rule \irref{diffcut} is special in that it changes the dynamics of the system (it adds a constraint to the system evolution domain region that the resulting system is never allowed to leave by \rref{def:HP-transition}), but it is still sound, because this change does not reduce the reachable set.
The benefit of rule \irref{diffcut} is that $C$ will (soundly) be available as an extra assumption for all subsequent \irref{diffind} uses on the right premise of \irref{diffcut}.
The differential cut rule \irref{diffcut} can be used to strengthen the right premise with more and more auxiliary differential invariants $C$ that cut down the state space and will be available as extra assumptions to prove the right premise, once they have been proven to be differential invariants in the left premise.
\begin{example}[Multidimensional nonlinear dynamics]
\def\tmpa{y^5}%
\def\tmpDy{y^2}%
\begin{figure*}[tb]
\advance\leftskip-0.5cm
  \begin{minipage}{\textwidth}\footnotesize
\renewcommand{\linferPremissSeparation}{~~}%
\def\tmpa{y^5}%
\def\tmpDy{y^2}%
\renewcommand{\arraystretch}{1.6}%
\begin{sequentdeduction}[array]
  \linfer[diffcut]
    {\linfer[diffind]
    {\linfer
      {\linfer[qear]
        {\lclose}
        {\lsequent{}{5y^4y^2\geq0}}
      }
      {\lsequent{}{\subst[(5y^4\D{y}\geq0)]{\D{x}}{(x-3)^4+\tmpa}\subst[\,]{\D{y}}{\tmpDy}}}
    }
    {\lsequent{\tmpa\geq0}{\dbox{\hevolve{\D{x}=(x-3)^4+\tmpa\syssep\D{y}=\tmpDy}}{\tmpa\geq0}}}
    !
    \linfer[diffind]
    {\linfer
      {\linfer[qear]
        {\lclose}
        {\lsequent{\tmpa\geq0}{2x^2((x-3)^4+\tmpa)\geq0}}
      }
      {\lsequent{\tmpa\geq0}{\subst[(2x^2\D{x}\geq0)]{\D{x}}{(x-3)^4+\tmpa}\subst[\,]{\D{y}}{\tmpDy}}}
    }
    {\lsequent{x^3\geq-1}{\dbox{\hevolvein{\D{x}=(x-3)^4+\tmpa\syssep\D{y}=\tmpDy}{\tmpa\geq0}}{x^3\geq-1}}}
    }
    {\lsequent{x^3\geq-1\land\tmpa\geq0}{\dbox{\hevolve{\D{x}=(x-3)^4+\tmpa\syssep\D{y}=\tmpDy}}{x^3\geq-1}}}
\end{sequentdeduction}
\end{minipage}
  \caption{Differential cut proof for multidimensional nonlinear dynamics}
  \label{fig:diffcut-ex}
\end{figure*}%
The proof in \rref{ex:quartic} depends on evolution domain constraint $a\geq0$.
If we replace variable $a$ in \rref{ex:quartic} by a term, say $\tmpa$, then we first need to use \irref{diffcut} with the choice $C\mequiv\tmpa\geq0$ and prove that $\tmpa\geq0$ is indeed an invariant of the dynamics before we can use the proof from \rref{ex:quartic}, which depends on the evolution domain constraint $y^5\geq0$.
In \rref{ex:quartic}, this is trivial since $y$ has no differential equation, so that we assume \m{\D{y}=0} by convention.
Figure~\ref{fig:diffcut-ex} shows a \dL proof for a more interesting case, where $y$ changes during the continuous evolution.
The proof uses \irref{diffcut} once and \irref{diffind} twice.
\end{example}
Besides \irref{diffind}, the following simple proof rule for \emph{differential weakening} also benefits from strengthening evolution domain constraints via \irref{diffcut}:

\centerline{
\begin{calculus}
  \cinferenceRule[diffweak|DW]{differential weakening}
  {\linferenceRule[sequent]
    {\lsequent{\ivr}{\inv}}
    {\lsequent{\inv}{\dbox{\hevolvein{\D{x}=\theta}{\ivr}}{\inv}}}
  }{}
\end{calculus}
}%

\noindent
This rule is obviously sound, because the system \m{\hevolvein{\D{x}=\theta}{\ivr}}, by definition, is never allowed to leave the evolution domain constraint $\ivr$ anyhow, hence, if $\ivr$ implies $\inv$, then $\inv$ is an invariant, no matter what the dynamics \m{\hevolve{\D{x}=\theta}} does.
The simple rule \irref{diffweak} alone cannot prove very interesting properties, because it only works when $\ivr$ is very informative.
It can, however, be useful in combination with \irref{diffcut} with which we can soundly restrict the evolution domain constraint to subregions.

Using differential cuts repeatedly in a process called \emph{differential saturation} has turned out to be extremely useful in practice and even simplifies the invariant search, because it leads to several simpler invariants to find and prove instead of a single complex property \cite{DBLP:conf/cav/PlatzerC08,DBLP:journals/fmsd/PlatzerC09,Platzer10}.

Differential cuts do not only help in practice, but are a fundamental proof principle in theory.
We have refuted the \emph{differential cut elimination hypothesis} \cite{DBLP:journals/logcom/Platzer10}.
Differential cuts have a simple intuition.
Similar to a cut in first-order logic, they can be used to first prove a lemma and then use it.
By the seminal cut elimination theorem of Gentzen \cite{Gentzen35I,Gentzen35}, standard logical cuts do not change the deductive power and can be eliminated, even if that may have significant cost \cite{DBLP:journals/jphil/Booloes84}.
Unlike standard cuts, however, differential cuts actually increase the deductive power \cite{DBLP:journals/lmcs/Platzer12}.
There are properties of differential equations that can only be proven using differential cuts, not without them.
Hence, differential cuts are a fundamental proof principle for differential equations.
\begin{theorem}[Differential cut power \cite{DBLP:journals/lmcs/Platzer12}] \label{thm:diffcut-power}
  The deductive power with differential cuts (rule \irref{diffcut}) exceeds the deductive power without differential cuts.
\end{theorem}
We refer to previous work \cite{DBLP:journals/lmcs/Platzer12} for details on the differential cut elimination hypothesis \cite{DBLP:journals/logcom/Platzer10}, the proof of its refutation \cite{DBLP:journals/lmcs/Platzer12}, and a complete investigation of the relative deductive power of several classes of differential invariants.
}

\subsection{Differential Auxiliaries} \label{sec:diffaux}

It is well-known that auxiliary variables may be necessary to conduct proofs about conventional discrete programs.
We have studied auxiliary differential variables, and have shown that some differential equation systems can only be proven after introducing auxiliary differential variables into the dynamics \cite{DBLP:journals/lmcs/Platzer12}.
That is, the addition of auxiliary differential variables increases the deductive power.
This is captured in the following proof rule \emph{differential auxiliaries} (\irref{diffaux}) for introducing auxiliary differential variables \cite{DBLP:journals/lmcs/Platzer12}:
\begin{center}
\begin{calculus}
\cinferenceRule[diffaux|DA]{differential auxiliary variables}
{\linferenceRule[sequent]
  {\lsequent{}{\phi\lbisubjunct\lexists{y}{\psi}}
  &\lsequent{\psi} {\dbox{\hevolvein{\D{x}=\theta\syssep\D{y}=\vartheta}{\ivr}}{\psi}}}
  {\lsequent{\phi} {\dbox{\hevolvein{\D{x}=\theta}{\ivr}}{\phi}}}
}{}%
\end{calculus}
\end{center}
Rule \irref{diffaux} is applicable if $y$ is a new variable (or vector of variables) and the new differential equation \m{\hevolve{\D{y}=\vartheta}} has global solutions on $\ivr$ (e.g., because term $\vartheta$ satisfies a Lipschitz condition \cite[Proposition 10.VII]{Walter:ODE}, which is definable in first-order real arithmetic and thus decidable).
Without a condition like this, adding \m{\D{y}=\vartheta} could limit the duration of system evolutions incorrectly.
In fact, it would be sufficient for the domains of definition of the solutions of \m{\hevolve{\D{y}=\vartheta}} to be no shorter than those of $x$.

Intuitively, rule \irref{diffaux} can help proving properties, because it may be easier to characterize how $x$ changes in relation to  auxiliary variables $y$ that co-evolve with a suitable differential equation (\m{\hevolve{\D{y}=\vartheta}}).
We have proved that the addition of auxiliary differential variables increases the deductive power, even in the presence of differential cuts \cite{DBLP:journals/lmcs/Platzer12}.
That is, there are system properties that can only be proven using auxiliary differential variables in the dynamics.
Hence, auxiliary differential variables are also a fundamental proof principle for differential equations.
\begin{theorem}[Auxiliary differential variable power \cite{DBLP:journals/lmcs/Platzer12}] \label{thm:diffaux-power}
  The deductive power with auxiliary differential variables (rule \irref{diffaux}) exceeds the deductive power without auxiliary differential variables even in the presence of differential cuts.
\end{theorem}

\subsection{Implementation and Applications} \label{sec:KeYmaera}

Differential dynamic logic \cite{DBLP:conf/tableaux/Platzer07,DBLP:journals/jar/Platzer08,DBLP:conf/lics/Platzer12b} and its proof calculus \cite{DBLP:conf/tableaux/Platzer07,DBLP:journals/jar/Platzer08}, including differential invariants, differential variants, and differential cuts \cite{DBLP:journals/logcom/Platzer10} have been implemented in the automatic and interactive theorem prover \KeYmaera \cite{DBLP:conf/cade/PlatzerQ08},\footnote{Available at \url{http://symbolaris.com/info/KeYmaera.html}} which is based on Ke\kern-0.1emY \cite{KeYBook2007}.
The name \KeYmaera is a homophone to Chim{\ae}ra, the hybrid animal from ancient Greek mythology.
\KeYmaera implements a sequent calculus version \cite{DBLP:journals/jar/Platzer08} of the axiomatization in \rref{fig:dL}, because the sequent calculus is more suitable for automation.
Differential dynamic logic forms the basis for an automatic proof search procedure searching for invariants and differential invariants \cite{DBLP:conf/cav/PlatzerC08,DBLP:journals/fmsd/PlatzerC09,Platzer10} that has been implemented in \KeYmaera.

Differential dynamic logic and \KeYmaera have been used successfully for verifying system-level properties of local lane controllers for highway car traffic \cite{DBLP:conf/fm/LoosPN11}, car controllers for intersections \cite{DBLP:conf/itsc/LoosP11}, intelligent speed adaptation for variable speed limit control and incident management by traffic centers on highways \cite{DBLP:conf/iccps/MitschLP12}, CICAS-SLTA left-turn assist controllers for cars at intersections \cite{DBLP:conf/acc/ArechigaLPK12}, flyable roundabout collision avoidance maneuvers for aircraft \cite{DBLP:conf/fm/PlatzerC09}, the cooperation protocols of the European Train Control System ETCS \cite{DBLP:conf/icfem/PlatzerQ09}, and analog circuits \cite{Platzer10}.
\KeYmaera has been used to prove safety requirements of a distributed elevator controller, medical robotic surgery systems,  robotic factories, and to study biological models. 
Properties proved about these systems using \dL include safety, controllability, reactivity, liveness, and characterization properties.
More details about \dL and some of its applications are described in a book \cite{Platzer10} about the theory, practice, and applications of \dL and its extensions \DAL for differential-algebraic hybrid systems and \dTL for temporal properties.

\section{Quantified Differential Dynamic Logic for Distributed Hybrid Systems} \label{ch:QdL}

In this section, we study \emph{quantified differential dynamic logic} \QdL \cite{DBLP:conf/csl/Platzer10,DBLP:journals/lmcs/Platzer12b}, the \emph{logic of distributed hybrid systems}, i.e., systems that combine the dynamics of distributed systems with the discrete and continuous dynamics of hybrid systems.

Not all cyber-physical systems and certainly not all dynamical systems can be modeled faithfully as hybrid systems.
Cyber-physical systems typically combine \emph{communication, computation, and control}.
They may even form dynamic distributed networks, where neither structure nor dimension stay the same while the system follows hybrid dynamics.

Combining computation and control leads to \emph{hybrid systems}, whose behavior involves both discrete and continuous dynamics originating, e.g., from discrete control decisions and differential equations of movement (\rref{ch:dL}).
Combining communication and computation leads to \emph{distributed systems} \cite{Lynch,DBLP:conf/concur/AttieL01,AptdeBoerOlderog10}, whose dynamics are discrete transitions of system parts that communicate with each other.
They may form \emph{dynamic distributed systems}, where the structure of the system is not fixed but evolves over time and agents may appear or disappear during the system evolution.

Combinations of all three aspects (communication, computation, and control) are used in sophisticated applications, e.g., cooperative distributed car control \cite{HsuEskafiSachsVaraiya1991} and decentralized aircraft control \cite{PallottinoSFB06}.
Neither the structure nor dimension of the system stay the same, because new cars can appear on the street or leave it; see \rref{fig:distributed-car-control-new} on p.~\pageref{fig:distributed-car-control-new}.
These systems are \emph{(dynamic) distributed hybrid systems} \cite{DBLP:conf/hybrid/DeshpandeGV96,DBLP:conf/hybrid/Rounds04,DBLP:conf/hybrid/KratzSPL06,DBLP:conf/hybrid/MeseguerS06,DBLP:journals/taas/GilbertLMN09,DBLP:conf/csl/Platzer10,DBLP:journals/lmcs/Platzer12b}.
More generally, distributed hybrid systems are multi-agent hybrid systems that interact through remote communication or physical interaction.
They cannot be considered just as a distributed system (because, e.g., the continuous evolution of positions and velocities matters crucially for collision freedom in car control) nor just as a hybrid system (because the evolving system structure and appearance of new agents or structural changes in the system can make an otherwise collision-free system unsafe).
It is generally not possible to split the analysis of distributed hybrid systems soundly into an analysis of a distributed system (without continuous movement) and an analysis of a hybrid system (without structural changes or appearance), because all kinds of dynamics interact.
Just like hybrid systems that are very difficult to analyze from a purely discrete or a purely continuous perspective \cite{DBLP:conf/lics/Henzinger96}.
See previous work \cite{DBLP:conf/lics/Platzer12b} for a complete discussion of the relationship between discrete and continuous dynamics in hybrid systems.

As a formal logic for specifying and verifying correctness properties of distributed hybrid systems, we have introduced \emph{quantified differential dynamic logic} (\QdL) \cite{DBLP:conf/csl/Platzer10,DBLP:journals/lmcs/Platzer12b}.
\QdL extends \dL to distributed hybrid systems.
\QdL combines dynamic logic for reasoning about all (\m{\dbox{\alpha}{\phi}}) or some (\m{\ddiamond{\alpha}{\phi}}) system runs of a system~$\alpha$ \cite{Harel_et_al_2000} with many-sorted first-order logic for reasoning about all (\m{\lforall[C]{i}{\phi}}) or some (\m{\lexists[C]{i}{\phi}}) objects of a sort $C$, e.g., the sort of all cars.

The most important defining characteristic of \QdL is that~$\alpha$ can be a distributed hybrid system, because the \QdL system model of \emph{quantified hybrid programs} (\QHP) supports quantified operations that affect \emph{all} objects of a sort $C$ at once.
If $C$ is the sort of cars, the quantified assignment \m{\pupdate{\lforall[C]{i}{\pumod{a(i)}{a(i)+1}}}} increases the respective accelerations $a(i)$ of \emph{all cars} $i$ at once by a single instantaneous discrete jump.
It can be used to model simultaneous discrete changes in multiple agents at once.
Discrete changes where only some of the cars change their acceleration, others do not, are easy to model with quantified assignments by masking.
The quantified differential equation \m{\hevolve{\lforall[C]{i}{\D{v(i)}=a(i)}}} represents a continuous evolution of the respective velocities $v(i)$ of \emph{all cars}~$i$ at the same time according to their acceleration by their respective differential equations \m{\hevolve{\D{v(i)}=a(i)}}.
Again, continuous evolutions where only some of the cars evolve, others remain stopped, are easy to model with quantified differential equations by masking.
These quantified assignments and quantified differential equation systems of \QHPs are crucial for representing distributed hybrid systems where an unbounded number of objects co-evolve simultaneously, because no finite set of classical assignments and classical differential equations could represent that.
Note that, because of the close semantical relationship, we use the same quantifier notation \m{\pupdate{\lforall[C]{i}}}for quantified operations in programs and for quantifiers in logical formulas, instead of a separate notation $\Pi_{i:C}$ for parallel products in programs.

Interaction by communication can be modeled by (possibly quantified) discrete assignments to share data between agents $i$ and $j$ in \QHPs.
Physical interaction, instead, may be modeled either by (possibly quantified) discrete assignments when an agent $i$ activates a response in agent $j$ by an instantaneous discrete action (e.g., pushing a physical button) or by a (possibly quantified) differential equation involving multiple agents $i$ and $j$ when they come into physical contact and act jointly over a (nonzero) period of time (e.g., both agents jointly lifting and pulling on a rigid object).
Observe that the cyber structure of the system reconfigures dynamically when discrete communication topologies change, whereas the physical structure reconfigures dynamically when agents engage in physical contact.
\QHPs for the latter case may involve structural changes in the quantified differential equation.

We model the appearance of new participants in the distributed hybrid system, e.g., new cars entering the road, by a program \m{\pumod{\onew{}}{\pnew{C}}}. It creates a new object of type $C$, thereby extending the range of all subsequent quantified assignments or quantified differential equations ranging over created objects of type $C$.
With quantifiers and function terms, $\pnew{}$ can be defined and handled in an entirely modular way; see previous work \cite{DBLP:conf/csl/Platzer10,DBLP:journals/lmcs/Platzer12b} for details.
Overall, \dL, which we considered in \rref{ch:dL}, is for finite-dimensional hybrid systems, but \QdL can handle evolving or  infinite-dimensional distributed hybrid systems.
\QdL and \QHPs provide first-order state variables and quantifiers (also in the dynamics) over the arguments of first-order function symbols, which is the fundamental enabling technique for distributed hybrid systems.
This difference of \QdL compared to \dL is as fundamental as that of first-order logic compared to propositional logic, except that it also affects the dynamics not just the formulas.
The logic \dL and HPs only support primitive variables $x,v,a$ of type $\reals$, whereas \QdL supports first-order function symbols $x(i),v(i),a(i),d(i,j)$ and quantifiers, e.g., over $i$ and $j$ in the dynamics, so that an unbounded (instead of a statically fixed finite) number of agents can be described by the dynamics.

The model of \QHPs is of independent interest as a formal model for distributed hybrid systems.
Inside a \QHP, logical formulas can occur in state tests for conditional execution.
We thus explain logical formulas, terms, and sorts first.
Conversely, however, a \QHP~$\alpha$ occurs inside the modalities (\m{\dbox{\alpha}{}} and \m{\ddiamond{\alpha}{}}) of \QdL formulas, which state properties of the behavior of~$\alpha$.
Hence, \QHPs may occur inside \QdL formulas yet formulas may occur inside \QHPs.
The subsequent definitions of \QdL and \QHP are thus to be understood by simultaneous induction.

We first explain the logical formulas that \QdL provides for specification and verification (\rref{sec:QdL-formula}) and then explain the system model of quantified hybrid programs that \QdL provides for modeling distributed hybrid systems (\rref{sec:QHP}).
We define the semantics (\rref{sec:QdL-semantics}) and then explain reasoning principles, axioms, and proof rules for verifying \QdL formulas (\rref{sec:QdL-calculus}).
We then show soundness and relative completeness theorems (\rref{sec:QdL-complete}) and investigate stronger proof rules for quantified differential equations (\rref{sec:Qdiffind}).
Finally, we briefly discuss an implementation in the theorem prover \KeYmaeraD and applications (\rref{sec:KeYmaeraD}).

\subsection{\QdLbf Formulas} \label{sec:QdL-formula}
We have introduced quantified differential dynamic logic (\QdL) \cite{DBLP:conf/csl/Platzer10,DBLP:journals/lmcs/Platzer12b}, which is the first formal logic for specifying and verifying correctness properties of distributed hybrid systems.
\QdL is a combination of many-sorted first-order logic with dynamic logic, generalized to a system model (\QHPs) for distributed hybrid systems.

\subsubsection{Sorts}
\QdL supports a (finite) number of object sorts, e.g., the sort of all cars and that of all aircraft.
For continuous quantities of distributed hybrid systems like positions or velocities, we add the sort $\reals$ of real numbers.
It would be easy to add subtyping of sorts;  see previous work \cite{DBLP:conf/cade/BeckertP06} for details.
We refrain from doing so, because that just obscures the logical essence of our approach.

The primary purpose of the sorts is to distinguish different kinds of objects in multi-agent hybrid systems in which different kinds of agents occur, e.g., cars of sort $C$, traffic lights of sort $T$, lanes of sort $L$, and aircraft of sort $A$.

\subsubsection{Terms}
\QdL terms are built from a set of (sorted) function and variable symbols as in many-sorted first-order logic.
Each function symbol~$f$ has a fixed type \m{C_1\times\dots\times C_n\to D} for some $n\in\naturals$ and some sorts \m{D,C_1,\dots,C_n} such that $f$ only accepts argument terms~\m{\theta_1,\dots,\theta_n} of the respective sorts \m{C_1,\dots,C_n} and then \m{f(\theta_1,\dots,\theta_n)} is a term of sort $D$.
We use these function symbols to represent the state of the system or other parameters.
In a car control scenario like that in \rref{fig:distributed-car-control-new}, for example, we could use function symbol $x$ to represent the positions of cars, i.e., the term $x(i)$ could represent the position of car $i$ and $x(j)$ the position of car $j$.
Similarly, the term $v(i)$ could represent the velocity of car $i$ and $a(i)$ its acceleration.
These terms have sort $\reals$, whereas a term $l(i)$ that represents the car in front of car $i$ has sort $C$.

Unlike in first-order logic, the interpretation of function symbols can change when transitioning from one state to the other while following the dynamics of a distributed hybrid system.
The value of position $x(i)$ will change over time as car $i$ drives down the street.
The value of $x(i)$ also changes if the argument term $i$ changes its value and now refers to a different car than before.
Objects may appear or disappear as the distributed hybrid system evolves.
We use function symbol $\laexisting{\cdot}$ to distinguish between objects $i$ that actually exist (\m{\laexisting{i}=1}) and those that have not been created yet or exist no longer (\m{\laexisting{i}=0}), depending on the value of $\laexisting{i}$, which may change its interpretation from state to state.
We use $0,1,+,-,\cdot$ with the usual notation and fixed semantics for real arithmetic.
For \m{n\geq0} we abbreviate \m{f(s_1,\dots,s_n)} by \m{f(\vec{s})} using vectorial notation and we use \m{\vec{s}=\vec{t}} for component-wise equality.

\subsubsection{Formulas}
The formulas of \QdL are defined as in first-order dynamic logic plus many-sorted first-order logic
\begin{definition}[\QdL formula]
The formulas of \QdL are defined by the following grammar ($\phi,\psi$ are formulas, $\theta_1,\theta_2$ are terms of the same sort, $i$ is a variable of sort $C$, and $\alpha$ is a \QHP as defined in \rref{sec:QHP}):
\begin{equation*}
  \begin{array}{@{}l@{}}
  \phi,\psi ~\bebecomes~ %
  \theta_1=\theta_2 \alternative
  \theta_1\geq\theta_2 \alternative
  \lnot \phi \alternative
  \phi \land \psi \alternative
  \lforall[C]{i}{\phi} \alternative 
  \dbox{\alpha}{\phi} %
  \end{array}
\end{equation*}
\end{definition}
We use standard abbreviations to define $\leq,>,<,\lor,\limply,\lexists$.
The operator $\ddiamond{\alpha}{}$ dual to $\dbox{\alpha}{}$ is again defined by \m{\ddiamond{\alpha}{\phi} \mequiv \lnot\dbox{\alpha}{\lnot\phi}}.
Similarly, \m{\lexists[C]{i}{\phi} \mequiv \lnot\lforall[C]{i}{\lnot\phi}}.
Sorts \m{C\neq\reals} have no ordering and only \m{\theta_1=\theta_2} is allowed, not \m{\theta_1\geq\theta_2}.
For sort $\reals$, we abbreviate \m{\lforall[\reals]{x}{\phi}} by \m{\lforall{x}{\phi}} and \m{\lexists[\reals]{x}{\phi}} by \m{\lexists{x}{\phi}}.
All \QdL formulas and terms have to be well-typed.
For instance, \m{x(i)=l(i)} is no formula if $x$ has type \m{C\to\reals} and $l$ has type \m{C\to C} for a sort \m{C\neq\reals} or if $i$ has a sort \m{D\neq C}.
\QdL formula \m{\dbox{\alpha}{\phi}} expresses that \emph{all states} reachable by \QHP~$\alpha$ satisfy formula~$\phi$. Likewise, \m{\ddiamond{\alpha}{\phi}} expresses
that \emph{there is at least one state} reachable by~$\alpha$ for
which~$\phi$ holds.

For short notation, we allow \emph{conditional terms} of the form \m{\piif{\phi}{\theta_1}{\theta_2}} (where $\theta_1$ and $\theta_2$ have the same sort).
This term evaluates to $\theta_1$ if the formula $\phi$ is true and to $\theta_2$ otherwise.
We generally consider formulas with conditional terms as abbreviations, e.g.,
\m{\mapply{\psi}{\piif{\phi}{\theta_1}{\theta_2}}} abbreviates
\m{(\phi \limply \mapply{\psi}{\theta_1})  \land  (\lnot\phi \limply \mapply{\psi}{\theta_2})}.
Conditional terms can be understood as an additional operator for terms and formulas as well.

  \newcommand{\DCCS}{\textit{DCCS}\xspace}
  \newcommand{\abrake}{b}%
  \newcommand{\amax}{A}%
  \newcommand{\aset}{a}%
  \newcommand{\cyct}{\varepsilon}%

\newcommand*{\oa}[2]{#1(#2)}%
\newcommand*{\dcseparate}[2]{\mathcal{M}(#1,#2)}%
\let\dcaccelseparate\dcseparate
\newcommand*{\dcseparatet}[2]{#1\neq #2 \limply \big((\dcseparatetf{#1}{#2})\lor(\dcseparatets{#1}{#2})\big)}%
\newcommand*{\dcseparatetf}[2]{%
   \oa{x}{#1}<\oa{x}{#2}\land\oa{v}{#1}\leq\oa{v}{#2}\land\oa{a}{#1}\leq\oa{a}{#2}}%
\newcommand*{\dcseparatets}[2]{%
  \oa{x}{#1}>\oa{x}{#2}\land\oa{v}{#1}\geq\oa{v}{#2}\land\oa{a}{#1}\geq\oa{a}{#2}}%
\newcommand{\dcinv}{\laforall[C]{i,j}{\dcseparate{i}{j}}}%
\newcommand{\dcnu}{\onew{}}%
\newcommand{\dcnup}{\pumod{\dcnu}{\pnew{C}}}%
\newcommand{\dcnusep}{\laforall[C]{i}{\dcseparate{i}{\dcnu}}}%
\newcommand{\dcevo}{\pevolve{\laforall[C]{i}{(\D[2]{\oa{x}{i}}=\oa{a}{i})}}}%
\newcommand{\dcevot}{\pevolve{\laforall[C]{i}{(\D{\oa{x}{i}}=\oa{v}{i} \syssep \D{\oa{v}{i}}=\oa{a}{i})}}}%

\newcommand{\dcsys}{\dcnup; \ptest{\dcnusep}; \dcevo}%

  \newcommand*{\dcaccel}[1]{\mathcal{A}(#1)}%
  \newcommand*{\dcaccelt}[1]{\umod{\oa{a}{#1}}{\piif{\SBforma{#1}}{\amax}{\,{-}\abrake}}}%
  \newcommand*{\dcaccelall}{\pupdate{\laforall[C]{i}{\dcaccel{i}}}}%
  \newcommand*{\dcaccelallt}{\pupdate{\laforall[C]{i}{\dcaccelt{i}}}}%
\newcommand*{\dcaccelseparatetf}[2]{%
   \oa{x}{#1}<\oa{x}{#2}\land\oa{v}{#1}^2<\oa{v}{#2}^2+2\abrake(\oa{x}{#2}-\oa{x}{#1})}%
\newcommand*{\dcaccelseparatets}[2]{%
  \oa{x}{#1}>\oa{x}{#2}\land\oa{v}{#2}^2<\oa{v}{#1}^2+2\abrake(\oa{x}{#1}-\oa{x}{#2})}%
  \newcommand*{\SBforma}[1]{\laforall[C]{j}{\SBform{#1}{j}}}%
  \newcommand{\SBform}[2]{\textit{far}(#1,#2)}
  \newcommand{\SBformt}[2]{\oa{x}{#2}>\oa{x}{#1} \limply 
    \oa{x}{#2}>\oa{x}{#1} + \frac{\oa{v}{#1}^2-\oa{v}{#2}^2}{2\abrake} + \left(\frac{\amax}{\abrake}+1\right)\left(\frac{\amax}{2}\cyct^2+\cyct\oa{v}{#1}\right)
  }%
\newcommand{\dcaccelevot}{\pupdate{\umod{\tau}{0}};~ \hevolvein{\laforall[C]{i}{(\D{\oa{x}{i}}=\oa{v}{i} \syssep \D{\oa{v}{i}}=\oa{a}{i}\syssep\D{\tau}=1}}{\oa{v}{i}\geq0 \land \tau\leq\cyct)}}%

\begin{example}[Distributed car control] \label{ex:QdL}
  \let\laforall\lforall%
If $i$ is a term of type $C$ (for cars), let $x(i)$ denote the position of car $i$, $v(i)$ its current velocity, and $a(i)$ its current acceleration; see \rref{fig:distributed-car-control-new} on p.\,\pageref{fig:distributed-car-control-new}.
A car control system is collision-free at a state if all cars are at different positions (\m{\lforall[C]{i{\neq}j}{\oa{x}{i}{\neq}\oa{x}{j}}}).
Without a quantifier we could not describe that all cars on a highway are in a collision-free state, because there is a large number of cars on the highway and we may not know how many.
The car control system is globally collision-free if it will always stay collision-free.
The following \QdL formula expresses that a distributed car control system \DCCS (we develop QHP models for \DCCS later) controls cars such that they are always collision-free:
\begin{equation}
  {(\dcinv)\,} \limply {\,\dbox{\DCCS}{~\laforall[C]{i{\neq}j}{\oa{x}{i}{\neq}\oa{x}{j}}}}
  \label{eq:distributed-car-control-new}
\end{equation}
It says that cars following the distributed hybrid systems dynamics of \DCCS are always collision-free (postcondition), provided that \DCCS starts in an initial state satisfying a formula \m{\dcseparate{i}{j}} for all cars $i,j$ (precondition).
In particular, the modality \m{\dbox{\DCCS}{}} expresses that all states reachable by following the distributed hybrid system \DCCS satisfy the postcondition.
The simple-most choice for the formula \m{\dcseparate{i}{j}} in the precondition is a formula that characterizes a simple compatibility condition: for different cars $i\neq j$, the car that is further down the road (i.e., with greater position) neither moves slower nor accelerates slower than the other car, i.e.
$\dcseparate{i}{j}$ is
\begin{equation}
\begin{aligned}
\dcseparate{i}{j} \mequiv
  i \neq j &\limply \big((\dcseparatetf{i}{j})\qquad
  \\&~\lor(\dcseparatets{i}{j})\big)%
\end{aligned}
  \label{eq:distributed-car-control-separate}
\end{equation}
Based on a generalization of the \dL formulas in \rref{ex:dL}, we can also choose the following more permissive formula for \m{\dcseparate{i}{j}}, which allows cars to drive more aggressively with different accelerations, if only the respective safety distances are compatible with the different velocities:
\begin{equation}
\begin{aligned}
\dcseparate{i}{j} \mequiv
  i\neq j \limply\,& \oa{v}{i}\geq0\land\oa{v}{j}\geq0 \,\land
  \\& \big((\dcaccelseparatetf{i}{j})
  \\&
  \lor(\dcaccelseparatets{i}{j})\big)%
\end{aligned}
  \label{eq:distributed-car-control-accel-separate}  
\end{equation}
This formula characterizes that, for different cars $i\neq j$, the car that is further down the road is at a sufficient distance to allow its follower car to adapt its velocity by braking to be no faster than the car that is already further down the road.
With this choice of \m{\dcseparate{i}{j}}, observe the relationship of \QdL formula \rref{eq:distributed-car-control-new} to \dL formula \rref{eq:car-single-essentials} in \rref{ex:dL}.
For a more detailed and exhaustive study of possible initial conditions and invariants for distributed car control, we refer to previous work \cite{DBLP:conf/csl/Platzer10,DBLP:journals/lmcs/Platzer12b,DBLP:conf/fm/LoosPN11}.
\end{example}

The reader should note that more sophisticated combinations of nested quantifiers and modalities like the ones we considered for \dL (e.g., \rref{ex:dL}) are possible with \QdL and its axiomatization in \rref{sec:QdL-calculus} as well.

\subsection{Quantified Hybrid Programs} \label{sec:QHP}

As a formal model for distributed hybrid systems, we have introduced \emph{quantified hybrid programs}~(\QHPs) \cite{DBLP:conf/csl/Platzer10,DBLP:journals/lmcs/Platzer12b}.
These are regular programs from dynamic logic \cite{Harel_et_al_2000} to which we add quantified assignments and quantified differential equation systems for \emph{distributed} hybrid dynamics.
From these quantified assignments and quantified differential equations, \QHPs are built like a Kleene algebra with tests \cite{DBLP:journals/toplas/Kozen97}.
\begin{definition}[Quantified hybrid program]
\QHPs are defined by the following grammar ($\alpha,\beta$ are \QHPs, $i$ a variable of sort $C$, $f$ is a function symbol, $\vec{s}$ is a vector of terms with sorts compatible to the arguments of $f$, $\theta$ is a term with sort compatible to the result of $f$, and $\ivr$ is a formula of many-sorted first-order logic):
\[
  \alpha,\beta ~\bebecomes~
  \pupdate{\lforall[C]{i}{\pumod{f(\vec{s})}{\theta}}}
  \alternative
  \ptest{\ivr}
  \alternative
  \hevolvein{\lforall[C]{i}{\D{f(\vec{s})}=\theta}}{\ivr}
  \alternative
  \alpha\cup\beta
  \alternative
  \alpha;\beta
  \alternative
  \prepeat{\alpha}
\]
\end{definition}
In order to simplify technical difficulties, we impose regularity assumptions on $f(\vec{s})$ in quantified assignments and quantified differential equations.
We assume $\vec{s}$ to be either a vector of length 0 or that the mapping from the quantified variable $i$ to $\vec{s}$ is \emph{injective}.
That is, each value of $\vec{s}$ can be exhibited by at most one choice of $i$.
A system is injective, e.g., when at least one component of $\vec{s}$ is the quantified variable $i$.
These assumptions can be relaxed, but are sufficient for our purposes; see \rref{sec:QdL-semantics} for a discussion of injectivity.
For quantified differential equations, we further assume that $f$ is an $\reals$-valued function symbol so that derivatives can be defined.

\subsubsection{Quantified State Change}
The effect of \dfn[assignment!quantified]{quantified assignment} 
\m{\pupdate{\lforall[C]{i}{\pumod{f(\vec{s})}{\theta}}}} is an instantaneous discrete jump assigning~$\theta$ to~$f(\vec{s})$ simultaneously for all objects $i$ of sort $C$.
Hence all $f(\vec{s})$ that are affected by \m{\pupdate{\lforall[C]{i}{\pumod{f(\vec{s})}{\theta}}}} will change their value to the respective~$\theta$ simultaneously for all choices of~$i$ in a single discrete instant of time.
Usually,~$i$ occurs in term~$\theta$, but does not have to.
The effect of \dfn[differential~equation!quantified]{quantified differential equation}
\m{\hevolvein{\lforall[C]{i}{\D{f(\vec{s})}=\theta}}{\ivr}} is a continuous evolution where, for all objects $i$ of sort $C$, all differential equations \m{\D{f(\vec{s})}=\theta} hold at the same time and formula~$\ivr$ holds throughout the evolution (the state always remains in the region described by~$\ivr$, i.e., the evolution stops at any arbitrary time before it leaves~$\ivr$).
Again,~$i$ usually occurs in term~$\theta$.
For the trivial evolution domain restriction~\m{\ivr\mequiv\ltrue}, which is always satisfied, we also write \m{\hevolve{\lforall[C]{i}{\D{f(\vec{s})}=\theta}}} instead of \m{\hevolvein{\lforall[C]{i}{\D{f(\vec{s})}=\theta}}{\ltrue}}.

The dynamics of \QHPs changes the interpretation of terms over time:
\m{\D{f(\vec{s})}} is intended to denote the derivative of the interpretation of the term \m{f(\vec{s})} over time during continuous evolution, not the derivative of \m{f(\vec{s})} by its argument $\vec{s}$.
For \m{\D{f(\vec{s})}} to be defined, we assume $f$ is an $\reals$-valued function symbol.
Although our approach can be extended, we assume that~$f$ does not occur in~$\vec{s}$.
The most common choice of $\vec{s}$ in quantified assignments and quantified differential equations is just $i$.
Other choices are possible for $\vec{s}$, e.g., \m{\vec{s}=(i,f(i))} in \m{\pupdate{\lforall[C]{i}{\pumod{d(i,f(i))}{\frac{1}{2}a(i)+\frac{1}{2}a(f(i))}}}}.
The latter \QHP could be used to model that, for each car $i$, the average acceleration of a car $i$ and its follower $f(i)$ is assigned to a data field $d(i,f(i))$ that car $i$ and its follower communicate to determine their safe distance.

Time itself is not special but implicit. If a clock variable $t$ is needed in a \QHP, it can be axiomatized by \m{\hevolve{\D{t}=1}}, which is equivalent to
\m{\hevolve{\lforall[C]{i}{\D{t}=1}}} where $i$ does not occur in $t$.
For such \dfn[vacuous!quantifier]{vacuous quantification} ($i$ does not occur anywhere), we may omit $\lforall[C]{i}{}$from assignments and differential equations, which are then classical assignments and ordinary differential equations as in HPs (\rref{ch:dL}).
We may omit vectors $\vec{s}$ of length 0.

\subsubsection{Regular Programs}
The \emph{test} $\ptest{\ivr}$ of a first-order formula $\ivr$ of real arithmetic is as in HPs except that $\ivr$ is a formula of many-sorted first-order logic.
Compound QHPs are generated from atomic QHPs by nondeterministic choice ($\cup$), sequential composition ($;$), and Kleene's nondeterministic repetition ($\prepeat{}$), just like in \rref{sec:HP}.
The (decisive) difference of \QHPs of \QdL compared to HPs of \dL is that \QHPs can contain \emph{quantified} assignments and \emph{quantified} differential equations with first-order functions.

\ignore{
The \dfn{test} action~\m{\ptest{\ivr}} is used to define conditions. Its effect is that of a \textit{no-op} if the formula~$\ivr$ is true in the current state; otherwise, like \textit{abort}, it allows no transitions.
That is, if the test succeeds because formula~$\ivr$ holds in the current state, then the state does not change, and the system execution continues normally.
If the test fails because formula~$\ivr$ does not hold in the current state, then the system execution cannot continue, is cut off and not considered further.

The nondeterministic choice~\m{\pchoice{\alpha}{\beta}}, sequential composition~\m{\alpha;\beta}, and non\-de\-ter\-min\-is\-tic repetition~\m{\prepeat{\alpha}} of programs are as in regular expressions but generalized to a semantics in distributed hybrid systems.
\dfn[nondeterministic!choice]{Nondeterministic choice} \m{\pchoice{\alpha}{\beta}} is used to express behavioral alternatives between the transitions of~$\alpha$ and~$\beta$.
That is, the \QHP~\m{\pchoice{\alpha}{\beta}} can choose nondeterministically to follow the transitions of \QHP~$\alpha$, or, instead, to follow the transitions of \QHP~$\beta$.
The \dfn[composition!sequential]{sequential composition}~\m{\alpha;\beta} says that the \QHP~$\beta$ starts executing after \QHP~$\alpha$ has finished ($\beta$ never starts if~$\alpha$ does not terminate).
In~\m{\alpha;\beta}, the transitions of~$\alpha$ take effect first, until~$\alpha$ terminates (if it does), and then~$\beta$ continues.
Observe that, like repetitions, continuous evolutions within~$\alpha$ can take more or less time, which causes uncountable nondeterminism.
This nondeterminism is inherent in distributed hybrid systems, because they can operate in so many different ways, which is as such reflected in \QHPs.
\dfn[nondeterministic!repetition]{Nondeterministic repetition}~\m{\prepeat{\alpha}} is used to express that the \QHP~$\alpha$ repeats any number of times, including zero times.
When following~\m{\prepeat{\alpha}}, the transitions of \QHP~$\alpha$ can be repeated over and over again, any nondeterministic number of times (\m{{\geq}0}).
}

\QHPs (with their semantics and our proof rules) can be extended to systems of quantified differential equations, systems of simultaneous assignments to multiple functions $f,g$, and statements with multiple quantifiers (\m{\lforall[C]{i}{\lforall[D]{j}{\dots}}}) similar to vectorial generalizations in discrete programs \cite{DBLP:conf/cade/BeckertP06,DBLP:conf/lpar/Rummer06}.

\begin{example}[Distributed car control] \label{ex:QHP}
  \let\laforall\lforall%
Continuous movement of position $x(i)$ of car $i$ with acceleration $a(i)$ is expressed by differential equation \m{\hevolve{\D[2]{x(i)}=a(i)}}, which corresponds to the first-order differential equation system \m{\hevolve{\D{x(i)}=v(i)\syssep\D{v(i)}=a(i)}} where $v(i)$ is the velocity of car $i$.
Simultaneous movement of all cars with their respective accelerations $a(i)$ is expressed by the quantified differential equation
\m{\dcevo} where quantifier $\lforall[C]{i}$ranges over all cars, such that all cars co-evolve along their respective differential equations at the same time.

In addition to continuous dynamics, cars have discrete control.
In the following \QHP, discrete and continuous dynamics interact (repeatedly because of the~$\prepeat{}$ repetition operator):
\begin{equation}
  \renewcommand*{\dcaccelt}[1]{({\oa{a}{#1}}\,{\mathrel{{:}{=}}}\,{\piif{\SBforma{#1}}{\amax}{\,{-}\abrake}})}%
  \big(\dcaccelallt;\,
  \dcevo \prepeat{\big)}
  \label{eq:distributed-car-control-accel}
\end{equation}
First, all cars $i$ control their acceleration $a(i)$.
Each car $i$ chooses maximum acceleration \m{\amax\geq0} for \m{a(i)} if its distance to all other cars $j$ is far enough (some condition \m{\textit{far}(i,j)} that depends on the velocities and either on the acceleration of $j$ or on reaction times $\varepsilon$ as in \rref{ex:dL}).
Otherwise, $i$ chooses full braking \m{-\abrake<0}.
After all accelerations have been set, all cars move continuously along \m{\dcevo}.
Accelerations may change repeatedly, because the repetition operator $\prepeat{}$ can repeat the  \QHP after the continuous evolution stops, which it can do at any time.
When \DCCS denotes QHP \rref{eq:distributed-car-control-accel}, the \QdL formula \rref{eq:distributed-car-control-new} from \rref{ex:QdL} is valid, when choosing \rref{eq:distributed-car-control-separate} for \m{\dcseparate{i}{j}}.
For more elaborate car models, verification results, and formal proofs, including verification results about distributed car control with dynamic appearance and disappearance of cars on highways with arbitrarily many cars on arbitrarily many lanes including onramps and exits, we refer to previous work \cite{DBLP:conf/csl/Platzer10,DBLP:journals/lmcs/Platzer12b,DBLP:conf/fm/LoosPN11}.
Observe that the \QHP \rref{eq:distributed-car-control-accel} requires car $i$ to check \m{\textit{far}(i,j)} for all other cars $j$, which is easy to model, but hard to implement.
We refer to previous work \cite{DBLP:conf/fm/LoosPN11} for \QdL proofs for distributed car control models that are globally safe even though each car reaches its control decisions solely based on local sensor/communication input and local control decisions.
\end{example}

Note that the presence of the function argument $i$ in $x(i), v(i), a(i)$ is a decisive difference when comparing the \QHP in \rref{ex:QHP} to the HP in \rref{ex:HP} and when comparing the \QdL formula in \rref{ex:QdL} to the \dL formula in \rref{ex:dL}.
In hybrid systems, we are limited to using variables $x,v,a$ of a single car.
If we want to add a second car to a hybrid system model, we need to add new state variables $y,w,c$, new dynamics \m{\hevolve{\D{y}=w\syssep\D{w}=c}}, and new control for the second car.
We can keep on adding any fixed finite number of state variables that way, but we need to know exactly how many cars there are on the street.
This does not work when we want to model and verify situations with arbitrarily many cars or in distributed car control scenarios like \rref{fig:distributed-car-control-new}, where new cars appear or disappear during the evolution of the system.
A quantified differential equation like \m{\hevolve{\lforall[C]{i}{(\D{x(i)}=v(i)\syssep\D{v(i)}=a(i))}}}, for example, cannot be expressed in hybrid systems, because we do not know how many cars $i$ ranges over.
If $i$ did range over exactly 3 cars, called 1, 2, and 3, we could replace it by
\[\hevolve{\D{x(1)}=v(1)\syssep\D{v(1)}=a(1)\syssep\D{x(2)}=v(2)\syssep\D{v(2)}=a(2)\syssep\D{x(3)}=v(3)\syssep\D{v(3)}=a(3)}\]
and change notation to obtain primitive state variables \(x_1,v_1,a_1,x_2,v_2,a_2,x_3,v_3,a_3\) in an ordinary differential equation system
\[\hevolve{\D{x_1}=v_1\syssep\D{v_1}=a_1\syssep\D{x_2}=v_2\syssep\D{v_2}=a_2\syssep\D{x_3}=v_3\syssep\D{v_3}=a_3}\]
But this replacement does not work unless we know exactly how many cars are in the system.
Even for systems with a fixed known but large number of participants, such flat representations as (non-distributed) hybrid systems are inefficient, because the system dimension is exponential in the number of participants and all reasoning needs to be repeated for each participant, or even for each pair of participants (collision freedom requires each pair of cars to remain safely separated).

In \QdL formulas and in \QHP models, we can leverage the distributed structure in systems, make the models more expressive, and make the reasoning more efficient by exploiting their first-order structure.
Only in \QdL can properties of distributed hybrid systems with an unknown or evolving number of participants be proved.
See previous work \cite{DBLP:conf/fm/LoosPN11} for a detailed practical illustration of those phenomena in verification of local and of distributed car control.

\subsection{Semantics} \label{sec:QdL-semantics}
The \QdL semantics is a \emph{constant domain Kripke semantics \cite{Fitting_Mendelsohn_1999} with first-order structures as states} that associate total functions of appropriate type with function symbols.
In constant domain, all states share the same domain for quantifiers.
We choose to represent object creation not by changing the domain of states, but by changing the interpretation of the createdness flag $\laexisting{i}$ of the object denoted by $i$.
With $\laexisting{i}$, object creation is definable in a modular way by masking the effect of QHPs to objects $i$ with $\laexisting{i}=1$ \cite{DBLP:conf/csl/Platzer10,DBLP:journals/lmcs/Platzer12b}.

\subsubsection{States}
A \emph{state} $\iget[state]{\I}$ 
associates an (infinite) set $\idomain{\I}{C}$ of objects with each sort $C$, and it
associates a function $\iget[state]{\I}(f)$ of appropriate type with each function symbol $f$, including $\laexisting{\cdot}$.
For simplicity, $\iget[state]{\I}$ also associates a value $\iget[state]{\I}(i)$ of appropriate type with each variable $i$.
The domain of $\reals$ and the interpretation of $0,1,+,-,\cdot$ is that of real arithmetic.
We assume \emph{constant domain} for each sort $C$: all states $\iget[state]{\I},\iget[state]{\It}$ share the same domains \m{\idomain{\I}{C}=\idomain{\It}{C}} for~$C$.
Sorts $C\neq D$ are disjoint: \m{\idomain{\I}{C} \cap \idomain{\I}{D} = \emptyset}.
The set of all states is again denoted by \m{\linterpretations{\Sigma}{V}}, but different from the set of states of \dL.
The state
\m{\iget[state]{\imodif[state]{\I}{i}{e}}} agrees with~$\iget[state]{\I}$ except for the interpretation of variable~$i$, which is changed to~\m{e \ignore{\idomain{\I}{C}}}.
We assume $\laexisting{\cdot}$ to have (unbounded but) finite support, i.e., each state only has a finite number of positions $i$ at which \m{\laexisting{i}=1}.
This makes sense in practice, because there is a varying and possibly large but still finite numbers of participants (e.g., cars).

\subsubsection{Formulas} \label{sec:QdL-valuation}
We use $\ivaluation{\I}{\theta}$ to denote the value of term~$\theta$ at $\iname[state]{\I}$~$\iget[state]{\I}$, which is defined as in first-order logic.
Especially, \m{\ivaluation{\imodif[state]{\I}{i}{e}}{\theta}} denotes the value of $\theta$ in state \m{\iget[state]{\imodif[state]{\I}{i}{e}}}, i.e., in state $\iget[state]{\I}$ with $i$ interpreted as $e$.
Further, \m{\iaccess[\alpha]{\I}\ignore{\subseteq \linterpretations{\Sigma}{V} \times \linterpretations{\Sigma}{V}}} denotes the state transition relation of \QHP~$\alpha$, which we define below.
{%
    \newcommand{\Id}{\imodif[state]{\I}{i}{e}}%
\begin{definition}[\QdL semantics] \label{def:QdL-valuation}
  The \dfn{interpretation} \m{\imodels{\I}{\phi}} of \QdL formula~$\phi$ with respect to $\iname[state]{\I}~\iget[state]{\I}$ is defined inductively as:
  \begin{itemize}
  \item \(\imodels{\I}{(\theta_1=\theta_2)}\)
    iff \(\ivaluation{\I}{\theta_1} = \ivaluation{\I}{\theta_2}\).
  \item \(\imodels{\I}{(\theta_1\geq\theta_2)}\)
    iff \(\ivaluation{\I}{\theta_1} \geq \ivaluation{\I}{\theta_2}\).
  \item \(\imodels{\I}{\lnot\phi}\) iff
    it is not the case that \(\imodels{\I}{\phi}\).
  \item \(\imodels{\I}{\phi \land \psi}\) iff
    \(\imodels{\I}{\phi}\) and \(\imodels{\I}{\psi}\).
  \item \(\imodels{\I}{\lforall[C]{i}{\phi}}\)
    iff
    \(\imodels{\Id}{\phi}\)
    for all objects \m{e\in\idomain{\I}{C}}.
    \index{$\lforall{}{}$}
  \item \(\imodels{\I}{\lexists[C]{i}{\phi}}\)
    iff
    \(\imodels{\Id}{\phi}\)
    for some object \m{e\in\idomain{\I}{C}}.
    \index{$\lexists{}{}$}
  \item
    \(\imodels{\I}{\dbox{\alpha}{\phi}}\)
      iff
      \(\imodels{\It}{\phi}\)
      for all states~$\iget[state]{\It}$ with
      \(\related{\iaccess[\alpha]{\I}}{\iget[state]{\I}}{\iget[state]{\It}}\).
      \index{$\dbox{\alpha}{}$}
     \item \(\imodels{\I}{\ddiamond{\alpha}{\phi}}\)
       iff
       \(\imodels{\It}{\phi}\)
       for some~$\iget[state]{\It}$ with
       \(\related{\iaccess[\alpha]{\I}}{\iget[state]{\I}}{\iget[state]{\It}}\).
       \index{$\ddiamond{\alpha}{}$}
     \end{itemize}
If \m{\imodels{\I}{\phi}}, then we say that $\phi$ is true at $\iportray{\I}$.
\QdL formula $\phi$ is \dfn{valid}, written \m{\entails\phi}, iff \m{\imodels{\I}{\phi}} for all $\iportray{\I}$.
\end{definition}

\subsubsection{Programs} \label{sec:QdL-QHP-transition}
The transition semantics of \QHPs is defined similar to the transition semantics of HPs, except that the quantified assignments and quantified differential equations need to be defined.
{\newcommand{\Ii}{\imodif[state]{\I}{i}{e}}%
   \newcommand{\ws}{\sigma}%
    \newcommand{\Ifz}[1][\zeta]{\iconcat[state=\varphi(#1)]{\stdI}}%
    \newcommand{\Ifzi}[1][\zeta]{\imodif[state]{\Ifz[#1]}{i}{e}}%
\begin{definition}[Transition semantics of \QHPs]
   The \dfn[transition~relation]{transition relation, \(\iaccess[\alpha]{\I}\subseteq \linterpretations{\Sigma}{V} \times \linterpretations{\Sigma}{V}\)}, of \QHP~$\alpha$
    specifies which $\iname[state]{\I}$ $\iget[state]{\It}\ignore{\in \linterpretations{\Sigma}{V}}$ is reachable from $\iget[state]{\I}\ignore{\in \linterpretations{\Sigma}{V}}$ by running \QHP $\alpha$.
    It is defined inductively:
    \begin{enumerate}
    \item \label{case:QdL-QHP-transition-assign}
      \(\relateds{\iaccess[\pupdate{\lforall[C]{i}{\pumod{f(\vec{s})}{\theta}}}]{\I}}{\iget[state]{\I}}{\iget[state]{\It}}\)
      iff $\iname[state]{\I}$~$\iget[state]{\It}$ is identical to~$\iget[state]{\I}$ except that
      at each position \m{\vec{o} \ignore{\in \iget[state]{\I}(\vec{S})}} of $f$:
      if \m{\ivaluation{\Ii}{\vec{s}} = \vec{o}} for some object \m{e\in\idomain{\I}{C}}, then
      \m{\iget[state]{\It}(f)\big(\ivaluation{\Ii}{\vec{s}}\big) = \ivaluation{\Ii}{\theta}}.
      If there are multiple objects $e$ giving the same position \m{\ivaluation{\Ii}{\vec{s}} = \vec{o}}, then all of the resulting states $\iget[state]{\It}$ are reachable.
      
    \item \label{case:QdL-QHP-transition-evolve}
      \(\relateds{\iaccess[\hevolvein{\lforall[C]{i}{\D{f(\vec{s})}=\theta}}{\ivr}]{\I}}{\iget[state]{\I}}{\iget[state]{\It}}\)
      iff
      there is a\ignore{ (\emph{flow})} function
      \({{\varphi}{:}{\interval{[0,r]}\to\linterpretations{\Sigma}{V}}}\)
      for some \m{r\geq0} with
      \(\varphi(0)=\iget[state]{\I}\) and \(\varphi(r)=\iget[state]{\It}\)
      satisfying the following conditions.
      At each time \m{t \in \interval{[0,r]}}, state $\iget[state]{\Ifz[t]}$ is identical to $\iget[state]{\I}$, except that
      at each position \m{\vec{o} \ignore{\in \iget[state]{\I}(\vec{S})}} of $f$:
        if \m{\ivaluation{\Ii}{\vec{s}} = \vec{o}} for some object \m{e\in\idomain{\I}{C}}, then, at each time \m{\zeta \in \interval{[0,r]}}:
      \begin{itemize}
        \item All differential equations hold and corresponding derivatives exist (trivial for \m{r=0}):
        \[
        \D[t]{\,(\ivaluation{\Ifzi[t]}{f(\vec{s})})} (\zeta) = \ivaluation{\Ifzi[\zeta]}{\theta}
        \]
      \item The evolution domain is respected:
      \m{\imodels{\Ifzi}{\ivr}}.
      \end{itemize}
      If there are multiple objects $e$ giving the same position \m{\ivaluation{\Ii}{\vec{s}} = \vec{o}}, then all of the resulting states $\iget[state]{\It}$ are reachable.
    \item \(\iaccess[\ptest{\ivr}]{\I} =
      \{(\iget[state]{\I},\iget[state]{\I}) {\with}  \imodels{\I}{\ivr}\}\)
\item \m{\iaccess[\pchoice{\alpha}{\beta}]{\I} = \iaccess[\alpha]{\I} \cup \iaccess[\beta]{\I}}
\item \m{\iaccess[\alpha;\beta]{\I} = \iaccess[\beta]{\I} \compose\iaccess[\alpha]{\I}}
\item \m{\iaccess[\prepeat{\alpha}]{\I} = \displaystyle\cupfold_{n\in\naturals}\iaccess[{\prepeat[n]{\alpha}}]{\I}} 
with \m{\prepeat[n+1]{\alpha} \mequiv \prepeat[n]{\alpha};\alpha} and \m{\prepeat[0]{\alpha}\mequiv\,\ptest{\ltrue}}.
    \end{enumerate}
\end{definition}
The semantics is \emph{explicit change}:
nothing changes unless an assignment or differential equation specifies how.
In cases~\ref{case:QdL-QHP-transition-assign}--\ref{case:QdL-QHP-transition-evolve}, only $f$ changes and only at positions of the form \m{\ivaluation{\Ii}{\vec{s}}} for some interpretation \m{e\in\idomain{\I}{C}} of $i$.
If there are multiple such $e$ that affect the same position $\vec{o}$, any of those changes can take effect by a nondeterministic choice.
\QHP \m{\hupdate{\lforall[C]{i}{\umod{x}{a(i)}}}} may change $x$ to \emph{any} $a(i)$.
Hence,
\m{\dbox{\hupdate{\lforall[C]{i}{\umod{x}{a(i)}}}}{\mapply{\phi}{x}} \mequiv \lforall[C]{i}{\mapply{\phi}{a(i)}}},
because that modality considers \emph{all} possibilities of changing $x$ to \emph{any} $a(i)$.
In contrast,
\m{\ddiamond{\hupdate{\lforall[C]{i}{\umod{x}{a(i)}}}}{\mapply{\phi}{x}} \mequiv \lexists[C]{i}{\mapply{\phi}{a(i)}}},
because that modality considers \emph{some} possibility of changing $x$ to \emph{any} $a(i)$.
Similarly, $x$ can evolve along \m{\hevolve{\lforall[C]{i}{\D{x}=a(i)}}} with any of the slopes $a(i)$. But evolutions cannot start with slope $a(c)$ and then switch to a different slope $a(d)$ later.
Any choice for the quantified variable $i$ is possible but $i$ remains unchanged during each evolution.

We call a quantified assignment \m{\pupdate{\lforall[C]{i}{\pumod{f(\vec{s})}{\theta}}}}  or a quantified differential equation \m{\hevolvein{\lforall[C]{i}{\D{f(\vec{s})}=\theta}}{\ivr}} \dfn[injective!QHP]{injective} iff there is at most one $e$ satisfying cases~\ref{case:QdL-QHP-transition-assign}--\ref{case:QdL-QHP-transition-evolve}.
For injective quantified assignments and injective quantified differential equations, conditions~\ref{case:QdL-QHP-transition-assign}--\ref{case:QdL-QHP-transition-evolve} can be simplified as follows:
\begin{enumerate}
\item
      \(\relateds{\iaccess[\pupdate{\lforall[C]{i}{\pumod{f(\vec{s})}{\theta}}}]{\I}}{\iget[state]{\I}}{\iget[state]{\It}}\)
      iff $\iname[state]{\I}$~$\iget[state]{\It}$ is identical to~$\iget[state]{\I}$ except that
      for each \m{e\in\idomain{\I}{C}}:
      \m{\iget[state]{\It}(f)\big(\ivaluation{\Ii}{\vec{s}}\big) = \ivaluation{\Ii}{\theta}}.
\item
      \(\relateds{\iaccess[\hevolvein{\lforall[C]{i}{\D{f(\vec{s})}=\theta}}{\ivr}]{\I}}{\iget[state]{\I}}{\iget[state]{\It}}\)
      iff
      there is a\ignore{ (\emph{flow})} function
      \({{\varphi}{:}{\interval{[0,r]}\to\linterpretations{\Sigma}{V}}}\)
      for some \m{r\geq0} with
      \(\varphi(0)=\iget[state]{\I}\) and \(\varphi(r)=\iget[state]{\It}\)
      such that
      for each \m{e\in\idomain{\I}{C}} and each time \m{\zeta \in \interval{[0,r]}}:
      \begin{itemize}
        \item All differential equations hold and corresponding derivatives exist (trivial for \m{r=0}):
        \[
        \D[t]{\,(\ivaluation{\Ifzi[t]}{f(\vec{s})})} (\zeta) = \ivaluation{\Ifzi[\zeta]}{\theta}
        \]
      \item The evolution domain is respected:
      \m{\imodels{\Ifzi}{\ivr}}.
      \end{itemize}
\end{enumerate}
We call quantified assignments and quantified differential equations \dfn[schematic!QHP]{schematic} iff $\vec{s}$ is $i$ (thus injective) and the only arguments to function symbols in $\theta$ are $i$.
Schematic quantified differential equations like \m{\hevolvein{\lforall[C]{i}{\D{f(i)}=a(i)}}{\ivr}} are very common, because distributed hybrid systems often have a family of similar differential equations replicated for multiple participants $i$. Their synchronization often comes from discrete communication on top of their continuous dynamics. Physically coupled differential equations are possible as well.
They correspond to continuous physical interactions, e.g., if a car bumps into another car from the side, it radically changes the structure of the differential equations that determine its movement.
Either case can be represented in \QHPs, even if the schematic case is more common.

\subsection{Axiomatization} \label{sec:QdL-calculus}

\begin{figure*}[tbh]
  \renewcommand*{\irrulename}[1]{\text{#1}}%
  \newdimen\linferenceRulehskipamount%
  \linferenceRulehskipamount=1mm%
  \newdimen\lcalculuscollectionvskipamount%
  \lcalculuscollectionvskipamount=0.1em%
\advance\leftskip-0.5cm
  \begin{calculuscollections}{\columnwidth}
    \begin{calculus}
      \cinferenceRule[Qupskip|${[:]}$]{update skip}
      {
        \linferenceRule[equiv]
        {\mapply{\mascriptor}{\dbox{\pupdate{\lforall[C]{i}{\umod{f(\vec{s})}{\theta}}}}{\vec{u}}}}
        {\dbox{\pupdate{\lforall[C]{i}{\umod{f(\vec{s})}{\theta}}}}{\mapply{\mascriptor}{\vec{u}}}}
      }{\m{f\neq\mascriptor}}%
      \cinferenceRule[Qassignb|${[:=]}$]{quantified assignment}
      {\linferenceRule[equiv]
        {
            \piif{\lexists[C]{i}{\vec{s}=\dbox{\jupd}{\vec{u}}}}
            {\lforall[C]{i}{(\vec{s}=\dbox{\jupd}{\vec{u}} \limply \mapply{\phi}{\theta})}}
            {\mapply{\phi}{f(\dbox{\jupd}{\vec{u}})}}
        }
        {\mapply{\phi}{\dbox{\pupdate{\lforall[C]{i}{\umod{f(\vec{s})}{\theta}}}}{f(\vec{u})}}}
      }{}%
      \cinferenceRule[Qassignsb|$\dibox{:=}_s$]{quantified assignment  axiom, schematic case}
      {\linferenceRule[equiv]
        {
            \lforall[C]{i}{(i=\dbox{\pupdate{\lforall[C]{i}{\umod{f(i)}{\theta}}}}{u} \limply \mapply{\phi}{\theta})}
        }
        {\mapply{\phi}{\dbox{\pupdate{\lforall[C]{i}{\umod{f(i)}{\theta}}}}{f(u)}}}
      }{}
      \cinferenceRule[Qassignrb|${[{:}{*}]}$]{random assignment}
      {\linferenceRule[equiv]
        {\lforall[C]{j}{\mapply{\phi}{\theta}}}
        {\dbox{\pupdate{\lforall[C]{j}{\pumod{\onew{}}{\theta}}}}{\mapply{\phi}{\onew{}}}}
      }{}
      \cinferenceRule[testb|$\dibox{?}$]{test}
      {\linferenceRule[equiv]
        {(\ivr \limply \phi)}
        {\dbox{\ptest{\ivr}}{\phi}}
      }{}
      \cinferenceRule[Qevolveb|$\dibox{'}$]{evolve}
      {\linferenceRule[equiv]
        {\lforall{t{\geq}0}{\dbox{\pupdate{\lforall[C]{i}{\pumod{f(\vec{s})}{\solf_{\vec{s}}(t)}}}}{\phi}}\quad}
        {\dbox{\hevolve{\lforall[C]{i}{\D{f(\vec{s})}=\theta}}}{\phi}}
      }{\m{\D{\solf_{\vec{s}}}(t)=\genDE{\solf}\,\forall i}}%
    \cinferenceRule[Qevolveinb|${[\&]}$]{evolution domain restriction} %
      {\linferenceRule[equiv]
        {\lforall{t_0{=}\stime}{\dbox{\hevolve{\lforall[C]{i}{\D{f(\vec{s})}=\genDE{x}}}}{}}
        {{\big(\dbox{\hevolve{\lforall[C]{i}{\D{f(\vec{s})}=-\genDE{x}}}}{(\stime\geq t_0\limply\ivr)} \limply \phi\big)}}}
        {\dbox{\hevolvein{\lforall[C]{i}{\D{f(\vec{s})}=\genDE{x}}}{\ivr}}{\phi}}
      }{}%
      \cinferenceRule[choiceb|$\dibox{\cup}$]{axiom of nondeterministic choice}
      {\linferenceRule[equiv]
        {\dbox{\alpha}{\phi} \land \dbox{\beta}{\phi}}
        {\dbox{\pchoice{\alpha}{\beta}}{\phi}}
      }{}
      \cinferenceRule[composeb|$\dibox{{;}}$]{composition} %
      {\linferenceRule[equiv]
        {\dbox{\alpha}{\dbox{\beta}{\phi}}}
        {\dbox{\alpha;\beta}{\phi}}
      }{}
      \cinferenceRule[iterateb|$\dibox{{}^*}$]{iteration/repeat unwind} %
      {\linferenceRule[equiv]
        {\phi \land \dbox{\alpha}{\dbox{\prepeat{\alpha}}{\phi}}}
        {\dbox{\prepeat{\alpha}}{\phi}}
      }{}
      \cinferenceRule[newex|\usebox{\exbox}]{new existence pool}
      {\lexists[C]{\onew{}}{\lnaexisting{\onew{}}}}
      {\m{C\neq\reals}}
      {}%
      \cinferenceRule[K|K]{K axiom / modal modus ponens} %
      {\linferenceRule[impl]
        {\dbox{\alpha}{(\phi\limply\psi)}}
        {(\dbox{\alpha}{\phi}\limply\dbox{\alpha}{\psi})}
      }{}
      \cinferenceRule[I|I]{loop induction}
      {\linferenceRule[impl]
        {\dbox{\prepeat{\alpha}}{(\inv\limply\dbox{\alpha}{\inv})}}
        {(\inv\limply\dbox{\prepeat{\alpha}}{\inv})}
      }{}
      \cinferenceRule[C|C]{loop convergence}
      {\linferenceRule[impl]
        {\dbox{\prepeat{\alpha}}{\lforall{v{>}0}{(\mapply{\var}{v}\limply\ddiamond{\alpha}{\mapply{\var}{v-1}})}}}
        {\lforall{v}{(\mapply{\var}{v} \limply
            \ddiamond{\prepeat{\alpha}}{\lexists{v{\leq}0}{\mapply{\var}{v}}})}}
      }{\m{v\not\in\alpha}}%
      \cinferenceRule[B|B]{Barcan$\dbox{}{}\forall{}$} %
      {\linferenceRule[impl]
        {\lforall[C]{x}{\dbox{\alpha}{\phi}}}
        {\dbox{\alpha}{\lforall[C]{x}{\phi}}}
      }{\m{x\not\in\alpha}}
      \cinferenceRule[V|V]{vacuous $\dbox{}{}$}
      {\linferenceRule[impl]
        {\phi}
        {\dbox{\alpha}{\phi}}
      }{\m{FV(\phi)\cap BV(\alpha)=\emptyset}}%
      \cinferenceRule[G|G]{$\dbox{}{}$ generalisation} %
      {\linferenceRule[formula]
        {\phi}
        {\dbox{\alpha}{\phi}}
      }{}
      \cinferenceRule[MP|MP]{modus ponens}
      {\linferenceRule[formula]
        {\phi\limply\psi \quad \phi}
        {\psi}
      }{}%
      \cinferenceRule[gena|$\forall$]{$\forall{}$ generalisation}
      {\linferenceRule[formula]
        {\phi}
        {\lforall[C]{x}{\phi}}
      }{}%
    \end{calculus}%
  \end{calculuscollections}
  \caption{Quantified differential dynamic logic axiomatization}
  \label{fig:QdL}
\end{figure*}

Our axiomatization of \QdL is shown in \rref{fig:QdL}.
To again highlight the logical essentials, we use an axiomatization that is significantly simplified compared to our earlier work \cite{DBLP:conf/csl/Platzer10,DBLP:journals/lmcs/Platzer12b}.
The axiomatization we use here is in the spirit of our simpler \dL axiomatization that we show in \rref{sec:dL-calculus}.
We use the first-order Hilbert calculus (modus ponens \irref{MP} and $\forall$-generalization rule \irref{gena}) as a basis and allow all instances of valid formulas of many-sorted first-order logic and first-order real arithmetic as axioms.
The first-order theory of real-closed fields is decidable \cite{tarski_decisionalgebra51}.
More constructive deduction modulo rules, which can be used to combine first-order real arithmetic of many-sorted first-order logic with the proof calculus presented here and are suitable for automation, have been reported in previous work \cite{DBLP:conf/csl/Platzer10,DBLP:journals/lmcs/Platzer12b}.
Note that the combination of first-order real arithmetic augmented with many-sorted function symbols is more challenging than the decidable first-order arithmetic of real-closed fields used as a basis for \dL.
It can still be handled with a combination of free variables, instantiation, requantification, and quantifier elimination \cite{DBLP:conf/csl/Platzer10,DBLP:journals/lmcs/Platzer12b}, which we use to lift quantifier elimination to the context of \dL (\rref{lem:qelim-lift}) using real-valued free variables, Skolemization, and Deskolemization for automation purposes \cite{DBLP:journals/jar/Platzer08}.

We write \m{\infers \phi} iff \QdL formula $\phi$ can be \emph{proved} with \QdL rules from \QdL axioms (including first-order rules and axioms); see \rref{fig:QdL}.

The \QdL axioms \irref{testb}, \irref{choiceb}, \irref{composeb}, \irref{iterateb}, \irref{K}, \irref{I}, \irref{C}, \irref{B}, \irref{V}, and rule \irref{G} are as for \dL in \rref{sec:dL-calculus}, because \QdL is a modular extension of \dL and the operators $\ptest{},\pchoice{},;,\prepeat{}$ have the same compositional semantics as in \dL.
We use the same first-order rules \irref{MP} and \irref{gena}, except that \irref{gena} applies to variables $x$ of any sort $C$.
The axioms \irref{Qassignb}, \irref{Qassignsb}, \irref{Qassignrb}, \irref{Qupskip} for quantified assignments, \irref{Qevolveb}, \irref{Qevolveinb} for quantified differential equations, and \irref{newex} for object creation are specific to \QdL.
Observe that, despite the radical semantical generalization, an important principle of generalizing \dL to \QdL is modularity.
One local, but decisive change is from \dL's primitive variables to \QdL's first-order variables.
The other changes are modular in the syntax, semantics, and axiomatization and consist in adding new cases for quantified assignments and for quantified differential equations (and object creation).

Axiom \irref{Qupskip} characterizes the fact that quantified assignments to $f$ have no effect on all other operators $\mascriptor\neq f$ (including other function symbols, $\land$, $\piif{}{}{}$), so that $\mascriptor$ will not be affected by the quantified assignment and can be skipped over.
The argument $\vec{u}$ may still be affected by the quantified assignment, hence \irref{Qupskip} prefixes $\vec{u}$ (component-wise) by \m{\pupdate{\lforall[C]{i}{\umod{f(\vec{s})}{\theta}}}}.
Thus, the \irref{Qupskip} axiom maps a quantified assignment over all arguments homomorphically.
For example, if $\mascriptor$ is an operator taking two arguments and is not the function symbol $f$, then axiom \irref{Qupskip} derives the proof step
\begin{sequentdeduction}
        \linfer[Qupskip]
        {\lsequent[s]{}{\mapply{\mascriptor}{\dbox{\pupdate{\lforall[C]{i}{\umod{f(\vec{s})}{\theta}}}}{u_1},\dbox{\pupdate{\lforall[C]{i}{\umod{f(\vec{s})}{\theta}}}}{u_2}}}}
        {\lsequent[s]{}{\dbox{\pupdate{\lforall[C]{i}{\umod{f(\vec{s})}{\theta}}}}{\mapply{\mascriptor}{u_1,u_2}}}}
\end{sequentdeduction}

Axiom \irref{Qassignb} characterizes how a quantified assignment to $f$ affects the value of a term $f(\vec{u})$.
The effect depends on whether the quantified assignment \m{\pupdate{\lforall[C]{i}{\umod{f(\vec{s})}{\theta}}}} \dfn[match]{matches} \m{f(\vec{u})}, i.e., there is a choice for $i$ such that \m{f(\vec{u})} is affected by the assignment, because $\vec{u}$ is of the form $\vec{s}$ for some $i$.
Whether it matches or not cannot always be decided statically, because it may depend on the particular interpretations.
Hence, axiom \irref{Qassignb} makes a case distinction on matching by yielding an \keywordfont{if-then-else} formula.
The formula \m{\piif{\phi}{\phi_1}{\phi_2}} is short notation for
\m{(\phi \limply \phi_1)  \land  (\lnot\phi \limply \phi_2)}.
If the quantified assignment does not match (\keywordfont{else} part), the occurrence of $f$ in \m{\mapply{\phi}{f(\vec{u})}} will be left unchanged, because $f$ is not changed at position $\vec{u}$.
If it matches (\keywordfont{then} part), the term $\theta$ assigned to \m{f(\vec{s})} is used instead of \m{f(\vec{u})}, for all possible \m{\hastype{i}{C}} that match \m{f(\vec{u})}.
In all cases, the original quantified assignment \m{\pupdate{\lforall[C]{i}{\umod{f(\vec{s})}{\theta}}}}, which we abbreviate by $\jupd$ in \irref{Qassignb}, will be applied to \m{\vec{u}} in the premise, because the value of argument \m{\vec{u}} may also be affected by \m{\jupd}, recursively.
Recall that axioms \irref{Qassignb} and \irref{Qupskip} assume \m{\pupdate{\lforall[C]{i}{\umod{f(\vec{s})}{\theta}}}} to be injective or vacuous.

Axiom \irref{Qassignsb} is an important special case of \irref{Qassignb} that applies to the schematic case where $\vec{s}$ is of the form $i$, which matches trivially.
If $f$ does not occur in $u$, then \irref{Qupskip} simplifies this further:
\[
      \linfer[Qassignsb+Qupskip]
        {\lsequent[s]{}{
            \lforall[C]{i}{(i=u \limply \mapply{\phi}{\theta})}
        }}
        {\lsequent[s]{}{\mapply{\phi}{\dbox{\pupdate{\lforall[C]{i}{\umod{f(i)}{\theta}}}}{f(u)}}}}
\]
When $f$ does not occur in $u$, standard first-order reasoning can simplify further ($\subst[\theta]{i}{u}$ is the term $\theta$ with $i$ replaced by $u$):
\[
      \linfer[Qassignb]
        {\lsequent[s]{}{
            \mapply{\phi}{\subst[\theta]{i}{u}}
        }}
        {\lsequent[s]{}{\mapply{\phi}{\dbox{\pupdate{\lforall[C]{i}{\umod{f(i)}{\theta}}}}{f(u)}}}}
\]
Together with \irref{Qupskip} to propagate the change to both arguments of $\neq$, this derived rule proves, e.g., the following proof step:
\begin{sequentdeduction}[array]
  \linfer[Qassignb+Qupskip]
                   {\lsequent{}%
                     {\lforall{j{\neq}k}{({-}{\frac{b}{2}}\skolem{s}^2+v(j)\skolem{s} + x(j) \neq {-}{\frac{b}{2}}\skolem{s}^2+v(k)\skolem{s} + x(k))}}}
               {\lsequent{}%
                 {\lforall{j{\neq}k}{\dbox{\hupdate{\lforall{i}{\humod{x(i)}{{-}{\frac{b}{2}}\skolem{s}^2+v(i)\skolem{s} + x(i)}}}}{\,x(j){\neq}x(k)}}}}
\end{sequentdeduction}

Axiom \irref{Qassignrb} reduces nondeterministic assignments to universal quantification.
For the handling of other general nondeterministic assignments and nondeterministic differential equations, we refer to previous work \cite{DBLP:journals/logcom/Platzer10,Platzer10}.

Axioms \irref{Qassignb+Qupskip} also apply for assignments without quantifiers, which correspond to vacuous quantification $\lforall[C]{i}{}$where $i$ does not occur anywhere; see previous work \cite{DBLP:conf/csl/Platzer10,DBLP:journals/lmcs/Platzer12b}.
That case amounts to a notational variant of the \irref{assignb} axiom of \dL from \rref{fig:dL}.
Vectorial extensions to systems of quantified assignments and systems of quantified differential equations with multiple function symbols are possible as well \cite{DBLP:conf/csl/Platzer10,DBLP:journals/lmcs/Platzer12b}.

Axiom \irref{Qevolveb} handles continuous evolutions for quantified differential equations with first-order definable solutions.
The difference compared to the axiom \irref{evolveb} of \dL is that \QdL handles infinite-dimensional quantified differential equation systems or quantified differential equation systems with evolving dimensions.
Their solutions are no longer expressible as assignments, but need quantified assignments.
Given a solution for the quantified differential equation system with symbolic initial values~$f(\vec{s})$, continuous evolution along differential equations can be replaced with a quantified assignment \m{\pupdate{\lforall[C]{i}{\pumod{f(\vec{s})}{\solf_{\vec{s}}(t)}}}} corresponding to the simultaneous solution (of the differential equations \m{\hevolve{\lforall[C]{i}{\D{f(\vec{s})}=\theta}}} with~\m{f(\vec{s})} as symbolic initial values)
and an additional quantifier for all evolution durations~$t\geq0$.

For schematic cases like \m{\hevolve{\lforall[C]{i}{\D{f(i)}=a(i)}}}, first-order definable solutions can be obtained by adding argument $i$ to first-order definable solutions of the deparametrized version \m{\hevolve{\D{f}=a}}.
For example, the following proof step uses axiom \irref{Qevolveb} to turn a quantified differential equation system into a quantified assignment with an extra quantifier for the common duration $t$ of the evolution.
{%
   \def\arraystretch{1.3}%
    \begin{sequentdeduction}[array]
          \linfer[Qevolveb]
           {\lsequent{}%
             {\lforall{t{\geq}0}{\dbox{\hupdate{\lforall{i}{\humod{x(i)}{{-}{\frac{b}{2}}t^2+v(i)t + x(i)}}}}{\,\lforall{j{\neq}k}{x(j){\neq}x(k)}}}}}
         {\lsequent{}%
           {\dbox{\hevolve{\lforall{i}{\D{x(i)}=v(i)\syssep\D{v(i)}=-b}}}{\,\lforall{j{\neq}k}{x(j){\neq}x(k)}}}}
  \end{sequentdeduction}
}%
The quantified assignment \m{\hupdate{\lforall{i}{\humod{x(i)}{{-}{\frac{b}{2}}t^2+v(i)t + x(i)}}}} solving the above quantified differential equation system can be obtained easily from the solution \m{\hupdate{\humod{x}{{-}{\frac{b}{2}}t^2+vt + x}}} of the deparametrized differential equation system \m{\hevolve{\D{x}=v\syssep\D{v}=-b}}, just by adding the parameter $i$ back in and checking whether this gives the solution.

The modular ``there and back again'' axiom \irref{Qevolveinb} that reduces quantified differential equations with evolution domain constraints to quantified differential equations without them works as in \dL (\rref{sec:dL-calculus}).
For an explanation how quantified differential equations have unique solutions as required for this to be sound, we refer to previous work \cite{DBLP:conf/csl/Platzer10,DBLP:journals/lmcs/Platzer12b}.
The \QdL axioms \irref{testb}, \irref{choiceb}, \irref{composeb}, \irref{iterateb}, \irref{K}, \irref{I}, \irref{C}, \irref{B}, \irref{V}, and rule \irref{G} are as for \dL, even if \irref{B} and rule \irref{gena} are now many-sorted.

Axiom \irref{newex} expresses that, for sort $C\neq\reals$, there always is a new object $\onew{}$ that has not been created yet ($\lnaexisting{\onew{}}$), because domains are infinite.
This is the only place where we are using the assumption about infinite domains.
The primary purpose is to simplify technicalities that would arise if object creation could run out of objects and may thus fail if, e.g., no more cars can be created; see previous work \cite{DBLP:conf/csl/Platzer10,DBLP:journals/lmcs/Platzer12b} for details on object/agent creation.

\begin{example}[Distributed car control] \label{ex:QdL-proof}
\def\prem{\lforall{i{\neq}j}{x(i){\neq}x(j)}}%
To illustrate how the \QdL proof calculus works, we use \QdL axioms as shown in \rref{fig:QdL-verification-example} to identify when the \QdL formula at the bottom is valid.
\begin{figure*}[tbh]
\def\prem{\lforall{i{\neq}j}{x(i){\neq}x(j)}}%
     \def\arraystretch{1.3}%
      \begin{sequentdeduction}[array]
          \linfer[evolveb]
            {\linfer[Qupskip]
              {\linfer[Qassignsb]
                {\linfer[qear]
                  {\linfer[qear]
                    {\lsequent{\prem} {\lforall{j{\neq}k}{(
                            x(j)\leq x(k)\land v(j)\leq v(k) \lor x(j)\geq x(k)\land v(j)\geq v(k))}}}
                 {\lsequent{\prem} {\lforall{j{\neq}k}{\lforall{t{\geq}0}{({-}{\frac{b}{2}}t^2+v(j)t + x(j) \neq {-}{\frac{b}{2}}t^2+v(k)t + x(k))}}}}
               }%
               {\lsequent{\prem} {\lforall{t{\geq}0}{\lforall{j{\neq}k}{({-}{\frac{b}{2}}t^2+v(j)t + x(j) \neq {-}{\frac{b}{2}}t^2+v(k)t + x(k))}}}}
             }%
             {\lsequent{\prem} {\lforall{t{\geq}0}{\lforall{j{\neq}k}{\dbox{\hupdate{\lforall{i}{\humod{x(i)}{{-}{\frac{b}{2}}t^2+v(i)t + x(i)}}}}{\,x(j){\neq}x(k)}}}}}
           }%
           {\lsequent{\prem}  {\lforall{t{\geq}0}{\dbox{\hupdate{\lforall{i}{\humod{x(i)}{{-}{\frac{b}{2}}t^2+v(i)t + x(i)}}}}{\,\lforall{j{\neq}k}{x(j){\neq}x(k)}}}}}
         }%
         {\lsequent{\prem} {\dbox{\hevolve{\lforall{i}{\D{x(i)}=v(i)\syssep\D{v(i)}=-b}}}{\,\lforall{j{\neq}k}{x(j){\neq}x(k)}}}}
      \end{sequentdeduction}
  \caption{Example of a \QdL prove about collision-freedom of simple distributed car control}
  \label{fig:QdL-verification-example}
\end{figure*}
The \QdL formula that we consider here follows the pattern of the running example formula \rref{eq:distributed-car-control-new}. We only consider the braking case for typesetting reasons and refer to \cite{DBLP:journals/lmcs/Platzer12b} for a full proof.
The case we consider is the distributed hybrid systems analog of \dL formula \rref{eq:car-single-braking-discovery} that we proved in \rref{ex:dL-proof}.
The other difference compared to \rref{eq:distributed-car-control-new} is that the formula in \rref{fig:QdL-verification-example} has a weaker assumption.
It only assumes that cars start from different positions (\m{\prem}), not that they respect the compatibility constraint \m{\let\laforall\lforall \dcinv}.
Like in \rref{ex:dL-proof}, we are using the \QdL axioms to find out by the derivation in \rref{fig:QdL-verification-example} how we need to choose \m{\dcseparate{i}{j}} to ensure collision freedom.

We start with the conjecture at the bottom of \rref{fig:QdL-verification-example} and successively reduce it by using \QdL axioms and proof rules.
First, we use axiom \irref{Qevolveb} to turn the quantified differential equation system into a quantified assignment with an extra quantifier for the duration $t$ of the evolution.
The quantified differential equation system is easy to solve.
The quantified assignment \m{\hupdate{\lforall{i}{\humod{x(i)}{{-}{\frac{b}{2}}t^2+v(i)t + x(i)}}}} solving it can be obtained easily from the solution \m{\hupdate{\humod{x}{{-}{\frac{b}{2}}t^2+vt + x}}} of the deparametrized differential equation system \m{\hevolve{\D{x}=v\syssep\D{v}=-b}}, just by adding the parameter $i$ back in and checking that the resulting terms solve the quantified differential equation.
The premise of the use of \irref{evolveb} has a quantifier \m{\forall{t}} as the top-most logical operator in the succedent.
Even though it is a quantifier over a real variable, we cannot use the decision procedure of quantifier elimination for real-closed fields \cite{tarski_decisionalgebra51} to handle it, because we do not have a formula of first-order real arithmetic, but still a \QdL formula with a modality expressing a property of all reachable states.
Instead, we use axiom \irref{Qupskip} to skip over the quantifier \m{\forall{j{\neq}k}} and then use axiom \irref{Qassignsb} or \irref{Qassignb} to let the quantified assignment to $x(i)$ (for all $i$) take effect on the postcondition \m{x(j)\neq x(k)} by skipping over the $\neq$ with axiom \irref{Qupskip} (not shown in \rref{fig:QdL-verification-example}) and then affecting $x(j)$ and $x(k)$ subsequently by rule \irref{Qassignsb}.

At this point (premise of the top-most use of axiom \irref{Qassignsb} in \rref{fig:QdL-verification-example}), we already have a first-order formula and it looks like we could apply quantifier elimination for real-closed fields \cite{tarski_decisionalgebra51} to the quantified real variable $t$.
This would not work, however, because quantifier elimination works from inside out and will try to eliminate the inner quantifier \m{\forall{j{\neq}k}} before the outer quantifier \m{\forall{\skolem{s}}}.
Yet, this formula is not an instance of first-order real arithmetic (not even when using \rref{lem:qelim-lift}), but of many-sorted first-order logic with quantified variables $j,k$ of sort $C$, because there are dependencies on the quantified variables $j,k$ in function arguments, which is fundamentally more difficult \cite{Platzer10}.
Instead, the proof in \rref{fig:QdL-verification-example} uses a tautology of first-order logic (marked \irref{qear}) to commute the quantifiers.
Now, we can apply quantifier elimination for real-closed fields \cite{tarski_decisionalgebra51} as an equivalence in first-order logic (written \irref{qear}), even though the formula is still not in first-order real arithmetic, because function symbols like $v(j)$ occur.
However, it is an instance ($v(\skolem{j})$ for $V$ and $x(\skolem{j})$ for $X$ and $v(\skolem{k})$ for $W$ and $x(\skolem{k})$ for $Y$) of the following formula of first-order real arithmetic:
\begin{equation}
\lforall{\skolem{s}{\geq}0}{\Big(\skolem{j}{\neq}\skolem{k} \limply
                          {{-}{\frac{b}{2}}\skolem{s}^2+V\skolem{s} + X \neq {-}{\frac{b}{2}}\skolem{s}^2+W\skolem{s} + Y}\Big)}
\label{eq:qelim-lift-ex}
\end{equation}
and, thus, quantifier elimination can be lifted by \rref{lem:qelim-lift}.
The result of quantifier elimination is an instance (with the same instantiation as above) of the result of applying \qelim{} to \rref{eq:qelim-lift-ex}.
Finally, we could use real arithmetic on the top-most formula, as an instance of the following formula of plain first-order real arithmetic (with the instantiation $x(j)$ for $X$, $v(j)$ for $V$, $x(k)$ for $Y$, and $v(k)$ for $W$) 
\[
j\neq k %
\limply X\leq Y\land V\leq W \lor X\geq Y \land V\geq W
\]
However, the formula is not valid, so the proof does not close.
This is good news for soundness, however, because the conjecture at the bottom of \rref{fig:QdL-verification-example} is not valid, unless the constraints at the top hold about the relation of the velocities and positions of the cars.
The constraints at the top of \rref{fig:QdL-verification-example} can be used to construct the constraints required for safety, which coincide with \m{\dcseparate{j}{k}} that we have shown in \rref{eq:distributed-car-control-separate}.
\end{example}

A comparison of the \QdL proof in \rref{ex:QdL-proof} compared to the \dL proofs in \rref{ex:dL-proof} shows that the interplay of instantiation in many-sorted first-order logic and quantifier elimination in first-order real arithmetic is an important challenge when reasoning about distributed hybrid systems.
We refer to previous work \cite{DBLP:conf/csl/Platzer10,DBLP:journals/lmcs/Platzer12b,DBLP:journals/jar/Platzer08,DBLP:conf/icfem/RenshawLP11} for details on how this can be automated in the \QdL calculus.

\subsection{Soundness and Completeness} \label{sec:QdL-complete}

Distributed hybrid systems have several independent sources of undecidability: discrete dynamics, continuous dynamics, and structural/dimensional dynamics \cite{DBLP:conf/csl/Platzer10,DBLP:journals/lmcs/Platzer12b}.
\begin{theorem}[Incompleteness of {\QdL} \cite{DBLP:conf/csl/Platzer10,DBLP:journals/lmcs/Platzer12b}] \label{thm:QdL-incomplete}
  \index{incomplete!QdL@\QdL}
  The discrete fragment of \QdL, the continuous fragment of \QdL, and the fragment of \QdL with structural and dimension-changing dynamics are \dfn[axiomatize]{not effectively axiomatizable}, i.e., they have no sound and complete effective calculus, because natural numbers are definable\index{definable} in each of those fragments.
  \index{integer!arithmetic}
  \index{natural!number!definable}
\end{theorem}
\begin{proof}
  Incompleteness of the discrete fragment and of the continuous fragment follows from \rref{thm:dL-incomplete}.
  G\"odel's incompleteness theorem \cite{Goedel_1931} applies to the fragment with only structural and dimensional dynamics, because natural numbers are definable in that fragment by chains of links along the values of a function $p$, encoding zero by constant symbol~$z$:
  \[
  \textit{nat}(n) ~\lbisubjunct~ \ddiamond{\prepeat{(\ptest{n\neq z};~\pupdate{\umod{n}{p(n)}})}}{~n=z}
  .
  \index{_nat_@$\textit{nat}$}
  \]%
  For details on the characterization of addition and multiplication, we refer to the full proof \cite{DBLP:journals/lmcs/Platzer12b}.
  The idea behind addition is shown in \rref{fig:QdL-incomplete-dim}: 
  \begin{figure}[tbh]
    \centering
    \begin{tikzpicture}[every join/.style={trans}]
      \tikzstyle{node}=[draw,circle]
      \tikzstyle{trans}=[draw,<-,blue,thick,>=latex]
      \node[node,label=180:$z$] (z) {};
      \begin{scope}[start chain=n,xshift=1.2cm,yshift=1cm]
        \foreach \i in {1,...,4} {
          \node[node,on chain=n,join] {};
        }
        \node[node,on chain=n,join,label=0:$n$] {};
      \end{scope}
      \draw[trans] (z) -- (n-begin);
      \begin{scope}[start chain=m,xshift=1.2cm,yshift=0cm]
        \tikzstyle{node}+=[fill=blue!20]
        \foreach \i in {1,...,2}
          \node[node,on chain=m,join] {};
        \node[node,on chain=m,join,label=0:$m$] {};
      \end{scope}
      \draw[trans] (z) -- (m-begin) node[pos=0.7,above] {$p$} node[pos=0.7,below] {$p$};
      \tikzstyle{node}+=[cloud,cloud puffs=8]
      \tikzstyle{trans}+=[dashed]
      \begin{scope}[start chain=s,xshift=1.2cm,yshift=-1cm]
        \foreach \i in {1,...,5}
          \node[node,on chain=s,join] {};
       \tikzstyle{node}+=[fill=blue!20]
        \foreach \i in {1,...,2}
          \node[node,on chain=s,join] {};
        \node[node,on chain=s,join,label=90:${s=n+m}$] {};
      \end{scope}
      \draw[trans] (z) -- (s-begin);
      \node (desc) at (0,-1.3) {};
      \draw[decorate,decoration={brace,amplitude=8pt},yshift=1cm] (desc -| s-3) -- node[below=6pt] {new copy of $n$} (desc -| z);
      \draw[decorate,decoration={brace,amplitude=8pt},yshift=1cm] (desc -| s-end) -- node[below=6pt] {append new copy of $m$} (desc -| s-4);
    \end{tikzpicture}
    \caption{Characterization of $\naturals$ addition with $p$ links in dimensional dynamics}
    \label{fig:QdL-incomplete-dim}
  \end{figure}
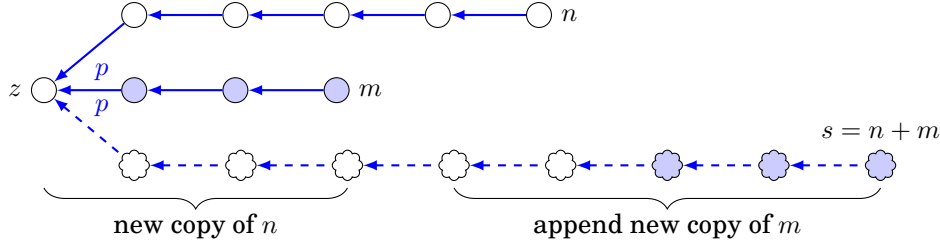%
  create a new chain of links along the values of $p$ by first creating exactly as many links as we can follow along $p$ when starting from $n$, and then continue creating exactly as many links as we can follow along $p$ when starting from $m$, instead.
  The number of links of the result $s$ is the sum of the respective numbers of links of $n$ and $m$.
\end{proof}

We have shown that the original \QdL calculus \cite{DBLP:conf/csl/Platzer10,DBLP:journals/lmcs/Platzer12b} is a sound and complete axiomatization of \QdL relative to the continuous fragment (\FOQD).
\FOQD is the \emph{first-order logic of quantified differential equations}, i.e., (many-sorted) first-order logic with real arithmetic augmented with formulas expressing properties of quantified differential equations, that is, \QdL formulas of the form \m{\dbox{\hevolvein{\lforall[C]{i}{\D{f(\vec{s})}=\theta}}{\ivr}}{F}}.
The dual formula \m{\ddiamond{\hevolvein{\lforall[C]{i}{\D{f(\vec{s})}=\theta}}{\ivr}}{F}} is expressible as \m{\lnot\dbox{\hevolvein{\lforall[C]{i}{\D{f(\vec{s})}=\theta}}{\ivr}}{\lnot F}}.
A combination of the original proof \cite{DBLP:conf/csl/Platzer10,DBLP:journals/lmcs/Platzer12b} and the proof of \rref{thm:dL-complete} can be used to show that the simplified \QdL calculus in \rref{fig:QdL} is sound and complete relative to \FOQD.
\begin{theorem}[Relative completeness of \QdL \cite{DBLP:conf/csl/Platzer10,DBLP:journals/lmcs/Platzer12b}] \label{thm:QdL-complete}
  \index{complete!relatively!QdL@\QdL}
  The \QdL calculus is a \emph{sound and complete axiomatization} of distributed hybrid systems relative to \FOQD, i.e.,
  every valid \QdL formula can be derived from \FOQD tautologies:
  \[
  \entails\phi ~~\text{iff}~~ \text{\upshape Taut}_{\FOQD} \infers \phi
  \]
\end{theorem}
\noindent

This central result shows that properties of distributed hybrid systems can be proven to exactly the same extent to which properties of quantified differential equations can be proven.
Proof-theoretically, the \QdL calculus completely lifts verification techniques for quantified continuous dynamics to distributed hybrid dynamics.
Even though distributed hybrid systems have numerous independent sources of undecidability, we have shown that all true \QdL formulas can be proven in our \QdL calculus, if only we manage to tame the complexity of the continuous dynamics.
Despite these new independent sources of undecidability, we have shown that \QdL can still be axiomatized completely relative to differential equations, only now they are quantified differential equations.

Another important consequence of this result is that decomposition is successful in taming the complexity of distributed hybrid systems.
The \QdL proof calculus is strictly compositional.
All \QdL axioms and proof rules prove logical formulas or properties of \QHPs by reducing them to structurally simpler \QdL formulas.
As soon as we understand that the distributed hybrid systems complexity comes from a combination of several simpler aspects, we can, hence, tame the system complexity by reducing it to analyzing the dynamical effects of simpler parts.
This decomposition principle is exactly how \QdL proofs can scale to interesting systems in practice.
The relative completeness theorem~\ref{thm:QdL-complete} gives the theoretical evidence why this principle works in general.
This is yet another illustration of our principle of multi-dynamical systems and even a proof that the decompositions behind the multi-dynamical systems approach are successful.

\subsection{Quantified Differential Invariants} \label{sec:Qdiffind}

Differential invariants (\rref{sec:diffind}) are the premier proof technique for proving properties of complicated differential equations, but they only work for hybrid systems.
Their generalization to quantified differential equations of distributed hybrid systems is called \emph{quantified differential invariant} \cite{DBLP:conf/hybrid/Platzer11}.
For quantified differential equations, one of the extra challenges is that the system does not have a fixed finite dimension but can be arbitrary-dimensional and even of evolving dimensions. Consequently, there is not even a finite vector space in which the local directions of the vector field of the differential equations can be described and checked.
Instead we need a criterion based on implicit properties of the local dynamics at uncountably many points in an essentially infinite-dimensional vector field of evolving dimensions.
We lift the syntactic derivation operator from \rref{sec:diffind} to first-order for that purpose.
\begin{definition}[Quantified derivation]
  The operator~\m{\DD{}\ignore{:\lterms{\Sigma{\cup}\D{\Sigma}}{V}\to\lterms{\Sigma{\cup}\D{\Sigma}}{V}}} that is defined as follows on terms is called \dfn[derivation!syntactic]{(quantified) syntactic (total) derivation}:
  \begin{subequations}
  \begin{align}
    \der{r} & = 0
      \hspace{0.85cm}\text{for}~r\in\rationals
    \label{eq:QDconstant}\\
    \der{x(s)} & =  \D{x(s)}
      \quad\text{for function symbol}~x:C\to\reals~\text{with}~C\neq\reals~\text{discrete}&
    \label{eq:QDpolynomial}\\
    \der{a+b} & = \der{a} + \der{b}
    \label{eq:QDadditive}\\
    \der{a-b} & = \der{a} - \der{b}
    \label{eq:QDsubtractive}\\
    \der{a\cdot b} & = \der{a}\cdot b + a\cdot\der{b}
    \label{eq:QDLeibniz}\\
    \der{a / b} & = (\der{a}\cdot b - a\cdot\der{b}) / b^2
    \label{eq:QDquotient}
  \end{align}
  \end{subequations}
  \index{differential!symbol}%
  We extend $\DD{}$ to first-order formulas $F$\internal{only without negative equalities\internal{also not in disguise}} in prefix disjunctive normal form
  as follows:
  \begin{align*}
    \der{\lforall[C]{i}{F}} &\,\mequiv\, \lforall[C]{i}{\der{F}}\\
    \der{\lexists[C]{i}{F}} &\,\mequiv\, \lforall[C]{i}{\der{F}}\\
    \der{F\land G} &\,\mequiv\, \der{F} \land \der{G}\notag\\
    \der{F\lor G} &\,\mequiv\, \der{F} \land \der{G}
    \notag
    \label{eq:QDformula}
    \\
    \internal{\der{\lnot F} &\,\mequiv\, \lnot\der{F}
    &&\text{for literals}~F\\}
    \der{a\geq b} &\,\mequiv\, \der{a} \geq \der{b}
    \quad\text{accordingly for \m{<,>,\leq,=}}
    .
    \notag
  \end{align*}
\end{definition}
The \emph{quantified differential induction} rule is a natural induction principle for quantified differential equations:
\begin{center}
    \begin{calculus}
      \cinferenceRule[Qdiffind|DI]{quantified differential induction}
      {\linferenceRule[sequent]
        {\lsequent{\ivr}{
              \dbox{\pupdate{\lforall[C]{i}{\umod{\D{f(\vec{i})}}{\theta}}}}{\der{F}}}}
        {\lsequent{F}{
              \dbox{\hevolvein{\lforall[C]{i}{\D{f(\vec{i})}=\theta}}{\ivr}}{F}}}
       }{}
    \end{calculus}
\end{center}

For a formula $F$ if we can prove the premise of rule \irref{Qdiffind}, i.e., that, after a differential substitution \m{\dbox{\pupdate{\lforall[C]{i}{\umod{\D{f(\vec{i})}}{\theta}}}}{}}, the total derivative \m{\der{F}} is valid in the evolution domain region $\ivr$, then the conclusion of \irref{Qdiffind} is valid, i.e., the system always stays in region $F$ when it starts in $F$ (left assumption in conclusion).
It is important that we add the quantified assignment \m{\pupdate{\lforall[C]{i}{\umod{\D{f(\vec{i})}}{\theta}}}} for quantified differential symbol $\D{f(\vec{i})}$ as a \emph{quantified differential substitution} in the premise, because, otherwise, the premise of \irref{diffind} is not a logical formula that would have a well-defined semantics when evaluated in a state.
Unlike $F$, the total derivative \m{\der{F}} contains differential function symbols like $\D{f(i)}$, which do not have a semantics in isolated states but only along a flow (an insightful differential semantics can be defined but is beyond the scope of this article).
The quantified assignment defines a value for those differential function symbols, which has a well-defined correspondence to the local dynamics of quantified differential equations based on a differential substitution property \cite[Lemma 2]{DBLP:conf/hybrid/Platzer11}.
The conjunctive definition of \m{\der{F\lor G}} is crucial (recall\rref{sec:diffind}) and so is the universal definition of \m{\der{\lexists[C]{i}{F}}}.
Differential cuts \irref{diffcut} (see \rref{sec:diffcut}) generalize to quantified differential equations immediately \cite{DBLP:conf/hybrid/Platzer11} and are as crucial as they are for hybrid systems.
A generalization of the derivation lemma (\rref{lem:derivationLemma}) to quantified differential equations is a key argument to relate analytic differentiation and syntactic derivations \cite[Lemma 1]{DBLP:conf/hybrid/Platzer11}.
For details, see previous work \cite{DBLP:conf/hybrid/Platzer11}.

\subsection{Implementation and Applications} \label{sec:KeYmaeraD}

Quantified differential dynamic logic and a sequent calculus variation of its proof calculus \cite{DBLP:conf/csl/Platzer10,DBLP:journals/lmcs/Platzer12b}, including quantified differential invariants and quantified differential cuts \cite{DBLP:conf/hybrid/Platzer11} have been implemented in the LCF-style tactics-based theorem prover \KeYmaeraD \cite{DBLP:conf/icfem/RenshawLP11}.\footnote{Available at \url{http://symbolaris.com/info/KeYmaeraD.html}}
The name \KeYmaeraD extends \KeYmaera with a D for distributed, which represents the fact that \KeYmaeraD can prove distributed hybrid systems and the fact that \KeYmaeraD is implemented with an architecture that supports distributed and parallel proving of both distributed and non-distributed hybrid systems.

Quantified differential dynamic logic and \KeYmaeraD have been used successfully to prove collision-freedom properties for distributed car controllers for arbitrarily many cars on arbitrarily many lanes on straight highways \cite{DBLP:conf/csl/Platzer10,DBLP:journals/lmcs/Platzer12b,DBLP:conf/fm/LoosPN11,DBLP:conf/icfem/RenshawLP11} and safe separation properties for roundabout collision avoidance maneuvers for arbitrarily many aircraft with dynamic appearance of aircraft during collision avoidance \cite{DBLP:conf/hybrid/Platzer11}.
They have also been used to prove properties of advanced flight control protocols and medical robotic applications.

\section{Stochastic Differential Dynamic Logic for Stochastic Hybrid Systems} \label{ch:SdL}
\makeatletter
\renewcommand*{\@dL@interpretationformatfix}[2]{{#2}^{#1}}\makeatother
\renewcommand{\precond}{I}
\renewcommand{\inv}{\mathsf{f}}
\renewcommand{\postcond}{g}

\newcommand{\Io}[1][t]{\SDLint}
\renewcommand{\stdI}[1][t]{\SDLint[state=X_{#1}]}
\let\I\stdI
\newcommand{\Iw}[1][t]{\SDLint[state=X_{#1}(\omega)]}
\renewcommand{\If}{\SDLint[flow=X]}
\newcommand{\Ix}[1][]{\iconcat[state=x,time=#1]{\stdI}}
\newcommand{\IZ}[1][]{\iconcat[state=Z,time=#1]{\stdI}}
\newcommand{\IZw}[1][]{\iconcat[state=Z(\omega),time=#1]{\stdI}}
\newcommand{\IXc}[1][t]{\iconcat[state=\check{X}_{#1},time=#1]{\stdI}}
\newcommand{\IZS}[1][]{\iconcat[state=S,time=#1]{\stdI}}
\newcommand{\Iy}[1][]{\iconcat[state=y,time=#1]{\I}}%
\newcommand{\IY}[1][]{\iconcat[state=Y,time=#1]{\I}}%
\newcommand{\IYw}[1][]{\iconcat[state=Y(\omega),time=#1]{\I}}%
\newcommand{\spair}[2]{#1\otimes#2}%

In this section, we study \emph{stochastic differential dynamic logic} \SdL \cite{DBLP:conf/cade/Platzer11}, the \emph{logic of stochastic hybrid systems}, i.e., systems that combine stochastic dynamics with the discrete and continuous dynamics of hybrid systems.

  In the previous sections, we have seen that logic of dynamical systems is a powerful tool for analyzing and verifying dynamical systems, including hybrid systems (\rref{ch:dL}) and distributed hybrid systems (\rref{ch:QdL}).
  Some applications also exhibit a stochastic behavior, however, either because of fundamental properties of nature, uncertain environments, or simplifications to overcome complexity.
  Uncertainties can sometimes be modeled well by taking a nondeterministic perspective where anything could happen and we are not concerned with how probable which outcome is.
  Nondeterminism can be used to model uncertainty in hybrid systems (\rref{ch:dL}) and distributed hybrid systems (\rref{ch:QdL}).
  That is useful, for example, if a car is uncertain about whether the car in front of it will brake or accelerate, so that the follower car has to be prepared to handle either choice safely.
  In other situations, however, stochastic models are needed to model uncertainty.
  For example, a nondeterministic model of a lossy communication channel would consider it possible for a message to arrive and possible for a message to disappear during transmission.
  But then one possible resolution of the nondeterminisms would cause all messages to disappear from all communication attempts, which is certainly possible, just extremely unlikely, except in adversarial situations.
  Whenever our analysis would be too imprecise with a nondeterministic view or whenever we have good stochastic models, probabilistic resolutions of uncertainties are more appropriate.
  This is what we study in this section.
  
  Discrete probabilistic systems have been studied successfully using logic \cite{DBLP:journals/jcss/Kozen81,DBLP:journals/jcss/Kozen85,DBLP:journals/jcss/FeldmanH84,McIverMorgan04}.
  In this section, we present our dynamic logic of stochastic hybrid systems \cite{DBLP:conf/cade/Platzer11}, i.e., systems with interacting discrete, continuous, and stochastic dynamics.
  Our results indicate that logic is a promising tool for understanding stochastic hybrid systems and can help taming some of their complexity.

Classical logic is about boolean yes/no truth.
That makes it tricky to use logic for systems with stochastic effects.
Logic has reached out into probabilistic extensions at least for discrete programs \cite{DBLP:journals/jcss/Kozen81,DBLP:journals/jcss/Kozen85,DBLP:journals/jcss/FeldmanH84,McIverMorgan04} and for first-order logic over a finite domain \cite{DBLP:journals/ml/RichardsonD06}.
Logic has been used for the purpose of specifying system properties in model checking finite Markov chains \cite{DBLP:journals/sttt/YounesKNP06} and probabilistic timed automata \cite{DBLP:journals/iandc/KwiatkowskaNSW07}.

Stochastic hybrid systems \cite{BujorianuL06,Cassandras2006,DBLP:conf/hybrid/HuLS00,DBLP:conf/cade/Platzer11} are more general systems with interacting discrete, continuous, and stochastic dynamics.
There is not just one canonical way to add stochastic behavior to a system model.
Stochasticity might be restricted to the discrete dynamics, as in piecewise deterministic Markov decision processes \cite{DBLP:journals/roystats/Davis84}, restricted to the continuous and switching behavior as in switching diffusion processes \cite{DBLP:journals/jcopt/GhoshAM97}, or allowed in many parts as in so-called General Stochastic Hybrid Systems; see \cite{BujorianuL06,Cassandras2006} for an overview.
Several different forms of combinations of probabilities with hybrid systems and continuous systems have been considered, both for model checking \cite{Cassandras2006,DBLP:journals/tsmc/KoutsoukosR08,DBLP:journals/jlp/FranzleTE10} and for simulation-based validation \cite{DBLP:conf/hybrid/MeseguerS06,DBLP:conf/hybrid/ZulianiPC10}.

We have introduced a very different approach \cite{DBLP:conf/cade/Platzer11} that is based on logic for dynamical systems.
We consider logic and proofs for stochastic hybrid systems\footnote{Note that there is a specific class of models called Stochastic Hybrid Systems \cite{DBLP:conf/hybrid/HuLS00}. We do not mean this specific model in the narrow sense but refer to stochastic hybrid systems as the broader class of systems that share discrete, continuous, and stochastic dynamics.} to transfer the success that logic has had in other domains.
Our approach is partially inspired by probabilistic PDL \cite{DBLP:journals/jcss/Kozen85} and by barrier certificates for continuous dynamics \cite{DBLP:journals/tac/PrajnaJP07}.
We follow the arithmetical view that Kozen identified as suitable for probabilistic logic \cite{DBLP:journals/jcss/Kozen85}.
Simple probabilistic effects can already be encoded in real variables of \dL and \QdL, but general stochastic dynamics requires the approach presented in this section.

Classical analysis is provably inadequate \cite{KloedenPlaten2010}\internal{\cite[Ch I.8]{Protter10}} for analyzing even simple continuous stochastic processes.
We heavily draw on both stochastic calculus and logic.
It is not possible to present all mathematical background exhaustively here.
We provide basic definitions and intuition and refer to the literature for more details and proofs \cite{DBLP:conf/cade/Platzer11,KaratzasShreve,Oksendal07,KloedenPlaten2010}.

We show the model of \emph{stochastic hybrid programs} (SHPs) \cite{DBLP:conf/cade/Platzer11} that combine discrete stochastic dynamics and stochastic differential equations for continuous stochastic dynamics, and we define a compositional semantics of SHP runs in terms of stochastic processes.
We have proved that the semantic processes are adapted, almost surely have \cadlag paths, and that their natural stopping times are Markov times.
We have introduced \emph{stochastic differential dynamic logic} (\SdL) for specifying and verifying properties of SHPs \cite{DBLP:conf/cade/Platzer11}.
We define a semantics and have proved that the semantics is measurable such that probabilities are well-defined and probabilistic questions become meaningful.
We present proof rules for \SdL and have proved their soundness \cite{DBLP:conf/cade/Platzer11}.
\SdL makes the rich semantical complexity and deep theory of stochastic hybrid systems accessible in a simple syntactic language.
This makes the verification of stochastic hybrid systems possible with elementary syntactic proof principles.

We first briefly recall the basics of stochastic processes and stochastic differential equations (\rref{sec:SDE}).
Then we explain the system model of stochastic hybrid programs that \SdL provides for modeling stochastic hybrid systems (\rref{sec:SHP}) and define their semantics.
We define the terms and logical formulas that \SdL provides for specification and verification (\rref{sec:SdL-formula}), show measurability results (\rref{sec:SdL-measurable}) and then provide reasoning principles, axioms, and proof rules for verifying \SdL formulas (\rref{sec:SdL-calculus}).
We then show soundness theorems (\rref{sec:SdL-sound}) and investigate proof rules for stochastic differential equations (\rref{sec:Sdiffind}).

\subsection{Preliminaries: Stochastic Differential Equations} \label{sec:SDE}
We fix a dimension $d\in\naturals$ for the Euclidean state space $\reals^d$ equipped with its \dfn[Borel!$\sigma$-algebra]{Borel $\sigma$-algebra} $\mathcal{B}$, i.e., the $\sigma$-algebra generated by all open subsets.
A \dfn{$\sigma$-algebra} on a set $\Omega$ is a nonempty set \m{\mathcal{F}\subseteq2^\Omega} that is closed under complement (\(E\in\mathcal{F}\) implies \(\Omega\setminus E\in\mathcal{F}\)) and countable union (\(E_i\in\mathcal{F}\) implies \(\cupfold_{i=1}^\infty E_i\in\mathcal{F}\)).
Even though it could also be constructed, we axiomatically fix a \dfn[probability!space]{probability space} \m{(\Omega,\mathcal{F},P)} with a $\sigma$-algebra $\mathcal{F}\subseteq2^\Omega$ of events on space $\Omega$ and a probability measure $P$ on $\mathcal{F}$.
Function \m{P:\mathcal{F}\to[0,1]} is a \dfn[probability!measure]{probability measure} on $\mathcal{F}$ if $P$ is countable additive (i.e., \(P(\cupfold_{i=1}^\infty E_i)=\sum_{i=1}^\infty P(E_i)\) when \(E_i\cap E_j=\emptyset\) for all $i\neq j$) and $P\geq0, P(\Omega)=1$.
We assume the probability space has been \emph{completed}, i.e., every subset of a null set (i.e., $P(A)=0$) is measurable.
A property holds $P$-\dfn[almost~surely]{almost surely} (\emph{a.s.}) if it holds with probability 1.
A \dfn{filtration} is a family \m{(\mathcal{F}_t)_{t\geq0}} of $\sigma$-algebras that is increasing, i.e., \m{\mathcal{F}_s \subseteq \mathcal{F}_t} for all \m{s<t}.
Intuitively, $\mathcal{F}_t$ are the events that can be discriminated at time $t$.
We always assume a filtration \m{(\mathcal{F}_t)_{t\geq0}} that has been completed to include all null sets and that is \dfn{right-continuous}, i.e., \m{\mathcal{F}_t=\capfold_{u>t} \mathcal{F}_u} for all $t$.
We generally assume the compatibility condition that $\mathcal{F}$ coincides with the $\sigma$-algebra \m{\mathcal{F}_\infty:=\sigma\left(\cupfold_{t\geq0}\mathcal{F}_t\right)}, i.e., the $\sigma$-algebra generated by all $\mathcal{F}_t$.

For a $\sigma$-algebra $\Sigma$ on a set $D$ and the Borel $\sigma$-algebra $\mathcal{B}$ on $\reals^d$, function \m{f:D\to\reals^d} is \dfn[measurable!function]{measurable} iff \m{f^{-1}(B) \in \Sigma} for all \m{B\in\mathcal{B}} (or, equivalently, for all open \m{B\subseteq\reals^d}).
An $\reals^d$-valued \dfn[random!variable]{random variable} is an $\mathcal{F}$-measurable function \m{X:\Omega\to\reals^d}.
All sets and functions definable in first-order logic over real arithmetic are Borel-measurable \cite{tarski_decisionalgebra51}.
A \dfn[stochastic!process]{stochastic process} $X$ is a collection \m{\{\iget[state]{\I}\}_{t\in T}} of $\reals^d$-valued random variables $\iget[state]{\I}$ indexed by some set $T$ for time.
That is, \m{X:T\times\Omega\to\reals^d} is a function such that at all $t\in T$, \m{X_t=X(t,\cdot):\Omega\to\reals^d} is a random variable.
Process $X$ is \dfn{adapted} to filtration $(\mathcal{F}_t)_{t\geq0}$ if $X_t$ is $\mathcal{F}_t$-measurable for each $t$.
That is, the process does not depend on future events.
We consider only adapted processes (which can be ensured, e.g., by using the completion of the natural filtration of a process or the completion of the optional $\sigma$-algebra for $\mathcal{F}$, see \cite{KaratzasShreve}).
A process $X$ is \emph{{\cadlag}} iff its \dfn[path]{paths} \m{t\mapsto X_t(\omega)} (for each $\omega\in\Omega$) are \cadlag a.s., i.e., \dfn{right-continuous} (\m{\lim_{s\searrow t} X_s(\omega)=X_t(\omega)}) and have \dfn[left~limit]{left limits} (\m{\lim_{s\nearrow t} X_s(\omega)} exists).

We further need an $e$-dimensional \emph{Brownian motion} $W$ \cite{KaratzasShreve,Oksendal07,KloedenPlaten2010}, that is, $W$ is a stochastic process starting at 0 ($W_0=0$) that is almost surely continuous and has independent increments that are normally distributed with mean 0 and variance equal to the time difference, i.e., \m{W_t-W_s \sim \mathcal{N}(0,t-s)}.
Brownian motion is mathematically extremely complex.
Its paths are a.s. continuous everywhere but a.s. differentiable nowhere and a.s. of unbounded variation.
Thus, Brownian motion is a.s. not of finite variation, hence, standard integral notions are inapplicable.
Brownian motion is also a.s. nonmonotonic on every interval.
Intuitively, $W$ can be understood as the limit of a random walk.
We denote the Euclidean vector norm by $|x|$ and use the Frobenius norm \m{|\sigma|:=\sqrt{\sum_{i,j}\sigma_{ij}^2}} for matrices $\sigma\in\reals^{d\times e}$.

We use stochastic differential equations \cite{Oksendal07,KloedenPlaten2010} to describe stochastic continuous system dynamics.
They are like ordinary differential equations but have an additional diffusion term that varies the state stochastically.
Stochastic differential equations are of the form 
\m{d{X_t} = b(X_t)dt + \sigma(X_t)d{W_t}}.
We consider It\=o stochastic differential equations, whose solutions are defined by the stochastic It\=o integral \cite{Oksendal07,KloedenPlaten2010}, which is a stochastic process.
\begin{figure}[tb]
  \includegraphics[width=0.5\columnwidth]{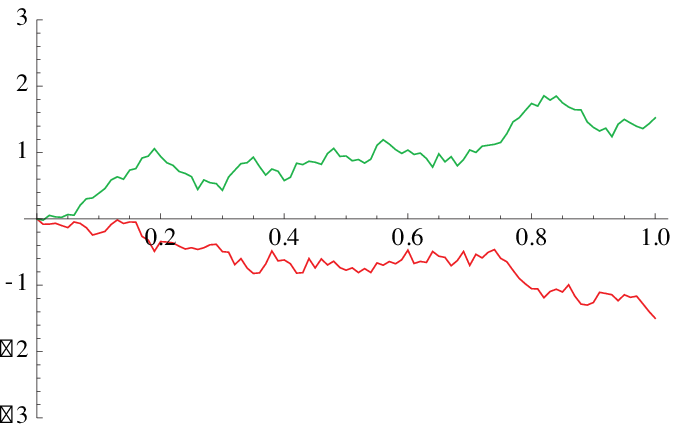}%
  \includegraphics[width=0.5\columnwidth]{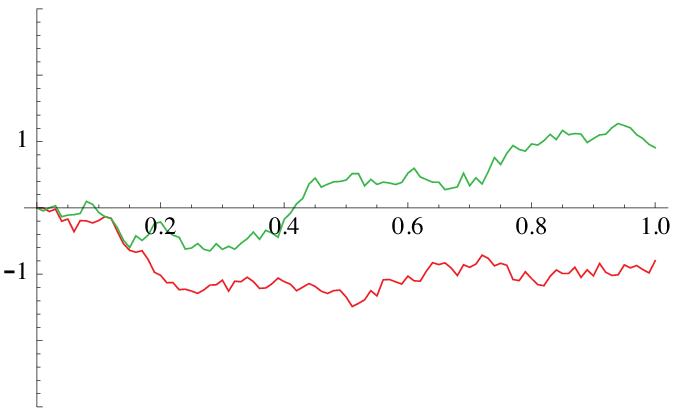}
  \caption{Sample paths with $b=1,\sigma=1$ (top) and $b=0,\sigma=1$ (bottom)}
  \label{fig:brownian-motion}
\end{figure}
Like in an ordinary differential equation, the drift coefficient $b(X_t)$ determines the deterministic part of how $X_t$ changes systematically over time as a function of its current value.
As a function of $X_t$, the diffusion coefficient $\sigma(X_t)$ determines the stochastic influence by integration with respect to the Brownian motion process $W_t$.
See \rref{fig:brownian-motion} for sample paths.
Ordinary differential equations are retained for $\sigma=0$.
We focus on the time-homogeneous case, where $b$ and $\sigma$ are time-independent, because time could be added as an extra state variable if needed.

\begin{definition}[Stochastic differential equation] \label{def:SDE}
  A stochastic process \(X:[0,\infty)\times\Omega\to\reals^d\) solves the (It\=o) \dfn[stochastic!differential~equation]{stochastic differential equation}
\begin{equation}
  d{X_t} = b(X_t)dt + \sigma(X_t)d{W_t}
  \label{eq:SDE}
\end{equation}
with \m{X_0=Z}, if
\m{X_t = Z + \int b(X_t)dt + \int \sigma(X_t)d{W_t}},
where \m{\int \sigma(X_t)d{W_t}} is an It\=o integral process \cite{Oksendal07,KloedenPlaten2010}.
We assume \m{b:\reals^d\to\reals^d} and \m{\sigma:\reals^d\to\reals^{d\times e}} to be \ignore{$\mathcal{L}^2$-}measurable and locally Lipschitz-continuous, i.e.,
for all $N$ there is a $C$ such that for all $x,y$ with \(|x|,|y|\leq N\):
\[
|b(x)-b(y)|\leq C|x-y| ~\text{and}~ |\sigma(x)-\sigma(y)|\leq C|x-y|
\]
\end{definition}
As an integral of an a.s. continuous process, solution $X$ has a.s. continuous paths \cite{Oksendal07}.
Local Lipschitz-continuity guarantees that the a.s. continuous solution $X$ is pathwise unique \cite[Ch 4.5]{KloedenPlaten2010}\internal{Lem 4.5.2 and text after Theorem 4.5.3} and enables us to compute the infinitesimal generator of $X$ from its differential generator (see \rref{sec:Sdiffind}).
Process $X$ is a strong Markov process for each initial value $x$ \cite[Theorem 7.2.4]{Oksendal07}.
We focus on the time-homogeneous case, where $b$ and $\sigma$ are time-independent, because time could be added as an extra state variable.

\subsection{Stochastic Hybrid Programs} \label{sec:SHP}

As a system model for stochastic hybrid systems, we have introduced stochastic hybrid programs (SHPs) \cite{DBLP:conf/cade/Platzer11}.
SHPs combine stochastic differential equations for describing the stochastic continuous system dynamics with program operations to describe the discrete stochastic choices, discrete switching, and jumps.
These primitive dynamics can be combined programmatically in flexible ways.
All basic terms in stochastic hybrid programs and stochastic differential dynamic logic are polynomial terms built over real-valued variables and rational constants.
Extensions to rational functions are possible.

\begin{definition}[Stochastic hybrid program]
\emph{Stochastic hybrid programs (SHPs)} are formed by the following grammar (where $x_i$ is a variable, $x$ a vector of variables, $\theta$ a basic term, $b$ a vector of basic terms, $\sigma$ a matrix of basic terms, $\ivr$ is a quantifier-free first-order real arithmetic formula, $\lambda,\nu\geq0$ are rational numbers):
\[
\alpha,\beta ~\bebecomes~ 
\pupdate{\pumod{x_i}{\theta}}
\alternative \prandom{x_i}
\alternative \ptest{\ivr}
\alternative \hevolvein{d{x} = b dt + \sigma d{W}}{\ivr}
\alternative\\ \lambda\alpha \pplus \nu\beta
\alternative \alpha;\beta
\alternative \prepeat{\alpha}
\]
\end{definition}
\emph{Assignment} \m{\pupdate{\pumod{x_i}{\theta}}} deterministically assigns term $\theta$ to variable $x_i$ instantaneously.
\emph{Random assignment} \m{\prandom{x_i}} randomly updates variable $x_i$, but unlike in classical dynamic logic \cite{DBLP:conf/focs/Pratt76} and differential dynamic logic, we assume a probability distribution for $x$.
As one example for a probability distribution, we consider uniform distribution in the interval [0,1], but other distributions can be used as long as they are computationally tractable, e.g., definable in first-order real arithmetic.

Most importantly, \m{\hevolvein{d{x} = b dt + \sigma d{W}}{\ivr}} represents a \emph{stochastic continuous evolution} along a stochastic differential equation, restricted to the evolution domain region $\ivr$, i.e., the stochastic process will stop when it leaves $\ivr$.
Unlike in \dL and \QdL, where we take a nondeterministic view, stochastic continuous evolutions do not stop prematurely, but exactly when they leave $\ivr$.
The time when evolutions stop is still random, but now described by the stochastic continuous evolution and its evolution domain constraint instead of by nondeterminism.
In particular, $\ivr$ represents control decisions when to interrupt the continuous stochastic process.
Because of that, we can show that the random variable describing when the stochastic continuous evolution stops is a Markov time (\rref{thm:cadlag-adapted-Markov}).
We assume that \m{\hevolve{d{x} = b dt + \sigma d{W}}} satisfies the assumptions of stochastic differential equations from \rref{def:SDE}.

\emph{Test} \m{\ptest{\ivr}} represents a stochastic process that fails (disappears into an absorbing state) if $\ivr$ is not satisfied yet continues unmodified otherwise.
\emph{Linear combination} \m{\lambda\alpha\pplus\nu\beta} evolves like $\alpha$ with probability $\lambda$ and like $\beta$ otherwise.
For conceptual simplicity, we assume \m{\lambda+\nu=1}, but other linear combinations are possible when taking care to ensure that the result still gives probabilities \cite{DBLP:journals/jcss/Kozen85}, e.g., when using complementary tests in the definition 
\(\pif{\ivr}{\alpha}{\beta} \mequiv (\ptest{\ivr};\alpha) \pplus (\ptest{\lnot\ivr};\beta)\).
It is possible to extend this to the case where $\lambda,\mu$ are terms.
Linear combination alias probabilistic choice \m{\lambda\alpha\pplus\nu\beta} is the counterpart of the nondeterministic choice \m{\pchoice{\alpha}{\beta}} from \rref{ch:dL}, but gives information about the probability with which the respective choices are taken.
\emph{Sequential composition} $\alpha;\beta$ and \emph{repetition} $\prepeat{\alpha}$ work similarly to differential dynamic logic, except that they combine SHPs instead of non-stochastic HPs.

The semantics of a SHP is the stochastic process that it generates.
The semantics $\iaccess[\alpha]{\Io}$ of a SHP $\alpha$ consists of a function \m{\iaccess[\alpha]{\Io}:(\Omega\to\reals^d)\to([0,\infty)\times\Omega\to\reals^d)} that maps any $\reals^d$-valued random variable $\iget[state]{\IZ}$ describing the initial state to a stochastic process $\iaccess[\alpha]{\IZ}$ together with a function \m{\istop[\alpha]{\Io}:(\Omega\to\reals^d)\to(\Omega\to\reals)} that maps any $\reals^d$-valued random variable $\iget[state]{\IZ}$ describing the initial state to a (random) stopping time $\istop[\alpha]{\IZ}$ indicating when to stop $\iaccess[\alpha]{\IZ}$.
Often, an $\mathcal{F}_0$-measurable random variable $\iget[state]{\IZ}$ or deterministic state is used to describe the initial state.
We assume independence of $\iget[state]{\IZ}$ from subsequent stochastic processes like Brownian motions occurring in the definition of \m{\iaccess[\alpha]{\IZ}}.
For an $\reals^d$-valued random variable $Z$, we denote the stochastic process \m{\pproc{Z}:\{0\}\times\Omega\to\reals^d; (0,\omega)\mapsto \pproc{Z}_0(\omega):=Z(\omega)} that is stuck at $Z$ by $\pproc{Z}$.
We write $\pproc{\iget[state]{\Ix}}$ for the random variable $Z$ that is a deterministic state $Z(\omega):=\iget[state]{\Ix}$ for all $\omega\in\Omega$.
We write \m{\iaccess[\alpha]{\Ix}} and \m{\istop[\alpha]{\Ix}} for \m{\iaccess[\alpha]{\IZ}} and \m{\istop[\alpha]{\IZ}} in that case.

In order to simplify notation, we assume that all variables are uniquely identified by an index, i.e., the only occurring variables are $x_1,x_2,\dots,x_d$.
We write \m{\imodels{\IZw}{\ivr}} if state $\iget[state]{\IZ}(\omega)$ satisfies first-order real arithmetic formula $\ivr$ and \m{\inonmodels{\IZw}{\ivr}} otherwise.
In the semantics we use a family of random variables $\{U_i\}_{i\in I}$ that are distributed uniformly in $[0,1]$ and independent of other $U_j$ and of all other random variables and stochastic processes in the semantics.
Hence, \m{U} satisfies
    \m{P(\{\omega\in\Omega \with U(\omega)\leq s\}) = \int_{-\infty}^s \charf{[0,1]}dt}
    with the usual extensions to other Borel subsets of $\reals$.
To describe this situation, we just say that ``\m{U\sim\mathcal{U}(0,1)} is i.i.d. (independent and identically distributed)'', meaning that $U$ is furthermore independent of all other random variables and stochastic processes in the semantics.
We denote the \dfn[characteristic!function]{characteristic function} of a set $S$ by \m{\charf{S}}, defined by \m{\charf{S}(x):=1} if \m{x\in S} and \m{\charf{S}(x):=0} if \m{x\not\in S}.

\begin{definition}[Process semantics of SHPs] \label{def:SHP-semantics}
  The \dfn[semantics!SHP]{semantics} of SHP $\alpha$ is defined by  
  \begin{align*}
    \iaccess[\alpha]{\Io}:&(\Omega\to\reals^d)\to([0,\infty)\times\Omega\to\reals^d); \iget[state]{\IZ}\mapsto\iaccess[\alpha]{\IZ} = (\iaccess[\alpha]{\IZ[t]})_{t\geq0}\\
    \istop[\alpha]{\Io}:&(\Omega\to\reals^d)\to(\Omega\to\reals\cup\{\infty\}); \iget[state]{\IZ}\mapsto\istop[\alpha]{\IZ}
  \end{align*}
  These functions are inductively defined for random variable \m{\iget[state]{\IZ}:(\Omega\to\reals^d)} by
\begin{enumerate}
  \item \label{case:assign}
    \m{\iaccess[\pupdate{\umod{x_i}{\theta}}]{\IZ} = \pproc{\iget[state]{\IY}}}
    where \m{\iget[state]{\IYw}_i=\ivaluation{\IZw}{\theta}}
    and \m{\iget[state]{\IY}_j=\iget[state]{\IZ}_j} for all $j\neq i$,
    and \m{\istop[\pupdate{\umod{x_i}{\theta}}]{\IZ} = 0}.
  \item \label{case:random}
    \m{\iaccess[\prandom{x_i}]{\IZ} = \pproc{U}}
    where \m{U_j=\iget[state]{\IZ}_j} for all $j\neq i$,
    and \m{U_i\sim\mathcal{U}(0,1)} is i.i.d. and $\mathcal{F}_0$-measurable.
    Further, \m{\istop[\prandom{x_i}]{\IZ} = 0}.
  \item \label{case:test}
    \m{\iaccess[\ptest{\ivr}]{\IZ} = \pproc{\iget[state]{\IZ}}} on the event \m{\{\iget[state]{\IZ}\models\ivr\}}
    and \m{\istop[\ptest{\ivr}]{\IZ} = 0} (on all events $\omega\in\Omega$).
    Note that \m{\iaccess[\ptest{\ivr}]{\IZ}} is not defined on the event \m{\{\iget[state]{\IZ}\nonmodels\ivr\}}.
\item \label{case:sde}
  \(\iaccess[\hevolvein{d{x} = b dt + \sigma d{W}}{\ivr}]{\IZ}\) is the stochastic process \m{\iget[flow]{\If}:[0,\infty)\times\Omega\to\reals^d} that solves the (It\=o) stochastic differential equation 
\m{d{\iget[state]{\I}} = \ivaluation{\I}{b}dt + \ivaluation{\I}{\sigma} d{B_t}} with \m{\iget[state]{\I[0]}=\iget[state]{\IZ}} on the event \m{\{\iget[state]{\IZ}\models\ivr\}},
  where $B_t$ is a fresh $e$-dimensional Brownian motion if $\sigma$ has $e$ columns.
  We assume that $\iget[state]{\IZ}$ is independent of the $\sigma$-algebra generated by $(B_t)_{t\geq0}$.
  Further, \m{\istop[\hevolvein{d{x} = b dt + \sigma d{W}}{\ivr}]{\IZ} = \inf \{t\geq0 \with \iget[state]{\I[t]} \not\in\ivr \}}.
    Note that \m{X} is not defined on the event \m{\{\iget[state]{\IZ}\nonmodels\ivr\}}.
  \item \label{case:pplus}%
    \begin{minipage}{6cm}
    \begin{align*}
    \iaccess[\lambda\alpha \pplus \nu\beta]{\IZ}
    &=
    \charf{U\leq\lambda}\iaccess[\alpha]{\IZ} + \charf{U>\lambda}\iaccess[\beta]{\IZ}
    =
    \begin{cases}
    \iaccess[\alpha]{\IZ} &\text{on the event}~\{U\leq\lambda\}\\
    \iaccess[\beta]{\IZ} &\text{on the event}~\{U>\lambda\}
    \end{cases}
    \\
    \istop[\lambda\alpha \pplus \nu\beta]{\IZ} &%
    =\charf{U\leq\lambda}\istop[\alpha]{\IZ} + \charf{U>\lambda}\istop[\beta]{\IZ}
    \end{align*}
    \end{minipage}

    where  \m{U\sim\mathcal{U}(0,1)} is i.i.d. and $\mathcal{F}_0$-measurable.
  \item \label{case:compose}%
    \begin{minipage}{6cm}
    \def\alphastop{\istop[\alpha]{\IZ}}%
    \def\Iint{\iaccess[\alpha]{\IZ[\alphastop]}}%
    \newcommand{\Icont}{\SDLint[state=\Iint,time=t-\alphastop]}%
    \begin{align*}
    \iaccess[\alpha;\beta]{\IZ[t]}
    &=
    \begin{cases}
      \iaccess[\alpha]{\IZ[t]} &\text{on the event}~ \{\istop[\alpha]{\IZ}>t\}\\
      \iaccess[\beta]{\Icont} &\text{on the event}~ \{\istop[\alpha]{\IZ}\leq t\}
    \end{cases}
    \\
    \istop[\alpha;\beta]{\IZ} &= \istop[\alpha]{\IZ} + \istop[\beta]{\Icont}
    \end{align*}
    \end{minipage}
  \item \label{case:repeat}%
    \begin{minipage}{6cm}
    \begin{align*}
    \iaccess[\prepeat{\alpha}]{\IZ[t]}
    &=
    \iaccess[\alpha^n]{\IZ[t]}
    ~\text{on the event}~ \{\istop[\alpha^n]{\IZ}>t\}
    \\
    \istop[\prepeat{\alpha}]{\IZ} &= \lim_{n\to\infty} \istop[\alpha^n]{\IZ}
    \end{align*}
    \end{minipage}
    
    where $\alpha^0 \mequiv \ptest{\ltrue}$, $\alpha^1 \mequiv \alpha$, and $\alpha^{n+1} \mequiv \alpha; \alpha^n$.
\end{enumerate}
\end{definition}
  For \rref{case:repeat}, note that \m{\istop[\alpha^n]{\IZ}} is monotone in $n$, hence the limit \m{\istop[\prepeat{\alpha}]{\IZ}} exists and is finite if the sequence is bounded. The limit is $\infty$ otherwise.
    Note that \m{\iaccess[\prepeat{\alpha}]{\IZ[t]}} is independent of the choice of $n$ on the event \m{\{\istop[\alpha^n]{\IZ}>t\}} (but not necessarily independent of $n$ on the event \m{\{\istop[\alpha^n]{\IZ}\geq t\}}, because $\alpha$ might start with a jump after $\alpha^n$).
    Observe that \m{\iaccess[\prepeat{\alpha}]{\IZ[t]}} is not defined on the event \m{\{\lforall{n}{\istop[\alpha^n]{\IZ}\leq t}\}}, which happens, e.g., for Zeno executions violating divergence of time.
    It would still be possible to give a semantics in this case, e.g., at \m{t=\istop[\alpha^n]{\IZ}}, but we do not gain much from introducing those technicalities.
    Note that the choice of inequalities in cases~\ref{case:compose} and~\ref{case:repeat} is important to obtain a \cadlag process.

In the semantics of \m{\iaccess[\alpha]{\IZ}}, time is allowed to end.
We explicitly consider \m{\iaccess[\alpha]{\IZ[t]}} as not defined for a realization $\omega$ if a part of this process is not defined, because of failed tests in $\alpha$.
The process is explicitly not defined when \m{\istop[\alpha]{\IZ}<t}.
Explicitly being not defined can be viewed as being in a special absorbing state that can never be left again, as in killed processes.
The stochastic process \m{\iaccess[\alpha]{\IZ}} is only intended to be used until time \m{\istop[\alpha]{\IZ}}.
We stop using \m{\iaccess[\alpha]{\IZ}} after time \m{\istop[\alpha]{\IZ}}, which is a random variable.

\begin{example}[Jumping rotational Brownian motion] \label{ex:SHP}
Consider a simple example illustrating how SHPs combine discrete and stochastic continuous dynamics to form stochastic hybrid systems.
Consider the following SHP:
\[
\ptest{x^2{+}y^2{\leq}\frac{1}{3}};~\pupdate{\pumod{x}{\frac{x}{2}}};~ \hevolvein{dx=\frac{-x}{2} dt - y dW, dy=\frac{-y}{2} dt + x dW}{\ivr}
\]
It starts with a test \m{\ptest{x^2{+}y^2{\leq}\frac{1}{3}}} checking whether the state variables $x,y$ are in a ball of radius $\sqrt{\frac{1}{3}}$ around 0. Only if it succeeds can the SHP continue.
Then, the SHP performs an instantaneous discrete jump reducing the value of $x$ to \m{\pupdate{\pumod{x}{\frac{x}{2}}}}.
Finally, the SHP follows a stochastic differential equation for rotational Brownian motion restricted to the evolution domain region $\ivr$ that we define as \m{\ivr\mequiv x^2+y^2<9}.
This gives rotational dynamics from the diffusion terms (compare \rref{ex:rotational}) and an exponentially contracting drift.
\end{example}

\subsection{\SdLbf Formulas} \label{sec:SdL-formula}

Formulas of stochastic differential dynamic logic are built out of \SdL function terms.
\begin{definition}[\SdL term]
\emph{Function terms} of \emph{stochastic differential dynamic logic} \SdL are formed by the grammar ($F$ is a primitive measurable function definable in first-order real arithmetic, e.g., the characteristic function $\charf{S}$ of a set $S$ definable in first-order real arithmetic, $B$ is a boolean combination of such characteristic functions using operators $\land,\lor,\lnot,\limply$ from \rref{fig:SdL-abbreviations}, $\lambda,\nu$ are rational numbers):
\[
f,g ~\bebecomes~ 
F
\alternative \lambda f + \nu g
\alternative Bf
\alternative \ddiamond{\alpha}{f}
\]
\end{definition}

\begin{wrapfigure}{r}{4cm}
\vspace{-1.5\baselineskip}
\begin{align*}
0 &\mequiv \charf{\emptyset}\\
1 &\mequiv \charf{\reals^d}\\
\lnot f &\mequiv 1-f\\
A\land B &\mequiv AB\\
A\lor B &\mequiv A+B-AB\\
A\limply B &\mequiv 1-A+AB\\
\dbox{\alpha}{f} &\mequiv \lnot\ddiamond{\alpha}{\lnot f}
\end{align*}
  \caption{Common \SdL and SHP abbreviations}
  \label{fig:SdL-abbreviations}
\end{wrapfigure}
One typical choice for a primitive measurable function $F$ is the characteristic function $\charf{S}$ of a set $S$ definable in first-order real arithmetic, which is then measurable \cite{tarski_decisionalgebra51}.
The \SdL term \m{\lambda f + \nu g} is a \emph{linear combination} of \SdL terms.
The \SdL term \m{Bf} is a boolean product with a boolean combination $B$ of characteristic functions of measurable sets.
\SdL term \m{\ddiamond{\alpha}{f}} represents the supremal value of $f$ along the process belonging to $\alpha$.
The syntactic abbreviations in \rref{fig:SdL-abbreviations} can be useful, especially for convenient operators on boolean combinations of characteristic functions.
Formulas of \SdL are simple, because \SdL function terms are powerful.
\emph{\SdL formulas} express equational and inequality relations between \SdL function terms~$f,g$.
\begin{definition}[\SdL formula]
The \emph{formulas of \SdL} are defined by the following grammar (where $f,g$ are \SdL function terms):
\[
\phi ~\bebecomes~ f\leq g \alternative f=g
\]
\end{definition}
The semantics of classical logics maps an interpretation to a truth-value.
This does not work for stochastic logic, because the state evolution of SHPs contained in \SdL formulas is stochastic, not deterministic.
Instead, the semantics of an \SdL function term is a generator for a random variable.

\begin{definition}[\SdL semantics] \label{def:SdL-semantics}
The \emph{semantics} \m{\ivaluation{\Io}{f}} of \SdL function term $f$ is a function \m{\ivaluation{\Io}{f}:(\Omega\to\reals^d)\to(\Omega\to\reals)} that maps any $\reals^d$-valued random variable $\iget[state]{\IZ}$ describing the current state to a random variable \m{\ivaluation{\IZ}{f}}.
It is defined inductively by
\begin{enumerate}
  \item \label{case:primitivef}
    \(\ivaluation{\IZ}{F} = F^\ell(\iget[state]{\IZ})\), i.e.,
    \(\ivaluation{\IZ}{F}(\omega) = F^\ell(\iget[state]{\IZ}(\omega))\)
    where $F^\ell$ is the function denoted by function symbol $F$
  \item \label{case:linear}
    \(\ivaluation{\IZ}{\lambda f + \nu g} = \lambda\ivaluation{\IZ}{f} + \nu\ivaluation{\IZ}{g}\)
  \item \label{case:multiply}
    \(\ivaluation{\IZ}{B f} = \ivaluation{\IZ}{B}\cdot\ivaluation{\IZ}{f}\), i.e., pathwise multiplication
  \(\ivaluation{\IZ}{B f}(\omega) = \ivaluation{\IZ}{B}(\omega)\cdot\ivaluation{\IZ}{f}(\omega)\)
  \item \label{case:diamond}
    \def\Iint{\iaccess[\alpha]{\IZ[t]}}%
    \newcommand{\Ilater}{\iconcat[state=\Iint]{\IZ}}%
    \(\ivaluation{\IZ}{\ddiamond{\alpha}{f}} = \sup \{\ivaluation{\Ilater}{f} \with 0\leq t\leq \istop[\alpha]{\IZ} \}\)
\end{enumerate}
\end{definition}
When $\iget[state]{\IZ}$ is not defined (results from a failed test), then \m{\ivaluation{\IZ}{f}} is not defined.

If $f$ is a characteristic function of a measurable set, then \m{\ivaluation{\IZ}{\ddiamond{\alpha}{f}}} corresponds to a random variable that reflects the supremal $f$ value that $\alpha$ can reach at least once during its evolution until stopping time \m{\istop[\alpha]{\IZ}} when starting in a state corresponding to random variable $\iget[state]{\IZ}$.
Then \m{P(\ivaluation{\IZ}{\ddiamond{\alpha}{f}}=1)} is the probability with which $\alpha$ reaches $f$ at least once.
Expected values of \SdL terms are well-defined, e.g., \m{E(\ivaluation{\IZ}{\ddiamond{\alpha}{(f+g)}})} is an expected value, given $\iget[state]{\IZ}$.
This includes the special case where $\iget[state]{\IZ}$ is a deterministic state \m{\iget[state]{\IZ}(\omega):=\iget[state]{\Ix}} for all \m{\omega\in\Omega}.

We say that \SdL formula $f\leq g$ is \dfn{valid}, written \m{\entails f\leq g}, if for all $\reals^d$-valued random variables $\iget[state]{\IZ}$:
\[
\ivaluation{\IZ}{f}\leq\ivaluation{\IZ}{g}, 
~\text{i.e.},~
(\ivaluation{\IZ}{f})(\omega)\leq(\ivaluation{\IZ}{g})(\omega) ~\text{for all}~\omega\in\Omega
\]
Validity of \SdL formula \m{f=g} is defined accordingly, hence, \m{\entails f=g} iff \m{\entails f\leq g} and \m{\entails g\leq f}.
As consequence relation on \SdL formulas, we use the \emph{(global) consequence} that we define as follows (similarly when some of the formulas are $f_i=g_i$):
\begin{align*}
  &f_1\leq g_1,\dots,f_n\leq g_n \entails[g] f \leq g
  \\\text{iff}~&
  \entails f_1\leq g_1,\dots, \entails f_n\leq g_n ~\text{implies}~ \entails f \leq g
\end{align*}
The (global) consequence \m{f_1\leq g_1,\dots,f_n\leq g_n \entails[g] f \leq g} holds \dfn{pathwise} if it holds for each $\omega\in\Omega$.

\begin{example}[Jumping rotational Brownian motion]  \label{ex:SdL}
Let $\alpha$ denote the SHP from \rref{ex:SHP}.
The \SdL term \m{\ddiamond{\alpha}{x^2{+}y^2}} represents the supremal value of the square, $x^2{+}y^2$, of the Euclidean norm along the process $\alpha$.
The semantics \m{\ivaluation{\IZ}{\ddiamond{\alpha}{x^2{+}y^2}}} of this \SdL term starting from an initial random variable $\iget[state]{\IZ}$ is a random variable.
The \SdL formula \m{\ddiamond{\alpha}{x^2{+}y^2\geq1}} expresses that this supremal value is greater or equal 1.
The following states that this happens (\SdL formula \m{\ddiamond{\alpha}{x^2{+}y^2\geq1}} holds) with probability at most $\frac{1}{3}$:
\begin{equation}
P(\ddiamond{\alpha}{x^2{+}y^2} {\geq} 1) \leq \frac{1}{3}
\label{eq:ex-SdL}
\end{equation}
\end{example}

\subsection{Measurability} \label{sec:SdL-measurable}

The semantics of SHPs and \SdL needs to satisfy measurability properties to make sure that probabilities are well-defined and probabilistic questions become meaningful.

A \dfn[Markov!time]{Markov time} (a.k.a. stopping time) is a non-negative random variable $\tau$ such that \m{\{\tau\leq t\} \in \mathcal{F}_t} for all $t$ (i.e., it does not depend on the future).
For a Markov time $\tau$ and a stochastic process $X_t$, the following process that is stuck after time $\tau$ is called \emph{stopped process} $X^\tau$
\[
X^\tau_t :=
X_{t\imin\tau} =
\begin{cases}
X_t &\text{if}~ t<\tau\\
X_{\tau} &\text{if}~ t\geq\tau
\end{cases}
\quad\text{where}\quad t\imin\tau := \min\{t,\tau\}
\]
A class $\mathcal{C}$ of processes is \emph{stable under stopping} if $X\in\mathcal{C}$ implies $X^\tau\in\mathcal{C}$ for every Markov time $\tau$.
Right continuous adapted processes, and processes satisfying the strong Markov property are stable under stopping \cite[Theorem 10.2]{Dynkin65}.

We have proved that the SHP semantics is well-defined.
This includes that the natural stopping times \m{\istop[\alpha]{\IZ}} are Markov times so that it is meaningful to stop process \m{\iaccess[\alpha]{\IZ}} at \m{\istop[\alpha]{\IZ}} and so that useful properties of \m{\iaccess[\alpha]{\IZ}} inherit to the stopped process 
\def\tmpstoppedt{t\imin\istop[\alpha]{\IZ}}%
\m{\iaccess[\alpha]{\IZ[\tmpstoppedt]}}.
Furthermore, we have shown that the process \m{\iaccess[\alpha]{\IZ}} is adapted (does not look into the future) and is a.s. \cadlag (right-continuous and has left limits), which is important to define a semantics for \SdL formulas.

\begin{theorem}[Adaptive \cadlag process with Markov times \cite{DBLP:conf/cade/Platzer11}] \label{thm:cadlag-adapted-Markov}
  For each SHP $\alpha$ and any $\reals^d$-valued random variable $\iget[state]{\IZ}$,
  \m{\iaccess[\alpha]{\IZ}} is a.s. a \cadlag process and adapted (to the completed filtration \m{(\mathcal{F}_t)_{t\geq0}} generated by $\iget[state]{\IZ}$ and the constituent Brownian motions \m{(B_s)_{s\leq t}} and uniform processes $U$)
  and \m{\istop[\alpha]{\IZ}} is a Markov time (for \m{(\mathcal{F}_t)_{t\geq0}}).
    \def\stoptime{\istop[\alpha]{\IZ}}%
  In particular, the end value \m{\iaccess[\alpha]{\IZ[\stoptime]}} is again \ignore{\m{\mathcal{F}_{\stoptime}}-}measurable.
\end{theorem}
Note in particular, that the event \m{\{\istop[\alpha^n]{\IZ}\geq t\}} is $\mathcal{F}_t$-measurable, thus, by \cite[Prop 1.2.3]{KaratzasShreve}, the event \m{\{\istop[\alpha^n]{\IZ}>t\}} in \rref{case:repeat} of the semantics of SHPs is $\mathcal{F}_t$-measurable.
As a corollary to \rref{thm:cadlag-adapted-Markov}, \m{\iaccess[\alpha]{\IZ}} is progressively measurable \cite[Prop 1.1.13]{KaratzasShreve}, which implies that the stopped processes are measurable.

We have proved that the \SdL semantics is well-defined and that \m{\ivaluation{\IZ}{f}} is, indeed, a random variable, i.e., measurable.
Without this, probabilistic questions about the value of \SdL terms and formulas would not be well-defined, because they are not measurable with respect to probability space \m{(\Omega,\mathcal{F},P)}.
\begin{theorem}[Measurability \cite{DBLP:conf/cade/Platzer11}] \label{thm:measurable-valuation}
  For any $\reals^d$-valued random variable $\iget[state]{\IZ}$, the semantics \m{\ivaluation{\IZ}{f}} of function term $f$ is a random variable (i.e., $\mathcal{F}$-measurable).
\end{theorem}
In particular, for each Borel-measurable set $S$, the probability \m{P(\ivaluation{\IZ}{f} \in S)} is well-defined so that probabilistic questions have a well-defined answer.
Note that well-definedness of case~\ref{case:diamond} of the semantics of \SdL uses \rref{thm:cadlag-adapted-Markov}.

\subsection{Axioms} \label{sec:SdL-calculus}

Stochastic hybrid systems are a very expressive model and can even represent systems with a very complicated behavior. 
Just like the systems they model, however, they still expose their rich semantical complexity.
In order to make this rich semantical complexity and the deep theory behind the stochastic dynamics accessible in a form that is amenable to a computational approach, we have introduced a proof calculus for \SdL \cite{DBLP:conf/cade/Platzer11}.
Just like the proof calculi for \dL and \QdL, the proof rules for \SdL are syntactic so that only the justification of their soundness requires the deep theory of stochastic processes and stochastic hybrid systems, their use does not.
This makes it possible to computerize stochastic calculus in the form of syntactic \SdL proofs.

We show axioms and proof rules that can be used to prove \SdL formulas in \rref{fig:SdL}.
\begin{figure*}[tb]
  \renewcommand*{\irrulename}[1]{\text{#1}}%
  \newdimen\linferenceRulehskipamount%
  \linferenceRulehskipamount=1mm%
  \newdimen\lcalculuscollectionvskipamount%
  \lcalculuscollectionvskipamount=0.1em%
  \begin{calculuscollections}{\columnwidth}
    \begin{calculus}
      \cinferenceRule[Sassignd|$\langle:=\rangle$]{assignment / substitution axiom}
      {\linferenceRule[eq]
        {\mapply[x]{f}{\theta}}
        {\ddiamond{\pupdate{\umod{x}{\theta}}}{\mapply[x]{f}{x}}}
      }{}%
      \cinferenceRule[Stestd|$\langle{?}\rangle$]{test}
      {\linferenceRule[eq]
        {\ivr f}
        {\ddiamond{\ptest{\ivr}}{f}}
      }{}
      \cinferenceRule[Scomposed|$\langle{{;}}\rangle$]{sequential composition} 
      {\linferenceRule[leq]
        {\ddiamond{\alpha}{(f\imax\ddiamond{\beta}{f})}}
        {\ddiamond{\alpha;\beta}{f}}
      }{}
      \cinferenceRule[Sscalarpost|$\langle\rangle\lambda$]{scalar multiplication in post condition}
      {\linferenceRule[eq]
        {\lambda\ddiamond{\alpha}{f}}
        {\ddiamond{\alpha}{(\lambda f)}}
      }{}
      \cinferenceRule[Slinpost|$\langle\rangle+$]{linear in post condition}
      {\linferenceRule[leq]
        {\lambda\ddiamond{\alpha}{f} + \nu\ddiamond{\alpha}{g}}
        {\ddiamond{\alpha}{(\lambda f+\nu g)}}
      }{}
      \cinferenceRule[Scharf|\usebox{\Ival}]{characteristic function}
      {0\leq B =B B \leq 1\hspace{3cm}}
      {\text{$B$ boolean from characteristic functions}}
      \cinferenceRule[SR|R]{regular modal or monotonicity}
      {\linferenceRule[entail]
        {f\leq g}
        {\ddiamond{\alpha}{f}\leq\ddiamond{\alpha}{g}}
      }{}
      \cinferenceRule[Sdiffweak|DW]{stochastic differential weakening}
      {\linferenceRule[entail]
        {\closure{\ivr}\limply f\leq \lambda}
        {\ddiamond{\hevolvein{dx = b dt + \sigma dW}{\ivr}}{f}\leq \lambda}
      }{\m{\lambda\in\rationals}}
      \cinferenceRule[Sdind|ind]{diamond induction}
      {\linferenceRule[entail]
        {\ddiamond{\alpha}{g}\leq g}
        {\ddiamond{\prepeat{\alpha}}{g}\leq g}
      }{}
    \end{calculus}
  \end{calculuscollections}
  \caption{Pathwise proof rules for \SdL}
  \label{fig:SdL}
\end{figure*}
First we present proof rules that are sound pathwise, i.e., satisfy the global consequence relation pathwise for each \m{\omega\in\Omega}.
Axiom \irref{Sassignd} corresponds to Hoare's assignment rule for deterministic discrete assignment dynamics.
Axiom \irref{Stestd} is the test axiom from dynamic logic, just using the product notation of \SdL instead of conjunctions.

Axiom \irref{Scomposed} is the \SdL form of the sequential composition axiom.
By $\imax$ we denote the binary maximum operator.
Operator $\imax$ coincides with $\lor$ for values in \{0,1\}, e.g., for \SdL terms built using operators $\land,\lor,\lnot,\ddiamond{\alpha}{}$ from characteristic functions.
As a supremum, $\ddiamond{\alpha}{B}$ only takes on values \{0,1\} if $B$ does.
Note that the two sides of \irref{Scomposed} are generally not equal, because $\alpha$ has to run to completion before $\beta$ starts in \m{\ddiamond{\alpha;\beta}{f}}, but $\alpha$ can stop early in \m{\ddiamond{\alpha}{(f\imax\ddiamond{\beta}{f})}} and $\beta$ can then start already.

Axiom \irref{Sscalarpost} and \irref{Slinpost} represent scalar multiplication and (sub)linearity for scalars $\lambda,\nu$.
The two sides of axiom \irref{Slinpost} are not equal if the suprema \m{\ddiamond{\alpha}{f}} and \m{\ddiamond{\alpha}{g}} are at different times.
Axiom \irref{Scharf} expresses idempotence and range constraints on boolean combinations (\rref{fig:SdL-abbreviations}) of characteristic functions.

Rule \irref{SR} is the generalization rule of regular modal logic C expressing monotonicity.
It is the counterpart of a corresponding regular generalization rule that is easily derivable from \irref{K} and \irref{G} in \dL.
Rule \irref{Sdiffweak} is the weakening rule for differential equations, yet generalized to stochastic differential equations.
Note that formula \m{\closure{\ivr}\limply f\leq \lambda} in \irref{Sdiffweak} is equivalent to \m{\closure{\ivr}f\leq\closure{\ivr}\lambda} but easier to read.
If $f$ is continuous, rule \irref{Sdiffweak} is sound when replacing the topological closure $\closure{\ivr}$ (which is computable by quantifier elimination) by $\ivr$, because the inequality is a weak inequality.
Rule \irref{Sdind} is an induction rule.
It corresponds to rule \irref{invind} of \dL from p.\,\pageref{ir:invind}.

Other rules are derivable from \rref{fig:SdL}:
\begin{center}
  \begin{calculus}
      \dinferenceRule[Spositive|$pos$]{positivity}
      {\linferenceRule[entail]
        {0\leq f}
        {0\leq\ddiamond{\alpha}{f}}
      }{}
      \dinferenceRule[Scomposeds|$\langle{{;}}\rangle_2$]{strong composition} 
      {\linferenceRule[leq]
        {\ddiamond{\alpha}{\ddiamond{\beta}{f}}\hspace{1cm}}
        {\ddiamond{\alpha;\beta}{f}}
      }{\m{\entails f\leq\ddiamond{\beta}{f}}}
      \dinferenceRule[Scomposedwp|$\langle{;}\rangle_3$]{positive weak composition}
      {\linferenceRule[leq]
        {\ddiamond{\alpha}{(f+\ddiamond{\beta}{f})}\qquad}
        {\ddiamond{\alpha;\beta}{f}}
      }{\m{0\leq f}}
  \end{calculus}
\end{center}
Rule \irref{Spositive} expresses positivity (if $f$ is positive then $\ddiamond{\alpha}{f}$ is) and is derivable\footnote{
  By \irref{SR}, \m{0\leq f \entails \ddiamond{\alpha}{0}\leq\ddiamond{\alpha}{f}}.
  By \irref{Sscalarpost}, \m{\ddiamond{\alpha}{0}=\ddiamond{\alpha}{(0\ast0)}=0\ddiamond{\alpha}{0}=0}.
}
from \irref{Sscalarpost} and \irref{SR}.
The simpler sequential composition principle \irref{Scomposeds} is derivable\footnote{
By \irref{qear}, \m{f\leq\ddiamond{\beta}{f}} entails \m{f\imax\ddiamond{\beta}{f} =\ddiamond{\beta}{f}}, from which \irref{SR} derives
\(\ddiamond{\alpha}{(f\imax\ddiamond{\beta}{f})} \leq \ddiamond{\alpha}{\ddiamond{\beta}{f}}\).
Thus, \(\ddiamond{\alpha;\beta}{f}\leq\ddiamond{\alpha}{\ddiamond{\beta}{f}}\) by \irref{Scomposed}.
}
from \irref{Scomposed} and \irref{SR}.
Its assumption \m{\entails f\leq\ddiamond{\beta}{f}} holds in particular if $\beta$ is continuous at 0 a.s..
A sufficient condition for SHP $\beta$ to be a.s. continuous at 0 is that, on all paths, the first atomic operation that is not a test is a stochastic differential equation, not an assignment or random assignment.
Formula \irref{Scomposedwp} is derivable\footnote{
By \irref{Spositive}, \(0\leq f\) entails \(0\leq\ddiamond{\beta}{f}\).
Thus, \irref{SR} and the semantics of $\imax$, 
yield
\(\ddiamond{\alpha}{(f\imax\ddiamond{\beta}{f})} \leq \ddiamond{\alpha}{(f+\ddiamond{\beta}{f})}\) by \m{0\leq f},
which derives \irref{Scomposedwp} together with
\(\ddiamond{\alpha;\beta}{f} \leq \ddiamond{\alpha}{(f\imax\ddiamond{\beta}{f})}\), which holds by \irref{Scomposed}.
}
from \irref{Scomposed}, \irref{Spositive}, and \irref{SR}.
Consequently the operator $\imax$ can either be added into the language or approximated conservatively by $+$ as in \irref{Scomposedwp}.

Not all reasoning for stochastic hybrid systems is sound pathwise.
We have introduced other \SdL axioms and proof rules that do not hold pathwise, but are still sound in distribution.
Rule \irref{Splusd} relates probabilities of linear combinations (alias probabilistic choices):
\begin{center}
\hspace{-0.3cm}
  \begin{calculus}
      \cinferenceRule[Splusd|$\langle{\oplus}\rangle$]{axiom of probabilistic choice}
      {\linferenceRule[eq]
        {\lambda P(\ddiamond{\alpha}{f}\in S) + \nu P(\ddiamond{\beta}{f}\in S)}
        {\hspace{-0.3cm}P(\ddiamond{\lambda\alpha\oplus\nu\beta}{f}\in S)}
      }{}
  \end{calculus}
\end{center}

How to prove properties about a random assignment \m{\prandom{x_i}} depends on the distribution that we use for the random assignment. For a uniform distribution in [0,1], e.g., we obtain the following proof rule that is sound in distribution:
\begin{center}
  \begin{calculus}
      \cinferenceRule[Srandomd|$\langle*\rangle$]{(uniform) random assignment}
      {\linferenceRule[eq]
        {\int_0^1 \charf{\ddiamond{\pupdate{\pumod{x_i}{r}}}{f}\in S} dr}
        {P(\ddiamond{\prandom{x_i}}{f} \in S)}
      }{}
  \end{calculus}
\end{center}
The integrand is measurable for measurable $S$ by \rref{thm:measurable-valuation}.
The rule is applicable when $f$ has been simplified enough using other proof rules such that the integral can be computed after using \irref{Sassignd} to simplify the integrand.

The \SdL rules and axioms generalize to probabilistic assumptions by the rule of partition, i.e., using
\[
P(C) = P(C|A)P(A) + P(C|\lnot A)P(\lnot A)
\]
to consider the case where assumption $A$ holds separately from the case where $A$ does not hold (giving less information about conclusion $C$, but the probability can by bounded).

\subsection{Soundness} \label{sec:SdL-sound}

The most critical result about the \SdL axioms and proof rules is that they are sound, so that their simple syntactic reasoning always gives correct results about the semantical behavior of stochastic hybrid systems.
This result splits into \SdL axioms and proof rules that are sound pathwise (i.e., the global consequence relation between premises and conclusion holds pathwise) and those that are sound in distribution (i.e., the indicated probability relations hold).
All proof rules that are sound pathwise are sound in distribution but not vice versa.

\begin{theorem}[Pathwise soundness of \SdL \cite{DBLP:conf/cade/Platzer11}] \label{thm:sound-global}
  The \SdL axioms and proof rules in \rref{fig:SdL} are globally sound pathwise.
\end{theorem}

\begin{theorem}[Soundness in distribution of \SdL \cite{DBLP:conf/cade/Platzer11}] \label{thm:sound-distribution}
  The \SdL proof rules \irref{Splusd} and \irref{Srandomd} are sound in distribution.
\end{theorem}

\subsection{Stochastic Differential Invariants} \label{sec:Sdiffind}

Proving properties of differential equations is very challenging.
Differential invariants (\rref{sec:diffind}) are the primary proof technique for more complicated differential equations that have no computable solutions.
Differential invariants are the ``only'' choice for differential equations with disturbances \cite{DBLP:journals/logcom/Platzer10,Platzer10}, because those do not have a single unique solution but depend on input functions.
In that respect, stochastic differential equations are like differential equations with disturbances, because both have noise terms in the right hand sides, which make solutions non-unique, depending on input functions and noise.
The difference is that stochastic differential equations have a stochastic model of the noise, whereas classical differential equations with disturbances have a nondeterministic model.
Solving stochastic differential equations is, for the most part, limited to sampling, and even that is difficult \cite{Oksendal07,KloedenPlaten2010}.

For \SdL, we lift the ideas behind differential invariants for differential equations to \emph{stochastic differential invariants} for stochastic differential equations \cite{DBLP:conf/cade/Platzer11}.
One critical element behind differential invariants is \rref{lem:derivationLemma}, which relates (computable) syntactic derivatives to (semantic) analytic differentiation.
The stochastic analogue of analytic differentiation are infinitesimal generators, the analogue of syntactic derivatives are differential generators.

\begin{definition}[Infinitesimal generator]
  The \dfn[generator!infinitesimal]{(infinitesimal) generator} of an a.s. right continuous strong Markov process (e.g., a solution from \rref{def:SDE}) is the operator $A$ that maps a function \m{f:\reals^d\to\reals} to the function \m{Af:\reals^d\to\reals} defined as
  \[
  Af(x) := \lim_{t\searrow0} \frac{E^x f(X_t) - f(x)}{t}
  \]
\end{definition}
We say that $Af$ is defined if this limit exists for all $x\in\reals^d$.
By Dynkin's formula \cite[p. 133]{Dynkin65}, infinitesimal generators can be used to determine, without solving the stochastic differential equation, the expected value of a function when following the process until a Markov time.
\begin{theorem}[{Dynkin's formula \cite[Theorem 7.4.1]{Oksendal07}}] \label{thm:Dynkin} %
  Let $X_t$ an a.s. right continuous strong Markov process (e.g., a solution from \rref{def:SDE}).
  If \m{f\in\continuouss[2]{\reals^d}{\reals}} has compact support and $\tau$ is a Markov time with $E^x \tau <\infty$, then
  \[
  E^x f(X_\tau) = f(x) + E^x \int_0^\tau Af(X_s)ds
  \]
\end{theorem}

Dynkin's formula is very useful, but only if we can compute the infinitesimal generator and its integral.
The generator $A$ is a stochastic expression in the sense that $A$ is defined in terms of a limit of an expectation of a stochastic process.
It has been shown, however, that, under fairly mild assumptions, it is equal to a deterministic expression of $x$ called the differential generator.
\begin{theorem}[{Differential generator \cite[Theorem 7.3.3]{Oksendal07}}] \label{thm:generatorid}
  Consider a solution $X_t$ of a stochastic differential equation from \rref{def:SDE}.
  If \m{\inv\in\continuouss[2]{\reals^d}{\reals}} is compactly supported, then $A\inv$ is defined and equal to the \emph{differential generator} $L\inv$ of $\inv$:
  \[
  L\inv(x) := \sum_i b_i(x)\Dp[x_i]{\inv}(x) + \frac{1}{2}\sum_{i,j}(\sigma(x)\transpose{\sigma(x)})_{i,j} \frac{\partial^2 \inv}{\partial x_i \partial x_j}(x)
  \]
\end{theorem}
Observe that this deterministic differential generator can be computed syntactically, yet is still used for a stochastic process in \rref{thm:Dynkin}.
We turn this principle into a fully syntactic proof rule for stochastic differential invariants.
\begin{theorem}[Soundness of stochastic differential invariants] \label{thm:sdi}
  If function \m{\inv\in\continuouss[2]{\reals^d}{\reals}} has compact support on $\ivr$ (which holds for all \m{\inv\in\continuouss[2]{\reals^d}{\reals}} if $\ivr$ represents a bounded set), then the proof rule \irref{sded} is sound for \m{\lambda>0,p\geq0}
  \begin{center}
  \begin{calculus}
    \cinferenceRule[sded|$\langle'\rangle$]{}{
      \linferenceRule
        {\lsequent{}{\ddiamond{\alpha}{(\ivr\limply\inv)} \leq \lambda p} & \lsequent{\ivr}{\inv\geq0} & \lsequent{\ivr}{L\inv\leq0}}
        {P(\ddiamond{\alpha}{\ddiamond{\hevolvein{dx = b dt + \sigma dW}{\ivr}}{\inv}}\geq \lambda)\leq p}
    }{}
  \end{calculus}
  \end{center}
\end{theorem}
See \cite{DBLP:conf/cade/Platzer11} for a proof of \rref{thm:sdi}.
The implications in the premises of \irref{sded} can be understood like that in \irref{Sdiffweak}.
Let $\ivr$ be given by first-order real arithmetic formulas.
If $\inv$ is polynomial and, thus, \m{\inv\in\continuouss[2]{\reals^d}{\reals}},
then the second and third premise of \irref{sded} are in first-order real arithmetic, hence decidable.

\begin{example}[Jumping rotational Brownian motion] 

With \m{\ivr\mequiv x^2+y^2<9}, the probabilistic statement about the \SdL formula \rref{eq:ex-SdL} from \rref{ex:SdL} can be proved easily in the \SdL calculus.
The first step is to use \irref{Scomposeds} as follows:
\begin{align*}
&P\Big(\Big\langle\ptest{x^2{+}y^2{\leq}\frac{1}{3}};~\pupdate{\pumod{x}{\frac{x}{2}}};~ \hevolvein{dx=\frac{-x}{2} dt - y dW\syssep dy=\frac{-y}{2} dt + x dW}{\ivr}\Big\rangle{x^2{+}y^2} {\geq} 1\Big)
\\\stackrel{\irref{Scomposeds}}{\leq}&
P\Big(\Big\langle\ptest{x^2{+}y^2{\leq}\frac{1}{3}};~\pupdate{\pumod{x}{\frac{x}{2}}}\Big\rangle\Big\langle \hevolvein{dx=\frac{-x}{2} dt - y dW\syssep dy=\frac{-y}{2} dt + x dW}{\ivr}\Big\rangle{x^2{+}y^2} {\geq} 1\Big) \leq \frac{1}{3}
\end{align*}
The last inequality (${\leq}\frac{1}{3}$) is by rule \irref{sded}.
The second premise of \irref{sded} is proved by \m{\inv\mequiv x^2+y^2\geq0}.
The third premise is proved as follows.
The stochastic differential equation has the form
\[
d\begin{pmatrix}x\\y\end{pmatrix} = b dt + \sigma dW
= \begin{pmatrix}-\frac{x}{2}\\-\frac{y}{2}\end{pmatrix}dt
+ \begin{pmatrix}-y\\x\end{pmatrix}dW
\]
Thus, in order to determine $L\inv$, we compute
\[
\sigma\transpose{\sigma} =
\begin{pmatrix}-y\\x\end{pmatrix} \begin{pmatrix}-y&x\end{pmatrix}
=
\begin{pmatrix}
y^2 & -xy\\
-xy & x^2
\end{pmatrix}
\]
and obtain by \rref{thm:generatorid}:
\begin{align*}
L\inv =& -\frac{x}{2}\Dp[x]{\inv}-\frac{y}{2}\Dp[y]{\inv} + \frac{1}{2}\left(y^2\Dp[x^2]{^2\inv}-2xy\frac{\partial^2\inv}{\partial x\partial y} +x^2\Dp[y^2]{^2\inv}\right)
\\
= &
-\frac{x}{2}2x-\frac{y}{2}2y+\frac{1}{2}\left(2y^2-0xy+2x^2\right)
\leq0
\end{align*}
The first premise of \irref{sded} can be proved in the \SdL calculus by the following chain of inequalities, which are justified by \SdL axioms and arithmetic as indicated:
\begin{align*}
  & \ddiamond{\ptest{x^2+y^2{\leq}\frac{1}{3}};~\pupdate{\pumod{x}{\frac{x}{2}}}}{(\ivr\limply\inv)}\\
  &\stackrel{\irref{Scomposed}}{\leq} \ddiamond{\ptest{x^2+y^2{\leq}\frac{1}{3}}}{\Big((\ivr\limply\inv)\imax\ddiamond{\pupdate{\pumod{x}{\frac{x}{2}}}}{(\ivr\limply\inv)}\Big)}\\
  &\stackrel{\irref{Sassignd}}{=} \ddiamond{\ptest{x^2+y^2{\leq}\frac{1}{3}}}{\Big((\ivr\limply\inv)\imax(
  \ddiamond{\pupdate{\pumod{x}{\frac{x}{2}}}}{\ivr}%
  \limply\Big(\frac{x}{2}\Big)^2+y^2)\Big)}\\
  &\stackrel{\irref{qear}}{=} \ddiamond{\ptest{x^2+y^2{\leq}\frac{1}{3}}}{(\ivr\limply\inv)}\\
  &\stackrel{\irref{Stestd}}{=} \Big(x^2+y^2{\leq}\frac{1}{3}\Big)(\ivr\limply\inv)\\
  &\stackrel{\irref{qear}}{=} \Big(x^2+y^2{\leq}\frac{1}{3}\Big)(x^2+y^2{<}9 \limply x^2+y^2) \leq \frac{1}{3}
\end{align*}
Together with \m{\ivr\limply\inv\geq0} and \m{\ivr\limply L\inv\leq0}, this inequality entails \SdL formula \rref{eq:ex-SdL} by \SdL proof rule \irref{sded} and \irref{Scomposeds}.
\end{example}

Observe that the \dL and \QdL calculus use equivalences and implications of \dL and \QdL axioms together with proof rules to prove \dL and \QdL formulas.
The \SdL calculus uses equations and inequalities of \SdL axioms together with \SdL proof rules to prove \SdL formulas.
Some of this reasoning is in the scope of a probability operator (e.g., when using axiom \irref{Splusd}), other inequalities and equations can be determined without a probability operator.

\section{Related Work} \label{ch:RelatedWork}

Since dynamical systems, hybrid systems, and their extensions are very active areas of research, a comprehensive overview of all results is impossible.
In this survey, we focus on logic and proofs for dynamical systems.
For more details on hybrid systems, we refer to the literature \cite{DBLP:conf/lics/Henzinger96,DBLP:conf/emsoft/Alur11}.
For background on classical logic and proving, we refer to the literature \cite{HughesCresswell96,Fitting96a,Harel_et_al_2000}.

Model checking and reachability analysis have been used successfully for hybrid systems. They work by state space exploration and use various abstractions or approximations \cite{DBLP:conf/lics/Henzinger96,DBLP:conf/lics/HenzingerNSY92,DBLP:journals/tcs/AlurCHHHNOSY95,DBLP:conf/lics/Henzinger96,DBLP:conf/csl/Franzle99,DBLP:conf/hybrid/AnaiW01,DBLP:journals/ijfcs/ClarkeFHKOST03,DBLP:conf/hybrid/Tiwari03,DBLP:conf/atva/MysorePM05,DBLP:journals/tecs/AlurDI06,DBLP:journals/tcs/AlurDI06}, including numerical approximations \cite{DBLP:journals/tac/ChutinanK03,DBLP:conf/hybrid/AsarinDG03}.
Lafferriere et al.~\cite{DBLP:conf/hybrid/LafferrierePY99,DBLP:journals/mcss/LafferrierePS2000,DBLP:journals/jsc/LafferrierePY01} have shown that finite-state bisimulations, which generally do not exist for hybrid systems \cite{DBLP:conf/lics/Henzinger96}, still work for o-minimal hybrid automata and classes of linear dynamics with a homogeneous eigenstructure, provided the discrete and continuous dynamics are completely decoupled.
Surveys of model checking techniques for hybrid systems appeared elsewhere \cite{DBLP:conf/lics/Henzinger96,DoyenPP12}.

Discretizations have been used for linear systems \cite{DBLP:conf/cav/GuernicG09}, to obtain abstractions of fragments of hybrid systems \cite{DBLP:journals/ieee/AlurHLP00,DBLP:journals/tecs/AlurDI06,DBLP:journals/fmsd/Tiwari08}, and to approximate nonlinear systems by hybrid systems \cite{DBLP:journals/tac/HenzingerHWT98} or by piecewise linear dynamics \cite{DBLP:conf/hybrid/AsarinDG03}.
Constraint-based verification approaches \cite{DBLP:journals/tac/PrajnaJP07,DBLP:journals/fmsd/SankaranarayananSM08,DBLP:conf/cav/GulwaniT08} have been considered, which are related to differential invariants.
Verification tools are based on logic and proofs \cite{DBLP:journals/jar/Platzer08,DBLP:conf/cade/PlatzerQ08,DBLP:conf/lics/Platzer12b}, polyhedral reachability analysis \cite{DBLP:journals/sttt/HenzingerHW97,DBLP:journals/sttt/Frehse08}, reachability analysis with support functions \cite{DBLP:conf/cav/GuernicG09,DBLP:conf/cav/FrehseGDCRLRGDM11}, interval-constraint propagation \cite{RatschanS07}, and numerical PDE solving \cite{DBLP:conf/hybrid/MitchellT05}.

Many languages have been proposed for modeling hybrid systems, including extended duration calculus \cite{DBLP:conf/hybrid/ChaochenRH92}, hybrid automata \cite{DBLP:conf/lics/Henzinger96}, hybrid programs \cite{DBLP:conf/tableaux/Platzer07,DBLP:journals/jar/Platzer08}, guarded commands \cite{DBLP:journals/tcs/RonkkoRS03}, hybrid $\pi$-calculus \cite{DBLP:conf/hybrid/RoundsS03}, and process algebra $\chi$ \cite{DBLP:journals/jlp/BeekMRRS06}.

Logic has been used successfully for real-time systems~\cite{DBLP:conf/lics/HenzingerNSY92,DBLP:conf/lics/Dutertre95,DBLP:journals/tcs/SchobbensRH02,ZhouH04,OlderogD08} and for timed-automata based model checking \cite{DBLP:journals/tcs/AlurD94,DBLP:conf/cav/Alur99,DBLP:conf/concur/ComonJ99,DBLP:journals/sttt/LarsenPY97,BaierKL08}.
More details about real-time systems can be found in books \cite{OlderogD08,BaierKL08}.

The importance of understanding dynamic / reconfigurable distributed hybrid systems was recognized in modeling languages SHIFT \cite{DBLP:conf/hybrid/DeshpandeGV96} and R-Charon \cite{DBLP:conf/hybrid/KratzSPL06} for simulation and compilation \cite{DBLP:conf/hybrid/DeshpandeGV96} or semantical considerations \cite{DBLP:conf/hybrid/KratzSPL06}.
For distributed hybrid systems, even giving a formal semantics is very challenging \cite{DBLP:conf/hybrid/ChaochenJR95,DBLP:conf/hybrid/Rounds04,DBLP:conf/hybrid/KratzSPL06,DBLP:journals/jlp/BeekMRRS06}.
Random simulation has been proposed for general dynamical systems \cite{DBLP:conf/hybrid/MeseguerS06}.

Discrete programs with random number generators have been studied in the literature, including \cite{DBLP:journals/jcss/Kozen81,DBLP:journals/jcss/FeldmanH84,DBLP:journals/jcss/Kozen85,McIverMorgan04}.
Probabilities and logic have been considered in many contexts, e.g., \cite{DBLP:journals/ml/RichardsonD06}.
Discrete probabilistic systems and finite Markov chains, have been studied using probabilistic model checking \cite{DBLP:conf/icalp/BaierCHKR97,DBLP:journals/tse/BaierHHK03,DBLP:journals/jacm/CourcoubetisY95} and statistical model checking \cite{DBLP:conf/cav/SenVA05,DBLP:journals/sttt/YounesKNP06,DBLP:conf/cmsb/JhaCLLPZ09}.
Extensions to probabilistic timed automata \cite{DBLP:journals/iandc/KwiatkowskaNSW07} and probabilistic hybrid automata \cite{DBLP:conf/hybrid/ZulianiPC10,DBLP:conf/cav/ZhangSRHH10} have been considered too.

For an overview of model checking techniques for various classes of stochastic hybrid systems, we refer to a survey \cite{Cassandras2006}.
Most verification techniques for stochastic hybrid systems use discretizations, approximations, or assume discrete time and bounded horizon \cite{DBLP:journals/tsmc/KoutsoukosR08,Cassandras2006,DBLP:journals/automatica/AbatePLS08,DBLP:conf/hybrid/HuLS00}.
Barrier certificates have been extended to the stochastic case \cite{DBLP:journals/tac/PrajnaJP07}.

The use of logic has been proposed for hybrid systems, e.g., in a propositional modal $\mu$-calculus \cite{DavorenNerode_2000} or in early work based on phase transition systems \cite{DBLP:conf/rex/MalerMP91}.
See \cite{DavorenNerode_2000} for an excellent overview.
We consider the first-order case, i.e., how to model and prove systems with concrete differential equations like \m{\D{x}=v\syssep\D{v}=a} and concrete control decisions like \m{\pupdate{\pumod{a}{-b}}}, instead of abstract propositional actions $A,B,C$ of unknown effects that propositional modal $\mu$-calculi consider \cite{DBLP:conf/focs/Pratt81,DavorenNerode_2000}.
The use of theorem provers has been suggested in hybrid systems, including STeP~\cite{DBLP:conf/hybrid/MannaS98,DBLP:journals/acta/KestenMP00} and PVS~\cite{DBLP:conf/iceccs/Abraham-MummSH01}.
Their working principles are different from what we show here.
They separate hybridness from the logic and proof by compiling a given global system invariant for a hybrid automaton into a single  verification condition expressing that the invariant is preserved under all transitions of the hybrid automaton.

In our approach, we, instead, take logic and hybridness at face value by developing and studying logics for hybrid systems, which directly integrate the logic and the hybrid dynamics (or extensions) within a single language.
That makes it easier to identify the core logical reasoning principles and transform formulas soundly in an entirely local way even for more general properties than invariance checking.
This view enables the study of logically more foundational questions, including completeness, deductive power, and relationships of differential invariants and differential cuts.
Benefits for automation of proofs and for computing invariants and differential invariants have been discussed elsewhere \cite{DBLP:journals/fmsd/PlatzerC09,Platzer10}.

\section{Summary and Outlook} \label{ch:Conclusions}

We have surveyed logics of dynamical systems, including hybrid systems, distributed hybrid systems, and stochastic hybrid systems.
The logic of discrete dynamical systems and the logic of continuous dynamical systems are fragments of the logic of hybrid systems.
We have surveyed differential dynamic logic (\dL) for hybrid systems,  quantified differential dynamic logic (\QdL) for distributed hybrid systems, and stochastic differential dynamic logic (\SdL) for stochastic hybrid systems.
We have recalled dynamical system models, dynamic logics, their semantics, their axiomatizations, and proof calculi for each of those dynamical systems. We have surveyed important theoretical results, including soundness and completeness, and results about the relative deductive power of differential cuts and of differential auxiliaries.
The differential dynamic logics and their induction techniques for differential equations, which are captured in various forms of differential invariants and differential variants, have been instrumental in proving properties for more advanced dynamical systems.
While the use of the theorem provers implementing differential dynamic logics is beyond the scope of this article, we have given references to more information, including a number of applications and case studies.

Not all important results about logics of dynamical systems and about differential dynamic logics are included in this survey. We still hope to have given the reader a good overview of logics for dynamical systems, and point out relationships and similarities among the techniques.
The reader should note that, for space reasons, not all important members of the family of differential dynamic logics have been presented in this article.
Prominent cases that are missing from this survey include differential-algebraic dynamic logic (\DAL) for hybrid systems with differential-algebraic constraints modeled in differential-algebraic program and  the temporal extension of differential temporal dynamic logic (\dTL).

The results summarized in this article demonstrate that logic is a powerful tool, not just for studying discrete phenomena, but also continuous phenomena, infinite-dimensional phenomena, and stochastic phenomena.
These logics set a strong logical foundation for dynamical systems, including logical foundations for cyber-physical systems.
Such stable foundations for the relatively young area of logic of dynamical systems make it a very promising direction for future research, including theoretical, practical, and applied research.
Given the tremendous progress that logic for programs has made since its conception, we expect to see no less from the area of logics for dynamical systems.

\begin{acks}
I am grateful to Jan-David Quesel for indispensable help with implementing \KeYmaera and to David Renshaw for implementing \KeYmaeraD.
My thanks go to Sicun Gao, David Harel, David Henriques, Oded Maler, Jo\~ao Martins, Sayan Mitra, Vaughan Pratt, and Sriram Sankaranarayanan for feedback on this article.
\end{acks}

\bibliographystyle{acmsmall}
\bibliography{platzer,bibliography}

\received{May 2012}{Month YYYY}{Month YYYY}
\end{document}